\documentclass[twoside]{article}
\usepackage{graphics}
\usepackage{geometry}
\usepackage{amsmath}
\usepackage{mathbbol}
\usepackage{mathrsfs}
\usepackage{amssymb}
\usepackage{multirow}
\usepackage{booktabs}
\usepackage{listings}
\usepackage{titlesec}
\usepackage{listings}
\usepackage{comment}
\usepackage{lmodern}
\usepackage{anyfontsize}
\usepackage[utf8]{inputenc}

\newcommand{\poly}{\mathrm{poly}}

\newcommand{\opt}{\ensuremath{\mathrm{Opt}}}

\usepackage{amsthm}
\newtheorem{theorem}{Theorem}[section]
\newtheorem{corollary}{Corollary}[theorem]
\newtheorem{lemma}[theorem]{Lemma}
\newtheorem{proposition}[theorem]{Proposition}
\newtheorem{definition}{Definition}
\newtheorem{observation}{Observation}
\newtheoremstyle{iremark}
  {\topsep}   
  {\topsep}   
  {\upshape}  
  {0pt}       
  {\itshape}  
  {.}         
  {5pt plus 1pt minus 1pt} 
  {\thmname{#1}\thmnumber{ \itshape#2}\thmnote{ (#3)}} 

\theoremstyle{iremark}

\newcommand{\disj}{\ensuremath{\mathsf{Disjointness}}}

\allowdisplaybreaks

\usepackage{xcolor}
\definecolor{myblue}{RGB}{20, 86, 128}
\definecolor{mygreen}{rgb}{0.0, 0.2, 0.13}
\definecolor{myred}{rgb}{0.8, 0.31, 0.36}
\definecolor{debianred}{rgb}{0.84, 0.04, 0.33}
\definecolor{deepskyblue}{rgb}{0.0, 0.75, 1.0}
\usepackage{hyperref}
\hypersetup{
  citecolor  = blue,
  linkcolor = debianred,
  colorlinks = true
}

\title{Optimal Multi-Dimensional Mechanisms are not Locally-Implementable}
\author{S. Matthew Weinberg\footnote{Computer Science, Princeton University. \texttt{smweinberg@princeton.edu}. Supported by NSF CCF-1955205.} \textcircled{r} \and Zixin Zhou\footnote{Computer Science, Stanford University. \texttt{jackzhou@stanford.edu}. Supported by a Stanford SoE Fellowship.} \textcircled{r}}

\usepackage{graphicx}
\usepackage{geometry}

\begin{document}

\maketitle
\begin{abstract}

We introduce \emph{locality}: a new property of \emph{multi-bidder auctions} that formally separates the simplicity of optimal single-dimensional multi-bidder auctions from the complexity of optimal multi-dimensional multi-bidder auctions. Specifically, consider the revenue-optimal, Bayesian Incentive Compatible auction for buyers with valuations drawn from $\vec{D}:=\times_i D_i$, where each distribution has support-size $n$. This auction takes as input a valuation profile $\vec{v}$ and produces as output an allocation of the items and prices to charge, $\opt_{\vec{D}}(\vec{v})$. When each $D_i$ is single-dimensional, this mapping is \emph{locally-implementable}: defining each input $v_i$ requires $\Theta(\log n)$ bits, and $\opt_{\vec{D}}(\vec{v})$ \emph{can be fully determined using just $\Theta(\log n)$ bits from each $D_i$}. This follows immediately from Myerson's virtual value theory~\cite{Myerson81}.

Our main result establishes that optimal multi-dimensional mechanisms are \emph{not locally-implementable}: in order to determine the output $\opt_{\vec{D}}(\vec{v})$ on one particular input $\vec{v}$, \emph{one still needs to know (essentially) the entire distribution $\vec{D}$}. Formally, $\Omega(n)$ bits from each $D_i$ is necessary: (essentially) enough to fully describe $D_i$, and exponentially more than the $\Theta(\log n)$ needed to define the input $v_i$. We show that this phenomenon already occurs with just two bidders, even when one bidder is single-dimensional, and even when the other bidder is barely multi-dimensional. More specifically, the multi-dimensional bidder is ``inter-dimensional'' from the FedEx setting with just two days~\cite{FiatGKK16}.

Our techniques are fairly robust: we additionally establish that optimal mechanisms for single-dimensional buyers with budget constraints are not locally-implementable. This again occurs even with just two bidders, even when one has no budget constraint, and even when the other's budget is public.
\end{abstract}
\addtocounter{page}{-1}
\newpage
\section{Introduction}\label{sec:intro}

Consider the problem of selling multiple items to multiple bidders, where each bidder's valuation function (for the items) is drawn independently from a distribution known to the seller. The seller desires a truthful auction (formally, Bayesian Incentive Compatible. See Section~\ref{sec:prelim}) maximizing her expected revenue.

	Since Myerson's seminal work, it is well-established that revenue-optimal \emph{single-dimensional} auctions are exceptionally simple, and satisfy many desirable properties. For example, revenue-optimal \emph{single-bidder} single-dimensional auctions offer the bidder a take-it-or-leave-it price (the bidder can pay the price and get the item, or not pay and get nothing). Optimal single-bidder single-dimensional auctions are therefore deterministic, computable in poly-time, revenue-monotone,\footnote{Specifically, if $D$ stochastically dominates $D'$, then the optimal revenue for $D$ exceeds that of $D'$.} and have menu-complexity one.\footnote{The menu-complexity of a single-bidder auction is the number of distinct non-trivial allocations they might receive.} In contrast, optimal \emph{single-bidder} multi-dimensional mechanisms (where, e.g., two distinct items are for sale) require randomization~\cite{Thanassoulis04, Pavlov11}, are computationally hard to find~\cite{DaskalakisDT14,ChenDPSY14,ChenDOPSY15}, are non-monotone~\cite{HartR15, RubinsteinW18},\footnote{Specifically, there exist distributions over additive valuations for two items $D$ and $D'$ which can be coupled so that $v \sim D$ values all sets of items more than $v' \sim D'$, yet the optimal revenue for $D'$ is infinite and the optimal revenue for $D$ is $1$!} and have unbounded menu complexity~\cite{BriestCKW15, HartN19, ManelliV07, DaskalakisDT17, BabaioffGN17}. This vast (and still growing) line of works clearly establishes that optimal \emph{single-bidder} multi-dimensional mechanisms are extremely complex when compared to their single-dimensional counterparts.

In the \emph{multi-bidder} setting, however, the story is less written. Of course, optimal multi-bidder multi-dimensional mechanisms inherit all the complexities of optimal single-bidder multi-dimensional mechanisms. Still, it remains largely unknown to what extent \emph{all} the complexities of optimal multi-bidder multi-dimensional auctions already manifest in the single-bidder setting. Indeed, in some special cases where the optimal single-bidder multi-dimensional auction is tractable, the optimal multi-bidder multi-dimensional auction is tractable as well~\cite{CaiDW12a}. Some further multi-dimensional special cases even admit formal multi-to-single-bidder reductions~\cite{AlaeiFHHM12, AlaeiFHH13, Alaei14}. In this direction, our work identifies a novel complexity of optimal multi-bidder multi-dimensional mechanisms \emph{driven by the multi-bidder aspect}. For example, our quantitative measure identifies complexity in broad classes of two-bidder multi-dimensional settings, even though the single-bidder problem for every instance in these classes is quite simple.\\


\noindent\textbf{Locally-Implementable Mechanisms.} Consider the following thought experiment: you run $k$-bidder auctions, and your Bayesian prior is that Bidder $i$'s valuation function is drawn from $D_i$. When each $D_i$ remains permanently fixed, it makes sense to hard-code the revenue-optimal auction for $\times_i D_i$, and plug in each new valuation profile $\vec{v}$ as input. But you are continuously gathering data on bidders' values to refine your beliefs (in fact, every additional auction executed itself refines your beliefs for future auctions). So while just a single $\vec{v}$ is given as input, the problem you aim to solve is parameterized by the prior $\vec{D}$.

\begin{definition}[Implementing a Revenue Optimal Auction] Given as input $k$ valuation functions $v_1,\ldots, v_k$, and parameterized by $k$ distributions $D_1,\ldots, D_k$, determine $\opt_{\vec{D}}(\vec{v})$: an allocation of items and payments charged on valuation profile $(v_1,\ldots, v_k)$ that is consistent with some revenue-optimal mechanism for $\times_i D_i$.
\end{definition}

At the heart of our paper is the following (for now, informally-posed) question: \emph{How much do you really need to know about each $D_i$ in order to compute $\opt_{\vec{D}}(\vec{v})$ for just one particular $\vec{v}$?}\\

When each $D_i$ is single-dimensional, not much is needed, and this follows immediately from Myerson's theory of (ironed) virtual values~\cite{Myerson81}. Indeed, if each $D_i$ is supported on $n$ valuations, and each valuation has an integer value between $0$ and $\poly(n)$ for each outcome, and the probability of each valuation is an integer multiple of $1/\poly(n)$, then $O(\log n)$ bits from each $D_i$ suffice to compute $\opt_{\vec{D}}(\vec{v})$. Inspired by the concept of locally-decodable codes (see survey~\cite{Yekhanin12}), we term this property \emph{locally-implementable}: to compute $\opt_{\vec{D}}(\vec{v})$, barely more bits are needed from each $D_i$ than the bits needed to state $v_i$ itself.\footnote{To draw the (very high-level) conceptual connection to locally-decodable codes, think of the $n$-bit codeword $C$ as parameterizing the decoding algorithm, which is given an index $i$ as input (and the desired output is $m_i$, the $i^{th}$ bit of the original message). Locality refers to the fact that $m_i$ can be determined by querying just $o(n)$ bits of $C$. Similarly, locality in our context refers to the fact that $\opt_{\vec{D}}(\vec{v})$ can be determined using just $O(\log n)$ bits from each $D_i$.}

To quickly see this (see Appendix~\ref{app:single} for a more detailed sketch), recall that Myerson's seminal work defines a (ironed) virtual valuation function $\bar{\varphi}^{D_i}_i(\cdot)$ (which depends only on $D_i$ and not $D_{-i}$) such that the revenue-optimal auction gives the item to Bidder $i^*:=\arg\max_{i} \{\bar{\varphi}^{D_i}_i(v_i)\}$ and charges them price $(\bar{\varphi}^{D_{i^*}}_{i^*})^{-1}(\max\{0,\max_{i \neq i^*}\{\bar{\varphi}^{D_i}_i(v_i)\}\})$ (if $\bar{\varphi}_{i^*}^{D_{i^*}}(v_{i^*}) \geq 0$, otherwise no one wins the item). In particular, \emph{the winner can be determined just by knowing $\bar{\varphi}_i^{D_i}(v_i)$ for all $i$}, and the payment charged can be further determined \emph{with one additional query to $\bar{\varphi}_{i^*}^{D_{i^*}}(\cdot)$}. Moreover, each $\bar{\varphi}^{D_i}_i(v)$ is the ratio of two integers of size at most $\poly(n)$ (subject to the conditions at the start of this paragraph), so \emph{just $O(\log n)$ bits from each $D_i$ suffice to compute $\opt_{\vec{D}}(\vec{v})$, while $\Omega(n)$ bits are necessary to fully specify each $D_i$.}

Our main result shows that optimal multi-dimensional mechanisms are \emph{not locally-implementable}, and in fact are as far from locally-implementable as possible: $\Omega(n)$ bits from each $D_i$ are needed to determine $\opt_{\vec{D}}(\vec{v})$ --- nearly as many bits as needed to fully specify $D_i$, and exponentially more than the $\Theta(\log n)$ bits needed to specify $v_i$. We further show that this already holds in essentially the simplest multi-dimensional setting: there are just two bidders, one of whom is single-dimensional, and another who is ``inter-dimensional'' according to the FedEx problem~\cite{FiatGKK16}. Specifically, there are two options for the item (call them one-day and two-day shipping). The single-dimensional bidder always has value $0$ for two-day shipping. The multi-dimensional bidder either has the same value for both options, or has value $0$ for two-day shipping. Such bidders are a (very) special case of unit-demand bidders.\footnote{A valuation is unit-demand if its valuation for a set of items $S$ is $v(S):=\max_{i \in S} v(\{i\})$.} They are also a special case of buyers with a private budget constraint~\cite{DevanurW17}, and single-minded buyers~\cite{DevanurGSSW20}. Formally, we use the lens of communication complexity for the following problem to establish our main result.

\begin{definition} Select-Outcome Problem is a communication problem between Alice and Bob. Alice is given as input $D_1$, and a valuation $v_1$ in its support. Bob is given $D_2$, and a valuation $v_2$ in its support. When the input is of size $n$, each distribution has support-size at most $n$, all valuations in the support have integer values $\leq n^3$ for all outcomes, and all probabilities are an integer multiple of $\frac{1}{b}$, for some integer $b \leq n^8$.\footnote{The particular choice of $n^3$ and $n^8$ are immaterial, and can be any sufficiently large (fixed) polynomials in $n$. }

A solution to Select-Outcome Problem outputs an outcome $x$ such that some revenue-optimal auction for $D_1 \times D_2$, on valuation profile $(v_1,v_2)$, selects outcome $x$ with non-zero probability.\footnote{For example, in a single-item auction with two bidders there are three outcomes: give the item to bidder one, bidder two, or no one. Note that many outcomes may be correct, both due to multiplicity of optimal auctions, and due to optimal auctions being randomized. Note also that the outcome selected may be the ``null'' outcome to keep all items with the seller.}
\end{definition}

In this language, the previous paragraphs state that Select-Outcome Problem can be solved in deterministic communication complexity $O(\log n)$ when both bidders are single-dimensional (and moreover, correct prices can be found in communication $O(\log n)$ as well). Formally, our main result is that the communication complexity is exponentially higher in multi-dimensional settings.

\begin{theorem}[Main Result]\label{thm:main} Even when $D_2$ is single-dimensional, and $D_1$ is a FedEx bidder with two options, the communication complexity of Select-Outcome Problem is $\Omega(n)$. This holds for deterministic protocols, as well as randomized protocols which succeed with probability $\geq 2/3$.
\end{theorem}

We quickly motivate our precise choices in defining Select-Outcome Problem. Because all input numbers are integers $\leq \poly(n)$, Theorem~\ref{thm:main} must follow because any solution to Select-Outcome Problem requires many bits from each distribution  (and not because $\Omega(n)$ bits are required just to do arithmetic on the input). Allowing the solution to be consistent with any allocation output with non-zero probability in any optimal auction ensures that hardness follows for any reasonable alternative definition as well (and not because of some technicality associated with multiplicity or randomization of optimal auctions).

We further note that the complexity uncovered by Theorem~\ref{thm:main} arises \emph{only} in the multi-bidder setting, as the single-bidder problems for both $D_1$ and $D_2$ are quite simple. Indeed, $D_2$ is a single-dimensional bidder, so the optimal single-bidder auction for $D_2$ is just a take-it-or-leave-it price. The optimal single-bidder auction for a FedEx bidder with two options is only slightly more complex: it has menu complexity at most two,\footnote{Specifically, it offers one-day shipping at a take-it-or-leave-it price. It may additionally offer one (perhaps randomized) option to receive two-day shipping at a discount.} is computationally tractable, etc.~\cite{FiatGKK16}. Additionally, for every possible instantiation of $D_1$ or $D_2$ used in our construction, the revenue-optimal single-bidder auction simply sets a take-it-or-leave-it price of $n^2+1$ (see Proposition~\ref{prop:singlesimple}). Put another way, every instantiation of $D_1$ and $D_2$ in our construction admits the simplest possible single-bidder solution. Yet, the optimal multi-bidder auction for both together is not locally-implementable (and recall that this phenomenon cannot occur with two single-dimensional bidders).

Finally, we emphasize that locality is an intrinsic measure of complexity for multi-bidder auctions. In this sense, the communication complexity of Select-Outcome Problem serves not as a problem to be solved in practice, but rather as a quantitative lens to view the extent to which the output of an optimal multi-bidder auction for one particular input $\vec{v}$ depends on the underlying prior $\vec{D}$.

\subsection{Extension: Budget-Constrained Bidders}
In addition to our main result, we also consider bidders who are barely beyond the classic single-dimensional setting in a different direction: they have a budget constraint. Specifically, bidders have a value $v$ and a budget $B$. If they receive the item and are charged $p \leq B$, their utility is $v-p$ as usual. If they are charged $p > B$, their utility is $-\infty$. That is, the bidder's utility is \emph{not} quasi-linear.\footnote{A bidder is quasi-linear if their utility for receiving the item is $v-p$ always. Throughout the paper, bidders will always be assumed to be quasi-linear unless otherwise specified.}

If $v$ and $B$ are both private information to the bidder, this is an inter-dimensional setting~\cite{DevanurW17}, and Theorem~\ref{thm:main} already establishes that optimal mechanisms are not locally-implementable. If instead the budget is \emph{public} (known to the designer), then this is still a single-dimensional setting (because the bidder's private information is just a single value), but it is \emph{non-linear} (because the buyer is not quasi-linear). Again, this is essentially the simplest non-linear setting (and perhaps the most well-studied within the TCS literature, arguably by a significant margin): the buyer is still single-dimensional, and her utility with respect to price is piecewise-linear with two segments. We also show that optimal mechanisms for single-dimensional budget-constrained buyers are not locally-implementable.

\begin{theorem}\label{thm:mainbudget} Even when $D_2$ is single-dimensional and quasi-linear, and $D_1$ is single-dimensional with a public budget constraint, the communication complexity of Select-Outcome Problem is $\Omega(n)$. This holds for deterministic protocols, as well as randomized protocols which succeed with probability $\geq 2/3$.
\end{theorem}

We again note that this complexity arises in the multi-bidder setting, despite the fact that each single-bidder problem is quite simple. Again, $D_2$ is single-dimensional and quasi-linear, so the optimal single-bidder mechanism is just a take-it-or-leave-it price. The optimal single-bidder mechanism for a single-dimensional buyer with a public budget has menu complexity at most two,\footnote{Specifically, it offers the option to receive the item at some price $p$. If $p = B$, it may additionally offer the option to receive the item with probability $q < 1$ at price $r < B$.} is computationally tractable, etc.~\cite{ChawlaMM11}. Additionally, for every possible instantiation of $D_1$ or $D_2$ used in our construction, the revenue-optimal single-bidder auction again sets a take-it-or-leave-it price of $n^2+1$ (see Proposition~\ref{prop:budgetsinglesimple}). We additionally emphasize that our constructions for Theorems~\ref{thm:main} and Theorem~\ref{thm:mainbudget} overlap significantly (see Section~\ref{sec:proof}), highlighting the robustness of our technical contributions.

\subsection{Additional Implications for Multi-Dimensional Virtual Values}
We now provide an additional lens through which to view the implications of our main result. Specifically, several works provide some form of ``multi-dimensional virtual values''~\cite{RochetS03, CaiDW12a, AlaeiFHH13, HartlineH15, CaiDW16, Carroll17}. Their precise uses and derivations differ (see Section~\ref{sec:related} for further detail), but they all share a theme of connecting truthful revenue maximization to algorithmic virtual welfare maximization. For example,~\cite{CaiDW16} derives multi-dimensional virtual values through Lagrangian duality, and proves that for all instances $D$, there exists a virtual valuation function $\Phi_i^{D}(\cdot)$ for each bidder $i$ such that for all valuation profiles $\vec{v}$, every revenue-optimal auction must select an outcome $x$ maximizing $\sum_i (\Phi_i^{D}(v_i))(x)$.


In particular, this view shows a correspondence between Lagrangian multipliers/dual variables in a natural Linear Programming formulation (overviewed in Section~\ref{sec:prelim}) and these virtual valuation functions. This Linear Program has variables, constraints, and dual variables for each bidder.

One additional beautiful aspect of Myerson's virtual value theory is the following: $\bar{\varphi}_i^{D_i}(\cdot)$ depends \emph{only on $D_i$ and not at all on $D_{-i}$}. In the language of LP duality, this implies the following remarkable property: in the LP formulation, \emph{the optimal dual variables for Bidder $i$ also depend only on $D_i$ and not at all on $D_{-i}$}! The fact that the optimal dual variables can be computed separately for each bidder is remarkable because the optimal primal solution certainly cannot (Bidder $i$'s allocation/price variables in the optimal auction certainly depend on $D_{-i}$). To emphasize the implications of this remarkable property: consider writing, separately for each $i$, the LP formulation to optimally sell a single item to a single bidder whose value is drawn from $D_i$. To solve the multi-bidder LP formulation for $k$ bidders whose values are drawn from $\times_i D_i$, one could try na\"{i}vely stapling the $k$ optimal primals to these single-bidder LPs together. There is no reason to expect this to succeed, and indeed it fails (in fact, it will generally fail to even produce a feasible primal). On the other hand, one could alternatively try na\"{i}vely stapling together the $k$ optimal dual variables together, and hope that this produces the optimal dual variables for the $k$-bidder LP formulation. Somewhat miraculously, \emph{this latter process succeeds in any single-dimensional setting} (with quasi-linear bidders).

For multi-dimensional bidders, the~\cite{CaiDW16} framework still establishes the existence of $\Phi_i^{D}(\cdot)$, but not necessarily that $\Phi_i^{D}(\cdot)$ is agnostic to $D_{-i}$ as in the single-dimensional case. Or in the language of LPs, the optimal dual variables for Bidder $i$ may a priori depend on the entire prior, rather than just $D_i$. Viewed through this lens,~\cite{AlaeiFHH13, HartlineH15} discover restricted multi-dimensional settings where optimal duals retain this remarkable ``bidder-separable'' property. Through the same lens, Theorem~\ref{thm:main} establishes that this bidder-separable property does not generally hold in multi-dimensional settings (even with just two bidders from the two-day FedEx problem). Theorem~\ref{thm:mainbudget} rules out the bidder-separable property for single-dimensional non-linear settings as well. In addition, Theorems~\ref{thm:disjoint} and~\ref{thm:notdisjoint} further give concrete examples of how optimal dual variables for Bidder $1$ can be quite sensitive to tiny changes in $D_2$. 

\subsection{Very Brief Technical Overview}
The proof of Theorem~\ref{thm:main} follows by a reduction from \disj\ (formally defined in Section~\ref{sec:prelim}). Our reduction makes heavy use of notation and concepts from prior work, so we defer an outline of the approach to Section~\ref{sec:proof} once appropriate language is built up. We provide here a brief highlight of the main challenge: we have just spent several paragraphs in Section~\ref{sec:intro} describing all the intractable properties that revenue-optimal auctions possess. To complete a reduction, we not only need to derive the optimal auction for a single instance, but for an entire class of instances. Moreover, this class must contain sufficiently many ``intractable instances'' in order to embed \disj. Indeed, reductions to Bayesian mechanism design are scarce, technically involved, and to-date exist only for single-bidder settings~\cite{DobzinskiFK11, CaiDW13b, DaskalakisDT14, ChenDPSY14,ChenDOPSY15, CollinaW20}. In multi-bidder settings, the state-of-the-art only recently characterized optimal auctions for all instances with two additive bidders and two items where item values are drawn i.i.d.~from distributions supported on $\{1,2\}$~\cite{Yao17}!

Fortunately, the FedEx setting is a sweet spot which is both rich enough for optimal mechanisms to be non-locally-implementable, yet also structured enough for a tractable reduction to optimal mechanism design. We hope that the proof outline in Section~\ref{sec:proof} may serve as a roadmap for potential future reductions, recalling that any complexities established for the FedEx setting extend to the (significantly more general) multi-dimensional unit-demand setting as well.

\subsection{Related Work}\label{sec:related}
\noindent\textbf{Complexity of Multi-Dimensional Mechanism Design.} We have already discussed the thematically most-related work, which identifies formal complexity measures separating revenue-optimal single- and multi-dimensional mechanisms~\cite{Thanassoulis04, ManelliV07, Pavlov11, BriestCKW15, HartN19, DaskalakisDT14, HartR15, ChenDPSY14, ChenDOPSY15, RubinsteinW18, DaskalakisDT17, BabaioffGN17, Yao17}. Among these, only~\cite{Yao17} explicitly studies the multi-bidder setting, and establishes that while optimal single-dimensional auctions are dominant-strategy truthful,\footnote{An auction is dominant-strategy truthful if it is in each bidder's interest to report their true valuation no matter the other bidders' reports. Contrast this with Bayesian Incentive Compatible (defined in Section~\ref{sec:prelim}).} optimal multi-dimensional auctions are not~\cite{Yao17}. In comparison to this line of works, our paper provides a novel complexity \emph{unique to multi-bidder settings}. Specifically, our work identifies complexity in broad classes of two-bidder settings, even thoug the single-bidder problem for every instance in these classes is quite simple. \\

\noindent\textbf{Multi-Dimensional Virtual Values.} Our work uses multi-dimensional virtual values in order to prove optimality of mechanisms. As previously referenced, several prior works introduce various notions of multi-dimensional virtual values~\cite{RochetS03, CaiDW12a, AlaeiFHH13, HartlineH15, CaiDW16, Carroll17}. Some of these works consider continuous distributions, and derive multi-dimensional virtual values by explicitly choosing paths along which the incentive constraints might bind and then doing integration by parts. Others consider discrete settings, and derive multi-dimensional virtual values by drawing a connection to LP duality. Because our setting is discrete (and necessarily so, in order for communication complexity to be a meaningful measure), we adopt the language used in~\cite{CaiDW16}, which uses the lens of LP duality. 

Within this line of works,~\cite{AlaeiFHH13, HartlineH15} also prove optimality in some multi-bidder settings, and in particular discover restricted settings where multi-dimensional virtual values are bidder-separable (termed ``revenue-linear'' and ``MR-log-supermodular'', respectively). In their language, our main results rule out any extension to two-bidder settings generally (even when one bidder is single-dimensional, and the other is a two-day FedEx bidder or single-dimensional with a public budget). In our language,~\cite{AlaeiFHH13, HartlineH15} discover restricted multi-dimensional settings where optimal mechanisms are locally-implementable.\\

\noindent\textbf{Interdimensional Mechanism Design.}~\cite{FiatGKK16} introduce the FedEx problem, and note that optimal single-bidder mechanisms inherit some-but-not-all of the nice properties of single-dimensional settings, along with some-but-not-all of the complexities associated with multi-dimensional settings. Our main result considers a bidder from the FedEx problem, and therefore our technical setup is similar to works such as~\cite{FiatGKK16, DevanurW17,DevanurHP17, DevanurGSSW20}, but there is not much overlap with these works beyond Section~\ref{sec:prelim}.\\

\noindent\textbf{Budget-Constrained Bidders.} There is a substantial body of work involving mechanism design for budget-constrained buyers. The most related works to ours design revenue-optimal single-item auctions in Bayesian settings for buyers with a public or private budget constraint. Here,~\cite{LaffontR96, CheG00, ChawlaMM11, DevanurW17} characterize the optimal single-buyer auction.\footnote{Specifically, it sets a single price if the budget is public and the valuation distribution is regular~\cite{LaffontR96}. It has menu complexity at most two if the budget is public, no matter the valuation distribution~\cite{ChawlaMM11}. It has menu complexity at most $k$ if the budget is private and drawn from a distribution of support at most $k$, and the valuation conditioned on each possible budget satisfies a condition called ``decreasing marginal revenues.'' It has menu complexity at most $3 \cdot 2^{k-1}-1$ if the budget is private and drawn from a distribution of support at most $k$, no matter the joint distribution of (value, budget) pairs.} Several works also identify tractable structure for the optimal auction in restricted cases. For example,~\cite{LaffontR96} consider the case of multiple buyers with values drawn i.i.d.~from the same regular distribution and an identical public budget, and ~\cite{PaiV14} consider the case that each buyer has a private budget and their value is drawn independently of their budget from an MHR distribution with decreasing density. In this context, our results establish that while optimal single-bidder mechanisms for budget-constrained buyers are quite tractable, optimal multi-bidder auctions remain intractable (without the restrictions imposed in works such as~\cite{LaffontR96,PaiV14}).\\

\noindent\textbf{Reductions in Bayesian Mechanism Design.} We have also briefly discussed reductions to optimal mechanism design, which previously exist only in single-bidder settings~\cite{DobzinskiFK11, CaiDW13b, DaskalakisDT14, ChenDPSY14,ChenDOPSY15, CollinaW20}. Other styles of single-bidder reductions have been used to special cases of optimal single-bidder mechanism design (such as finding the optimal deterministic auction)~\cite{BriestK07, DaskalakisDT12, ChenMPY18}. In comparison to this line of works, our paper provides a technical contribution via the first reduction to multi-bidder Bayesian mechanism design.\\

\noindent\textbf{Communication Complexity in Multi-Dimensional Mechanism Design.} Recent work of~\cite{BabaioffGN17} identifies a connection between the so-called \emph{menu complexity} of single-bidder auctions and the deterministic communication required to implement it. More recent work of~\cite{RubinsteinZ21} further considers the randomized communication complexity required to implement single-bidder auctions, and in particular establishes that randomized implementations of auctions may sometimes communicate exponentially fewer bits than deterministic implementations. While this model is incomparable to ours,\footnote{Specifically, these works study a single-bidder problem where the prior (and therefore the auction) is fully-known. Their goal is to implement the auction for a particular valuation $v$ without necessarily learning $v$.} this context makes it significant that Theorems~\ref{thm:main} and~\ref{thm:mainbudget} hold for randomized communication protocols. 


There is also a substantial body of work at the intersection of communication complexity and mechanism design generally, following seminal work of~\cite{NisanS06}. A parallel line of works following~\cite{FadelS09} considers the communication overhead specifically to compute payments (on top of any communication necessary to determine an outcome/allocation). Follow-up works of~\cite{BabaioffBS13,RubinsteinSTWZ21, DobzinskiR21} show that this overhead can be quite significant, even with just two players. While their model is also incomparable to ours,\footnote{For example, none of these works consider revenue-optimization at all.} this context makes it significant that our main results provide communication lower bounds \emph{just to solve Select-Outcome Problem} (rather than to solve Select-Outcome Problem and to also determine the payments).

\subsection{Summary and Roadmap}
We establish that optimal multi-dimensional mechanisms are not locally-implementable: \emph{executing the auction on just a single valuation profile requires knowing (essentially) the entire distribution}. Formally, we study the communication complexity of Select-Outcome Problem. In single-dimensional settings, Select-Outcome Problem can be solved with $O(\log n)$ bits of communication, while simply stating the valuation profile also requires $\Theta(\log n)$ bits. In multi-dimensional settings, Theorem~\ref{thm:main} gives a communication lower bound of $\Omega(n)$ on Select-Outcome Problem: exponentially more than the $\Theta(\log n)$ bits needed to state the input valuation profile, and nearly the $\Theta(n \log n)$ bits sufficient to fully define each $D_i$. In particular, recall that all $D_i$ in our construction are especially simple from the single-bidder perspective: the optimal single-bidder auction sets a take-it-or-leave-it price of $n^2+1$ (Proposition~\ref{prop:singlesimple}). This makes clear that non-locality truly arises due to complexity of \emph{multi-bidder} multi-dimensional auctions, and not simply due to complexity of the corresponding single-bidder problem. 

Section~\ref{sec:prelim} provides preliminaries. Section~\ref{sec:proof} provides a high-level overview of our approach. Section~\ref{sec:reduction} provides our reduction, and Section~\ref{sec:flow} our analysis. Section~\ref{sec:conclusion} provides concluding thoughts. The appendix contains omitted proofs. In particular, Appendix~\ref{sec:budget} contains our complete analysis for the case of single-dimensional buyers with a public budget.

\section{Preliminaries}\label{sec:prelim}
Below, we provide detailed preliminaries for our main result (Section~\ref{sec:fedexprelim}) and detailed background on Lagrangian duality for Bayesian mechanism design (Section~\ref{sec:dualityprelim}), so that we can present the key ideas behind our construction. We also formally define the setting we consider for our extension to budget-constrained bidders in Section~\ref{sec:budgetprelim} but defer to Section~\ref{sec:budgetprelimfull} full preliminaries necessary for the proofs. Section~\ref{sec:disjprelim} quickly states the communication problem of \disj, which we use in our reductions.

\subsection{The FedEx Problem}\label{sec:fedexprelim}

\noindent\textbf{Setup and Notation.} Our main result holds already when there are just two bidders and two options, which we refer to as day1 and day2. Bidders have a value and an interest. A bidder with (value, interest) pair $(v,1)$ receives value $v$ if they receive one-day shipping, and $0$ if they receive two-day shipping. A bidder with (value, interest) pair $(v,2)$ receives value $v$ if they receive either one-day or two-day shipping. A bidder's \emph{type} stores her full (value, interest) pair.

Each of the two bidders $i$ have $2 n_i + 1$ different types, we label them from $t^0_i$ to $t^{2n_i}_i$. There are $n_i$ possible values among all types, which we label as $v_i^j$, for $j \in [n_i]$. In this labeling, $t^{2k-1}_i$ represents the (value, interest) pair $(v^{k}_i,1)$ and $t^{2k}_i$ represents the (value, interest) pair $(v^{k}_i,2)$. Finally, $t^0_i$ represents not participating in the auction, and has value $v^0_i:=0$. We will alternate between referring to types as $t^{2k+j-2}$ and $(v_i^k,j)$, depending on which notation is cleaner. We denote by $f_i(t_i)$ the probability that bidder $i$ has type $t_i$, and we use $D_i$ to represent the distribution of bidder $i$. Finally, we will also use the notation $R_i((v_i^k,j)):=\sum_{k' \geq k} f_i((v_i^{k'},j))$.\footnote{Observe that $R_i(\cdot)$ is essentially a reverse CDF. Indeed, if there were only one possible interest, the definition would imply that $R_i(\cdot):=1-F_i(\cdot)$, where $F_i(\cdot)$ is the CDF for Bidder $i$.}\\


\noindent\textbf{Optimal Auctions.} We have one item to ship, and can ship it to either bidder using either one- or two-day shipping.\footnote{If desired, our construction can be easily modified so that the auctioneer has a copy of the item shippable on each day, and the bidders are unit-demand (or to many other settings), but we only present one to establish the desired hardness.} Note that this is a ``service-constrained'' environment, as defined in~\cite{AlaeiFHH13}. We seek the revenue-optimal Bayesian Incentive Compatible (BIC) auction, which asks each bidder to report their (value, interest) pair, and then decides to whom to ship the item (or to no one) and in how many days. It is observed in~\cite{FiatGKK16} that the revenue-optimal auction w.l.o.g.~always ships the item on the reported interest (that is, it will never ship the item in two days to a bidder whose interest is day1, or vice versa).\footnote{To quickly see this: observe that two-day shipping an item to a bidder with interest day1 gives them zero value, so the item may as well not be shipped. A bidder with day2 interest is indifferent between one-day and two-day shipping, so giving them two-day instead of one-day shipping does not affect their utility and makes other types of that bidder \emph{less} interested in misreporting this type.} With this in mind, a mechanism is defined by its ex-post allocation rule $X$ and ex-post payment rule $P$. Here, $X_{i}(t_1,t_2)$ denotes the probability that bidder $i$ is shipped the item (matching their interest) when the reported types are $t_1,t_2$, and $P_i(t_1,t_2)$ denotes the payment made by bidder $i$ when the reported types are $t_1,t_2$. Because we have one copy of the item to ship, an allocation rule is feasible iff for all $t_1,t_2$, $X_1(t_1,t_2) + X_2(t_1,t_2) \leq 1$.

An auction is Bayesian Incentive Compatible (BIC) if it is in each bidder's interest to report their true type in expectation over the types of the other bidder. More specifically, the revenue-optimal BIC auction is the solution to the following linear program. In the LP, the variables are $X,P, \pi,p$. $X$ and $P$ refer to the ex-post allocation/price rules, as defined above. $\pi, p$ refer to the interim allocation/price rules, which satisfy the equalities in Equations~\eqref{eq:interim} and~\eqref{eq:price}.

\begin{align}
\max_{X,P,\pi, p} \quad & \sum_{i} \sum_{j = 1}^{2n_i} f_i(t^j_i)\cdot p_i(t^j_i)\nonumber\\
\textrm{subject to} \quad & X_i(t^j_1,t^\ell_2) \in [0,1]\text{ for all bidders $i$ and all $j,\ell$}.\nonumber\\
&X_1(t_1^0,t_2^\ell) = P_1(t_1^0,t_2^\ell) =0\text{ for all $\ell \in [0,2n_2]$}.\nonumber\\
&X_2(t_1^k,t_2^0) = P_2(t_1^k,t_2^0) =0\text{ for all $k \in [0,2n_1]$.}\nonumber\\
&X_1(t_1^j, t_2^\ell) + X_2(t_1^j,t_2^\ell) \leq 1\text{ for all $j \in [2n_1],\ell \in [2n_2]$.}\nonumber\\
\quad & \pi_i(t^j_i)=\sum_{\ell=1}^{2n_{3-i}}f_{3-i}(t^\ell_{3-i}) \cdot X_i(t^j_i; t^\ell_{3-i}) \text{ for all bidders $i$ and $j \in [0,2n_i]$ }. \label{eq:interim}\\
 \quad & p_i(t^j_i)=\sum_{\ell=1}^{2n_{3-i}}f_{3-i}(t^\ell_{3-i}) \cdot P_i(t^j_i; t^\ell_{3-i})\text{ for all bidders $i$ and $j \in [0,2n_i]$}.\label{eq:price}\\
 & \pi_i((v_i^k,j)) \cdot v^k_i - p_i((v_i^k,j)) \ge \pi_i((v_i^{k'},j')) \cdot v^k_i - p_i((v_i^{k'},j')), \label{eq:bic}\\
 &\qquad \text { for all bidders $i$, all $k, k' \in [0,n_i]$, and $2 \geq j \geq j' \geq 1$.} \nonumber
\end{align}

The objective is simply the expected revenue. Constraints~\eqref{eq:interim} and~\eqref{eq:price} simply confirm that the interim rules are computed correctly. Equation~\eqref{eq:bic} guarantees that the mechanism is BIC. In particular, Equation~\eqref{eq:bic} observes that it is only necessary to ensure that bidders don't wish to underrepresent their interest (because overrepresenting their interest guarantees them non-positive utility).\\

\noindent\textbf{Payment Identity.} Myerson's payment identity provides a closed-form to compute revenue-maximizing payments for a fixed (monotone) allocation rule. Observe in particular that a payment rule satisfying the payment identity exists for any (monotone) allocation rule.

\begin{definition}[Monotone, Payment Identity] An interim allocation rule $\pi$ is \emph{monotone} if for both players $i$ and days $j$: $\pi_i( (\cdot,j))$ is monotone non-decreasing. $p$ satisfies the \emph{payment identity} for $\pi$ if for both players $i$, days $j$, and all $k$, we have: $p_i((v_i^{k},j))= \sum_{\ell = 1}^k v_i^\ell \cdot (\pi_i((v_i^\ell,j)) - \pi_i((v_i^{\ell-1},j)))$.
\end{definition}

\subsection{Lagrangian Duality}\label{sec:dualityprelim}

The purpose of this section is to build up the necessary notation/concepts in order to state Definition~\ref{def:witness} and Theorem~\ref{thm:CDW} at the end. Theorem~\ref{thm:CDW} provides an approach to claim that a mechanism is or isn't optimal for a given instance. This approach uses Lagrangian duality, and specifically the language adopted in~\cite{CaiDW16}. More specifically, we will put Lagrangian multipliers on the BIC constraints in the following manner, which creates a Lagrangian relaxation:

\begin{enumerate}
    \item[(i)] For constraints of the form: $\pi_i((v^k_i,2)) \cdot v^{k}_i - p_i((v^k_i,2)) \ge \pi_i((v_i^k,1)) \cdot v^{k}_i - p_i(v^k_i,1))$, we use a Lagrangian multiplier of $\alpha_i(k)$ (for all bidders $i$ and $k \in [1,n_i]$).
    \item[(ii)] For constraints of the form: $\pi_i((v^k_i,j)) \cdot v^{k}_i - p_i((v^k_i,j)) \ge \pi_i((v^{k-1}_i,j)) \cdot v^{k}_i - p_i((v^{k-1}_i,j))$, we use a Lagrangian multiplier of $\lambda_i^j(k)$ (for all bidders $i$, items $j$, and $k \in [1,n_i]$).
\item [(iii)] For all remaining BIC constraints, we use a Lagrangian multiplier of $0$.
\item [(iv)] To emphasize: for all other constraints (i.e. all the constraints which are unrelated to BIC), we don't use Lagrangian multipliers, and keep them as constraints.
\end{enumerate}

Constraints in (i) guarantee that the bidder will not misreport its interest, and constraints in (ii) guarantee that the bidder will not underreport their value by the minimal amount possible. Definitions~\ref{def:flow} and~\ref{def:vv}, and Theorem~\ref{thm:CDW} below specialize the~\cite{CaiDW16} framework to our setting. We refer the reader to~\cite{CaiDW16} for further details surrounding their framework, but give brief intuition for each definition throughout. Recall that every choice of Lagrangian multipliers $(\alpha,\lambda)$ induces a Lagrangian relaxation with objective function:
\begin{align*}
 \mathcal{L}(\alpha, \lambda)&:=\sum_i \sum_{j=1}^{2n_i} f_i(t_i^j) \cdot p_i(t_i^j)+ \sum_i\sum_{k=1}^{n_i} \alpha(k) \cdot \left(\pi_i((v_i^k,2)) \cdot v^k_i - p_i((v_i^k,2)) - \pi_i((v_i^{k},1)) \cdot v^k_i + p_i((v_i^{k},1))\right)\\
&+ \sum_i \sum_{k=1}^{n_i}\sum_{j=1}^2 \lambda_i^j(k) \cdot \left( \pi_i((v_i^k,j)) \cdot v^k_i - p_i((v_i^k,j)) - \pi_i((v_i^{k-1},j)) \cdot v^k_i + p_i((v_i^{k-1},j))\right).
\end{align*}

The constraints are the same as in the initial LP, except removing the BIC constraints. The first concept in the~\cite{CaiDW16} framework is that of a flow:

\begin{definition}[Flow]\label{def:flow} A set of Lagrangian multipliers form a \emph{flow} if the following hold for all $i$:
\begin{itemize}
\item $f_i(t_i^{2k-1}) +\lambda_i^1(k+1) +\alpha_i(k) = \lambda_i^1(k)$, for all $k \in [1,n_i-1]$.
\item $f_i(t_i^{2n_i-1}) + \alpha_i(n_i) = \lambda_i^1(n_i)$.
\item $f_i(t_i^{2k}) +\lambda_i^2(k+1) = \alpha_i(k) + \lambda_i^2(k)$, for all $k \in [1,n_i-1]$.
\item $f_i(t_i^{2n_i}) = \alpha_i(n_i) + \lambda_i^2(n_i)$.
\end{itemize}
\end{definition}

Intuitively, Definition~\ref{def:flow} captures the following. In the relaxation, there are no constraints on the payment variables \emph{at all}, so the relaxation is unbounded if any payment variable has a non-zero coefficient in $\mathcal{L}(\alpha, \lambda)$. $(\alpha, \lambda)$ form a flow if and only if all payment variables have a coefficient of zero in $\mathcal{L}(\alpha, \lambda)$.

\begin{definition}[Virtual Values]\label{def:vv} For a given set of Lagrangian multipliers $\alpha, \lambda$, define:\footnote{For simplicity of notation, denote by $\lambda_i^j(n_i+1) := 0$, $v_i^{n_i + 1}:=v_i^{n_i}$.} 
$$\Phi_{i}^{\alpha,\lambda}((v_i^k,j)):=v_i^k - \frac{(v_i^{k+1}-v_i^k)\cdot \lambda_i^j(k+1)}{f_i((v_i^k,j))}.$$


\end{definition}

\begin{observation}[\cite{CaiDW16}]\label{obs:vv} For any $(\alpha,\lambda)$ which form a flow: $\mathcal{L}(\alpha, \lambda)= \sum_i \sum_{j=1}^{2n_i} f_i(t_i^j) \cdot \pi_i(t_i^j) \cdot \Phi_i^{\alpha,\lambda}(t_i^j).$
\end{observation}

Intuitively, Observation~\ref{obs:vv} follows from algebraic manipulation, and Definition~\ref{def:vv} is made for the sole purpose of yielding Observation~\ref{obs:vv}, as it suggests that any optimal allocation rule for a particular Lagrangian relaxation should award the item to the bidder with highest virtual value (according to Definition~\ref{def:vv}).

\begin{definition}[Witness Optimality]\label{def:witness} Let $(\alpha,\lambda)$ be a flow and $(X,P)$ be a BIC auction such that:
\begin{itemize}
\item $p$ satisfies the payment identity for $\pi$.
\item For all $k$: $\alpha_i(k) >0 \Rightarrow \pi_i((v_i^k,2)) \cdot v^{k}_i - p_i((v_i^k,2)) = \pi_i((v_i^k,1)) \cdot v^{k}_i - p_i((v_i^k,1))$.
\item On all $(t_1^k,t_2^{k'})$, $X$ awards the item to a bidder with highest non-negative virtual value.
\end{itemize}
Then we say that $(\alpha,\lambda)$ \emph{witnesses optimality} for $(X,P)$, and $(X,P)$ \emph{witnesses optimality} for $(\alpha,\lambda)$.\footnote{\cite{HartlineH15} note that some $(\alpha,\lambda)$ cannot be optimal for any instance, because they cannot witness optimality for any incentive compatible $(X,P)$. However, note that the optimal $((\alpha, \lambda),(X,P))$ witness optimality for each other, by strong Lagrangian duality and complementary slackness.}
\end{definition}

\begin{theorem}[\cite{CaiDW16}]\label{thm:CDW} Let $(\alpha,\lambda)$ witness optimality for $(X,P)$. Then $(X,P)$ is a revenue-optimal BIC auction. Moreover, all revenue-optimal auctions witness optimality for $(\alpha,\lambda)$.
\end{theorem}

Intuitively, $(\alpha,\lambda)$ and $(X,P)$ witness optimality if $(X,P)$ is optimal for the Lagrangian relaxation induced by $(\alpha,\lambda)$ (bullet three), and also $(\alpha,\lambda)$ and $(X,P)$ satisfy complementary slackness (bullets one/two).

\subsection{Public Budget Constraints}\label{sec:budgetprelim}
\noindent\textbf{Setup and Notation.} Our main extension considers bidders with a (value, budget) pair. A bidder with value $v$ and budget $B$ enjoys utility $v-p$ if they receive the item and pay $p \leq B$, and utility $-\infty$ if they pay price $p > B$. Each of the two bidders $i$ have $n_i + 1$ different types. Because the budget is public, their type is fully specified by a value, which we label $v_i^0,\ldots, v_i^{n_i}$. Again, $v_i^0$ refers to non-participation in the auction.\\

\noindent\textbf{Optimal Auctions.} We have one item for sale, and can give it to either bidder. We again seek the revenue-optimal BIC auction. Because the bidders are not quasi-linear, we must also specify that we seek an \emph{ex-post} individually rational auction.\footnote{When bidders are quasi-linear, any interim individually rational auction can be made ex-post individually rational with a simple reduction. This reduction fails when bidders have budget constraints.} That is, even after learning the bid of the other player, and learning the outcome of all random coins of the mechanism, each bidder has non-negative utility. {Appendix~\ref{sec:budgetprelimfull} provides a linear program for this setting, and more detailed preliminaries similar to Section~\ref{sec:fedexprelim}.}

\subsection{Disjointness}\label{sec:disjprelim}
Our communication complexity lower bound provides a reduction from \disj. In \disj, Alice is given $\vec{x} \in \{0,1\}^n$, Bob is given $\vec{y} \in \{0,1\}^n$, and their goal is to determine whether there exists an $i$ such that $x_i = y_i = 1$. It is known that any deterministic communication protocol resolving \disj\ requires communication at least $n$, and any randomized protocol resolving \disj\ correctly with probability~at least $2/3$ requires communication $\Omega(n)$~\cite{KalyanasundaramS92,Razborov92,KushilevitzN97}.

\section{Proof Overview}\label{sec:proof}
Our proof of Theorems~\ref{thm:main} and~\ref{thm:mainbudget} both follow the same outline below. All steps below apply to both proofs, although the referenced technical sections are for Theorem~\ref{thm:main} (where significantly more detail is provided).
\begin{enumerate}
    \item[(i)] Section~\ref{sec:reduction} defines our reduction from \disj. Specifically, we define a mapping from Alice's input $x$ to a distribution $D_1$, and from Bob's input $y$ to a distribution $D_2$. Section~\ref{sec:reduction} states several properties of our reduction that will be used in later proofs. Here is an informal overview of some key properties.
\begin{itemize}
\item All values in our constructions lie in $\{n^2+1,n^2+2,\ldots, n^2+n+2\}$.
\item All distributions used in our constructions are nearly-uniform. Therefore, the optimal single-bidder auction for any distribution in our constructions is quite simple, and sets a price of $n^2+1$.
\item Depending on the input to \disj, the distribution is perturbed slightly at all values.
\end{itemize}
\item[(ii)] Section~\ref{sec:canonical} analyzes the canonical ``Myerson flow'' (see Section~\ref{sec:canonical} for definition) for our construction, which yields virtual values equal to Myersonian virtual values. Refer to this flow as $(\alpha_1,\lambda_1)$.
\begin{itemize}
\item We consider the allocation rule that awards the item to the bidder with maximum $\Phi_i^{\alpha_1,\lambda_1}(t_i)$ (and charges prices according to the payment identity). Refer to this auction $(X_1,P_1)$.
\item \emph{If and only if $\disj(x,y) = \mathsf{Yes}$}, $(X_1,P_1)$ happens to be a second price auction, breaking ties for Bidder One (Definition~\ref{def:SPA}). We then show that $(X_1,P_1)$ witnesses optimality for $(\alpha_1,\lambda_1)$ \emph{if and only if $\disj(x,y) = \mathsf{Yes}$}.
\item This means that when $\disj(x,y) = \mathsf{Yes}$, we've now found the optimal dual ($(\alpha_1,\lambda_1)$) and optimal auction (Definition~\ref{def:SPA}).
\item We also show that $\Phi^{\alpha_1,\lambda_1}_{1}(t^1_1) > \Phi^{\alpha_1,\lambda_1}_{2}(t^1_2) > 0$. 
\item Now, by Theorem~\ref{thm:CDW}, this leads to our first key conclusion: \emph{when $\disj(x,y) = \mathsf{Yes}$, every optimal auction must have $X_1(t_1^1, t_2^1) = 1$}.
\end{itemize}
\item[(iii)] Section~\ref{sec:modifymain} modifies the canonical Myerson flow, for instances where $\disj(x,y) = \mathsf{No}$.
\begin{itemize} 
\item For the setting of Theorem~\ref{thm:main}, $(X_1,P_1)$ is not BIC: when Buyer One has (value, interest) pair $(n^2+n+2,2)$, she would rather misreport $(n^2+n+2,1)$. For Theorem~\ref{thm:mainbudget}, $(X_1,P_1)$ is not budget-respecting: Buyer One with value $n^2+n+2$ would have to pay more than her budget.
\item We increase the Lagrangian multiplier for the violated constraint from the previous bullet, and adjust others in order to preserve flow-conservation. This step is the most intricate, and requires a very precise setting of each multiplier. Call this flow $(\alpha_2,\lambda_2)$. 
\item We next find an allocation rule that witnesses optimality for $(\alpha_2,\lambda_2)$, $(X_2,P_2)$. $(X_2,P_2)$ is also a second-price auction, but ties must be broken in a precise (randomized) manner (Definition~\ref{def:careful}). We show that $(X_2,P_2)$ witnesses optimality for $(\alpha_2, \lambda_2)$ \emph{if and only if $\disj(x,y) = \mathsf{No}$}. This step is also delicate, as we must simultaneously satisfy several constraints related to Theorem~\ref{thm:CDW}.
\item We also show that $0 < \Phi_1^{\alpha_2,\lambda_2}(t_1^1)< \Phi_2^{\alpha_2,\lambda_2}(t_2^1)$. 
\item Now, by Theorem~\ref{thm:CDW}, this leads to our second key conclusion: \emph{when $\disj(x,y) = \mathsf{No}$, every optimal auction must have $X_1(t_1^1,t_2^1) = 0$}. 
\end{itemize}
\item[(iv)] To conclude, (ii) and (iii) together establish that when $\disj(x,y) = \mathsf{Yes}$, \emph{every optimal auction awards bidder $1$ the item with probability $1$ on input $(t_1^1,t_2^1)$}. On the other hand, when $\disj(x,y) = \mathsf{No}$, \emph{every optimal auction awards bidder $1$ the item with probability $0$ on input $(t_1^1,t_2^1)$}. Therefore, if we know any outcome consistent with any optimal auction on $(t_1^1,t_2^1)$, we know $\disj(x,y)$.
\end{enumerate}

\section{Our Reduction and its Properties}\label{sec:reduction}
In this section, we define our reduction and state some useful properties. First, we define the type space. Throughout this section, $n$ denotes the size of the input to \disj. We state the concrete lemmas which are relevant to give a detailed technical proof overview, but all proofs of these lemmas are in Appendix~\ref{app:reduction}. Note that the purpose of this section is \emph{only} to define our flow and state basic properties. We will give intuition for these decisions in the subsequent sections as it will only become clear once we define our flow.\\

\noindent\textbf{The type space.} In our reduction, the type space does not depend on $x,y$ (only the distribution does). For every input, $n_1=n_2 = n+2$ (meaning that each bidder has a total of $n+2$ non-zero types per day, and $2n+4$ non-zero types in total). For all $k \in [n+2]$ and both $i$, $v_i^k := n^2+k$.\\

\noindent\textbf{The distribution.} The distribution in our construction depends on $x,y$, but in all cases is nearly uniform. Below, for simplicity of notation let $b := 10n^6, a := \frac{b - n^5}{n+1}$. All probabilities will be an integer multiple of $\frac{1}{2b}$. Below, Bidder One's day1 distribution is fixed, and does not depend on $x$.

\begin{definition}[Bidder One's day1 distribution] Define $f_1((v_1^k,1))$ (for all $x$) as follows:
\begin{enumerate}
\item Define $f_1((v_1^1,1)):=\frac{\frac{b}{10n}}{2b} = \frac{1}{20n}$.
\item For $k=1$ to $n$, first define helper $z_{k+1}:=\frac{b-\sum_{j=1}^k f_1( (v_1^j,1))\cdot 2b}{n-k+2}$, then define $f_1((v_1^{k+1},1)):= \frac{\lfloor z_{k+1} + \frac{n^3}{n-k+2}\rfloor}{2b}$.
\item For $k=n+1$, define helper $z_{n+2}:=b-\sum_{j=1}^{n+1} f_1( (v_1^j,1))\cdot 2b $, then define $f_1((v_1^{n+2},1)):= \frac{z_{n+2}}{2b}$.
\end{enumerate}
\end{definition}

We quickly establish that the total mass of Bidder One on day1 is always $1/2$ in this construction.
\begin{lemma}\label{lem:subprobs11}
$\sum_{k=1}^{n+2} f_1((v_1^k,1)) = 1/2$.
\end{lemma}

Lemma~\ref{crange} is one key property which will be useful in our later analysis. It states that Bidder One's day1 distribution is nearly-uniform over $v_1^2,\ldots, v_1^{n+2}$ (recall that $a> n^5$).

\begin{lemma}\label{crange}
$f_1((v_1^k,1)) \cdot 2b \in [a - 2n^3, a + 2n^3]$ for all $k \in [2,n+2]$.
\end{lemma}

We now proceed to construct the day2 distribution for Bidder One. Bidder One's day2 distribution depends on $x$, and is constructed so that $x_k$ has a significant impact on $f_1((v_1^{k+1},2))$.

\begin{definition}[Bidder One's day2 distribution] Define $f_1((v_1^k,2))$ (as a function of $x$) as follows:
\begin{enumerate}
    \item Set $f_1((v_1^1,2) := \frac{\frac{b}{10n}}{2b}  = \frac{1}{20n}$, for all $x$.
    \item For $k=1$ to $n$, define helper $z_{k+1} = \frac{b - \sum_{j = 1}^{k} f_1((v_1^j,2)) \cdot 2b}{n-k+2}$.
\begin{itemize}
\item If $x_k=0$, then set $f_1((v_1^{k+1},2)) := \frac{\left \lfloor z_{k+1} + \frac{n^3}{n - k + 2}  \right \rfloor}{2b}$.
\item Otherwise ($x_k=1$), set $f_1((v_1^{k+1},2)):= \frac{\left \lceil z_{k+1} \right \rceil}{2b}$.
\end{itemize}
    \item For $k=n+1$, define helper $z_{n+2} := b - \sum_{j = 1}^{n+1} f_1((v_1^{j},2)) \cdot 2b$. Set $f_1((v_1^{n+2},2)) := \frac{z_{n+2}}{2b}$
\end{enumerate}

\end{definition}

The two lemmas below similarly establish that the total mass of Bidder One on day2 is always $1/2$, and that Bidder One's day2 distribution is always nearly-uniform over $v_1^2,\ldots, v_1^{n+2}$.

\begin{lemma}\label{lem:subprobs12}
$\sum_{k=1}^{n+2} f_1((v_1^k,2)) = 1/2$.
\end{lemma}

\begin{lemma} \label{drange}
For all $x$, $f_1((v_1^k,2)) \cdot 2b \in [a - n^3, a + n^3]$ for all $k \in [2,n+2]$.
\end{lemma}

Finally, we define the distribution for Bidder Two day1. Bidder Two's distribution will be truly single-parameter in that their interest is always day1. Bidder Two's distribution depends on $y$, and is constructed so that $y_k$ has a significant impact on $f_2((v_2^{k+1},1))$.

\begin{definition}[Bidder Two's distribution] Define $f_2((v_2^k,1))$ (as a function of $y$) as follows:
\begin{enumerate}
    \item Set $f_2((v_2^1,1)) := \frac{\frac{b}{10n} - 1}{b}$.
    \item For $k=1$ to $n$, define helper $z_{k+1} := \frac{b - \sum_{j = 1}^{k} f_2((v_2^j,1)) \cdot b}{n - k + 2}$.
\begin{itemize}
\item If $y_k = 1$, then set $f_2((v_2^{k+1},1)) :=  \frac { \left \lfloor z_{k + 1} + \frac{n^2}{n-k+2} \right  \rfloor} {b}$.
\item Otherwise, set $f_2((v_2^{k+1},1)) := \frac{\left \lfloor z_{k + 1} - 1\right  \rfloor}{b} $.
\end{itemize}
    \item For $k=n+1$, define helper $z_{n+2} := b - \sum_{j = 1}^{n+1} f_2((v_2^{j},1)) \cdot b$, set $f_2((v_2^{n+2},1)) := \frac{z_{n+2}} {b}$
\end{enumerate}

\end{definition}

Again, we confirm quickly that this is a valid distribution, and that it is nearly-uniform over $v_2^2,\ldots, v_2^{n+2}$.

\begin{lemma}\label{lem:subprobs2}
$\sum_{k=1}^{n+2} f_2((v_2^k,1)) = 1$.
\end{lemma}

\begin{lemma}\label{erange}
$f_2((v_2^k,1)) \cdot b \in [a - 2n^3, a + 2n^3]$ for all $k \in [2,n+2]$.
\end{lemma}

Finally, we quickly state that the optimal \emph{single-bidder} auction for any distribution considered in our reduction is especially simple: it sets the same take-it-or-leave-it price of $n^2+1$. The proof is in Appendix~\ref{app:onebidder}.

\begin{proposition}\label{prop:singlesimple} For all $x$ (resp., $y$), the revenue-optimal single-bidder auction for the resulting distribution $D_1$ (resp. $D_2$) simply sets a take-it-or-leave-it price of $n^2+1$ on one-day shipping.
\end{proposition}

\section{Constructing a Flow}\label{sec:flow}
In this section, we construct a flow which is optimal for all instances of our construction. We proceed in two steps. First, we consider a canonical flow and establish that this flow is optimal if and only if $\disj(x,y)=\mathsf{yes}$. Next, we show how to modify the flow to be optimal when $\disj(x,y) = \mathsf{no}$.

\subsection{A Canonical Flow}\label{sec:canonical}
We first define a canonical flow, and then argue it is optimal when $\disj(x,y) = \mathsf{yes}$.

\begin{definition}[Canonical Flow] $(\alpha^M, \lambda^M)$ is the canonical Myerson flow, where: 
\label{canonical}
\begin{itemize}
\item $\alpha_i(k) = 0$ for both bidders $i$ and all $k \in [n+2]$.
\item $\lambda_i^j(k) = R_i((v_i^k,j))$ for both bidders $i$, days $j$, and all $k \in [n+2]$.
\end{itemize}
\end{definition}

It is easy to confirm that $(\alpha^M, \lambda^M)$ is a flow. We can also quickly execute Definition~\ref{def:vv} (recalling that $v_i^k=n^2+k$) to compute $\Phi^{\alpha^M,\lambda^M}$:

\begin{observation} For both bidders $i$, days $j$, and all $k \in [1,n+2]$, $\Phi^{\alpha^M,\lambda^M}_i((v_i^k,j))= v_i^k - \frac{R_i((v_i^{k+1},j))}{f_i((v_i^k,j))}$.
\end{observation}

Proposition~\ref{prop:Myerson} below captures the key properties of our construction and this flow. These properties are motivated by bullet three of Definition~\ref{def:witness}: we need to compare virtual values of types of Bidder One with those for types of Bidder Two to determine if a certain allocation is optimal. The proof is in Appendix~\ref{app:flow}.

\begin{proposition}\label{prop:Myerson}
For all $x,y$, the flow $(\alpha^M,\lambda^M)$ satisfies the following:
\begin{itemize}
\item For both $j$: $ k > k' \Rightarrow \Phi^{\alpha^M,\lambda^M}_1((v_1^k,j)) > \Phi^{\alpha^M,\lambda^M}_2((v_2^{k'},1))$.
\item For both $j$: $k < k' \Rightarrow \Phi^{\alpha^M,\lambda^M}_1((v_1^k,j)) < \Phi^{\alpha^M,\lambda^M}_2((v_2^{k'},1))$.
\item For all $k \in [n]$: if $x_{k} = 0 \text{ OR } y_{k} = 0$, then for both $j$: $\Phi^{\alpha^M,\lambda^M}_1((v_1^{k + 1},j)) > \Phi^{\alpha^M,\lambda^M}_2((v_2^{k + 1},1))$.
\item For all $k \in [n]$, if $x_k = 1 \text{ AND } y_k = 1$, then: $\Phi^{\alpha^M,\lambda^M}_1((v_1^{k + 1},1)) > \Phi^{\alpha^M,\lambda^M}_2((v_2^{k + 1},1)) > \Phi^{\alpha^M,\lambda^M}_1((v_1^{k + 1},2))$.
\item For both $j$: $\Phi^{\alpha^M,\lambda^M}_1((v_1^1,j)) > \Phi^{\alpha^M,\lambda^M}_2((v_2^1,1))$.
\item For both $j$: $\Phi^{\alpha^M,\lambda^M}_1((v_1^{n+2},j)) = \Phi^{\alpha^M,\lambda^M}_2((v_2^{n+2},1))$.
\item For both $i$, both $j$, and all $k \geq 1$: $\Phi_i^{\alpha^M,\lambda^M}((v_i^k,j))>0$.

\end{itemize}
\end{proposition}

The first two bullets assert that a bidder with strictly higher value has strictly higher virtual value as well. The next two bullets concern virtual values when both bidders' values are the same. Importantly, they assert that the relative comparison of virtual values when both bidders have value $n^2+k$ depends \emph{only on $x_k$ and $y_k$, and not on $x_{-k}$ or $y_{-k}$}. Just as importantly, they assert that Bidder One's interest is \emph{only relevant if $x_k = y_k = 1$}. Bullet seven implies that any allocation rule that witnesses optimality for $(\alpha^M,\lambda^M)$ must learn which bidder has higher virtual value. We analyze a potential such auction next.

\begin{definition}[Second-Price Auction, tie-breaking for Bidder One]\label{def:SPA} The \emph{second-price auction, tie-breaking for Bidder One}, gives the item to the bidder with highest value and breaks ties in favor of Bidder One. Payments are charged to satisfy the payment identity.
\end{definition}

\begin{theorem}\label{thm:disjoint} The second-price auction, tie-breaking for Bidder One, witnesses optimality for $(\alpha^M,\lambda^M)$ \emph{if and only if} $\disj(x,y) = \mathsf{yes}$.
\end{theorem}
\begin{proof}
First, it is well-known (and easy to see) that the second-price auction, tie-breaking for Bidder One, with the payment identity is BIC. After this, there are three bullets to check in Definition~\ref{def:witness}. We claim that the first two hold for all $x,y$, and the third holds if and only if $\disj(x,y) = \mathsf{yes}$.

The first bullet holds trivially, as payments are specifically defined to satisfy the payment identity.

The second bullet holds vacuously, as $\alpha^M_i(k) = 0$ for both $i$ and all $k$ (in fact, the implied condition holds anyway, as the bidders' allocation/payment doesn't depend on their interest).

To see that final bullet holds if and only if $\disj(x,y) = \mathsf{yes}$, observe that the second-price auction awards the item to the bidder with highest value, tie-breaking for Bidder One. So the final bullet holds if and only if: (a) a higher value implies a higher virtual value (which immediately follows from the first two bullets of Proposition~\ref{prop:Myerson}), (b) all virtual values are non-negative (which immediately follows from bullet seven of Proposition~\ref{prop:Myerson}), and (c) Bidder One has a higher virtual value whenever both bidders have the same value (which holds \emph{if and only if} $\disj(x,y) = \mathsf{yes}$, by bullets three and four of Proposition~\ref{prop:Myerson}).

This completes the proof: Definition~\ref{def:witness} is satisfied if and only if $\disj(x,y) = \mathsf{yes}$.
\end{proof}

Observe that Theorem~\ref{thm:disjoint} implies that the Second-Price Auction, tie-breaking for Bidder One is one optimal auction for $(D_1,D_2)$ when $\disj(x,y) = \mathsf{yes}$. We now conclude the following simple corollary:

\begin{corollary}\label{cor:disjoint} If $\disj(x,y) = \mathsf{yes}$, \emph{every} optimal BIC auction $(X,P)$ for $D_1,D_2$ has $X_1(t_1^1,t_2^1) = 1$.
\end{corollary}
\begin{proof}
Because a BIC auction witnesses optimality for $(\alpha^M, \lambda^M)$ (by Theorem~\ref{thm:disjoint}), every optimal BIC auction for $D_1,D_2$ witnesses optimality for $(\alpha^M,\lambda^M)$. Because $\Phi^{\alpha^M,\lambda^M}_1((v_1^1,1))>\Phi^{\alpha^M,\lambda^M}_2((v_2^1,1))$ by Proposition~\ref{prop:Myerson}, bullet three of Definition~\ref{def:witness} asserts that every optimal BIC auction satisfies $X_1(t_1^1,t_2^1) = 1$.
\end{proof}

Corollary~\ref{cor:disjoint} proves half of Theorem~\ref{thm:main}: that Bidder One wins the item in all optimal auctions on $(t_1^1,t_2^1)$ when $\disj(x,y) = \mathsf{yes}$. Theorem~\ref{thm:disjoint} makes clear the key distinction when $\disj(x,y) = \mathsf{no}$: $(\alpha^M,\lambda^M)$ does not witness optimality, so we need a new flow.

\subsection{Modifying the Canonical Flow}\label{sec:modifymain}
We now modify the canonical flow to find an optimal $(\alpha',\lambda')$ in the case when $\disj(x,y) = \mathsf{no}$. Fortunately, the necessary modification is simple to describe (although verifying the desired properties is comlpex). We will only make one modification, defined below, and first used in~\cite{DevanurW17}.

\begin{definition}[Boosting, \cite{DevanurW17}] Beginning with a flow $(\alpha,\lambda)$, \emph{boosting $(\alpha,\lambda)$ at $k$ for Bidder $i$ by $\varepsilon$}, for any $\varepsilon \leq \lambda_i^2(k')$ for all $k' \leq k$, produces a new flow $(\alpha',\lambda')$ with:
\begin{itemize}
\item  $\alpha'_i(k):=\alpha_i(k)+\varepsilon$.
\item  $(\lambda')_i^2(k'):=\lambda_i^2(k') - \varepsilon$, for all $k' \leq k$.
\item $(\lambda')_i^1(k'):=\lambda_i^1(k') +\varepsilon$, for all $k' \leq k$.
\item If not already specified, then $\alpha' = \alpha$ and $\lambda'=\lambda$.
\end{itemize}
\end{definition}

It is not hard to see that Boosting at $k$ preserves the flow conditions (provided that $\varepsilon \leq \lambda_i^2(k')$ for all $k' \leq k$). It is also not hard to see that Boosting at $k$ for Bidder $i$ \emph{increases} the virtual value for all $(v_i^{k'},2)$ for all $k' <k$, \emph{decreases} the virtual value for all $(v_i^{k'},1)$ for all $k' < k$, and leaves all other virtual values unchanged (see~\cite[Observation 3]{DevanurW17} --- although we will prove this ourselves whenever this is used in calculations).

\begin{definition}[Modified Flow] The modified flow $(\alpha^*,\lambda^*)$ proceeds as follows:
\begin{enumerate}
\item Begin with $(\alpha,\lambda) =(\alpha^M,\lambda^M)$.
\item Boost $(\alpha,\lambda)$ at $n+2$ for Bidder One by $\varepsilon$. Here, $\varepsilon$ is the minimum boost which results in $\Phi_1^{\alpha',\lambda'}((v_1^k,2))\geq \Phi_2^{\alpha',\lambda'}((v_2^k,1))$ for all $k$.
\end{enumerate}
\end{definition}

The rest of our analysis proceeds as follows. First, we need to establish that this modified flow indeed exists, because the required boost for Bullet 2 is small enough to be valid. Proposition~\ref{prop:valid} states this, and also several useful properties of this flow. The proof of Proposition~\ref{prop:valid} is in Appendix~\ref{app:flow}, and this relies on many of the precise choices in defining our instance.

\begin{proposition}\label{prop:valid}
For all $x,y$, $(\alpha^*,\lambda^*)$ is a valid flow. Moreover, it satisfies the following properties:
\begin{itemize}
\item For both $j$: $k > k' \Rightarrow \Phi_1^{\alpha^*,\lambda^*}((v_1^k,j))> \Phi_2^{\alpha^*,\lambda^*}((v_2^{k'},1))$.
\item For both $j$: $k < k' \Rightarrow \Phi_1^{\alpha^*,\lambda^*}((v_1^k,j))< \Phi_2^{\alpha^*,\lambda^*}((v_2^{k'},1))$.
\item For both $j$, and all $k \geq 2$, $\Phi_1^{\alpha^*,\lambda^*}((v_1^k,j)) \geq \Phi_2^{\alpha^*,\lambda^*}((v_2^k,1))$.
\item When $\disj(x,y) = \mathsf{no}$, there exists a $k^* \in [2,n+1]$ such that: $\Phi_1^{\alpha^*,\lambda^*}((v_1^{k^*},2)) = \Phi_2^{\alpha^*,\lambda^*}((v_2^{k^*},1))$.
\item When $\disj(x,y) = \mathsf{no}$, then: $\Phi^{\alpha^*,\lambda^*}_1((v_1^{1},2)) > \Phi^{\alpha^*,\lambda^*}_2((v_2^{1},1)) > \Phi^{\alpha^*,\lambda^*}_1((v_1^{1},1))$.
\item For both $j$: $\Phi^{\alpha^*,\lambda^*}_1((v_1^{n+2},j)) = \Phi^{\alpha^*,\lambda^*}_2((v_2^{n+2},1))$.
\item For both $i$, both $j$, and all $k \geq 1$: $\Phi_i^{\alpha^*,\lambda^*}((v_i^k,j))>0$.
\end{itemize}
\end{proposition}

We now define an auction that witnesses optimality for $(\alpha^*,\lambda^*)$, and conclude implications for $X_1(t_1^1,t_2^1)$.

\begin{definition}[Second-Price Auction, careful tie-breaking at $k^*$]\label{def:careful} The second-price auction with careful tie-breaking at $k^* \in [2,n+1]$ gives the item to the bidder with highest value. If both bidders have the same value $n^2+k$, break ties in the following manner (in all cases, charge payments satisfying the payment identity):
\begin{itemize}
\item If $k \neq 1$, and Bidder One's interest is day1, give the item to Bidder One.
\item If $k = 1$, and Bidder One's interest is day1, give the item to Bidder Two.
\item If $k \neq k^*$, and Bidder One's interest is day2, give the item to Bidder One.
\item If $k = k^*$, and Bidder One's interest is day2, give Bidder One the item with probability $1-\frac{f_2((v_2^1,1))}{f_2((v_2^{k^*},1))}$, and to Bidder Two with probability $\frac{f_2((v_2^1,1))}{f_2((v_2^{k^*},1))}$.\footnote{Observe that this is feasible, as we've guaranteed in our construction that $f_2((v_2^1,1)) < f_2((v_2^k,1))$ for all $k > 1$.}
\end{itemize}
\end{definition}

Let us quickly get some intuition for the Second-Price Auction with careful tie-breaking at $k^*$. First, observe that when Bidder One's value is neither $n^2+1$ nor $n^2+k^*$, the allocation rule is agnostic to Bidder One's interest. However, when Bidder One's value is $n^2+1$, ties are more often broken in favor of Bidder One when their interest is day2 versus day1. Similarly, when their value is $n^2+k^*$, ties are more often broken in favor of Bidder One when their interest is day1 versus day2. When calculating the payment identity, this implies that \emph{no matter Bidder One's value, their payment depends on their interest, even when their allocation probability does not}. In particular, the Second-Price Auction with careful tie-breaking at $k^*$ is not DSIC, and the precise probability chosen in bullet four is chosen exactly so that Lemma~\ref{lem:spaBIC} (below) holds.\footnote{For example, if Bidder One wins ties in bullet four with any probability $> 1-\frac{f_2((v_2^1,1))}{f_2((v_2^{k^*},1))}$, the auction would remain BIC (but the `Moreover,\ldots' property in Lemma~\ref{lem:spaBIC} would not hold). If Bidder One wins ties with probability $< 1-\frac{f_2((v_2^1,1))}{f_2((v_2^{k^*},1))}$, then Bidder One would have strict incentive to misreport their day2 interest as day1 whenever their value is $> n^2+k^*$.}

\begin{lemma}\label{lem:spaBIC} For all $x,y$ such that $\disj(x,y) = \mathsf{no}$, the Second-Price Auction with careful tie-breaking at $k^*$ is BIC. Moreover, $\pi_i((v_i^k,2)) \cdot v^{k}_i - p_i((v_i^k,2)) = \pi_i((v_i^k,1)) \cdot v^{k}_i - p_i((v_i^k,1))$ for all $k > k^*$.
\end{lemma}

With Lemma~\ref{lem:spaBIC} in hand, the proof of Theorem~\ref{thm:notdisjoint} follows similarly to that of Theorem~\ref{thm:disjoint}.

\begin{theorem}\label{thm:notdisjoint} When $\disj(x,y) = \mathsf{no}$, let $k^*$ be the index promised by bullet four of Proposition~\ref{prop:valid}. Then the second-price auction with careful tie-breaking at $k^*$ witnesses optimality for $(\alpha^*,\lambda^*)$.
\end{theorem}
\begin{proof}
We have already established in Lemma~\ref{lem:spaBIC} that the second-price auction with careful tie-breaking at $k^*$ is BIC, so we just need to check the three bullets. We have also explicitly defined payments to satisfy the payment identity, so bullet one is satisfied. For bullet two, the condition is vacuously satisfied for all $k < n+2$ because $\alpha_1(k) = 0$. At $k = n+2$, we are guaranteed that $\pi_i((v_i^{n+2},2)) \cdot v^{n+2}_i - p_i((v_i^{n+2},2)) = \pi_i((v_i^{n+2},1)) \cdot v^{n+2}_i - p_i((v_i^{n+2},1))$ by Lemma~\ref{lem:spaBIC}, as $k^* < n+2$. Therefore, bullet two is satisfied for all $k$.

Finally, we just need to confirm bullet three: that the auction always awards the item to a bidder with highest non-negative virtual value. Indeed, bullets one, two, and seven of Proposition~\ref{prop:valid} imply that the bidder with highest value also has the highest non-negative virtual value, so the second-price auction with careful tie-breaking at $k^*$ is correct whenever the two bidders have different values. Bullet three confirms that Bidder One's virtual value is always weakly higher than Bidder Two's in case they have the same value $> n^2+1$. This implies that the second-price auction with careful tie-breaking at $k^*$ breaks ties correctly in all cases when it gives the item to Bidder One. When both bidders have value $n^2+1$ and Bidder One's interest is day1, bullet five confirms that Bidder Two has higher virtual value (and the auction gives the item to Bidder Two). When both bidders have value $n^2+k^*$ and Bidder One's interest is $k^*$, bullet four confirms that the bidders have the same virtual value, so ties can be broken arbitrarily (and in particular, the randomization proposed is guaranteed to award the item to a bidder of highest virtual value).

This confirms all three bullets, and the proof.
\end{proof}

Again observe that Theorem~\ref{thm:notdisjoint} implies that the Second-Price Auction with careful tie-breaking at $k^*$ is one optimal auction for $(D_1,D_2)$ when $\disj(x,y) = \mathsf{no}$. We again conclude the following corollary:

\begin{corollary}\label{cor:notdisjoint} If $\disj(x,y)=\mathsf{no}$, every optimal BIC auction $(X,P)$ for $D_1,D_2$ has $X_1(t_1^1,t_2^1) = 0$.
\end{corollary}
\begin{proof}
Because a BIC auction witnesses optimality for $(\alpha^*,\lambda^*)$ (by Theorem~\ref{thm:notdisjoint}), every optimal BIC auction for $D_1,D_2$ witnesses optimality for $(\alpha^*,\lambda^*)$. Because $\Phi_1^{\alpha^*,\lambda^*}((v_1^1,1)) < \Phi_2^{\alpha^*,\lambda^*}((v_2^1,1))$ by Proposition~\ref{prop:valid}, bullet three of Definition~\ref{def:witness} asserts that every optimal BIC auction satisfies $X_1(t_1^1,t_1^2) = 0$.
\end{proof}

This wraps up the proof of Theorem~\ref{thm:main}.
\begin{proof}[Proof of Theorem~\ref{thm:main}]
Corollary~\ref{cor:disjoint} establishes that when $\disj(x,y) = \mathsf{yes}$, any optimal BIC auction must allocate the item to Bidder One on $(t_1^1,t_2^1)$ with probability one. Corollary~\ref{cor:notdisjoint} establishes that when $\disj(x,y) = \mathsf{no}$, any optimal BIC auction must allocate the item to Bidder One on $(t_1^1,t_2^1)$ with probability zero. Because $D_1$ can be constructed only as a function of $x$, and $D_2$ can be constructed only as a function of $y$, any communication protocol which correctly allocates the item on $(t_1^1,t_2^1)$ in accordance with \emph{any} optimal BIC mechanism (even with probability $2/3$) can also solve \disj\ (with probability $2/3$). Because any deterministic (resp. randomized, succeeding with probability $2/3$) protocol for disjointness requires communication $n$ (resp. $\Omega(n)$), this means that any deterministic (resp. randomized, succeeding with probability $2/3$) protocol which can correctly allocate the item on $(t_1^1,t_2^1)$ in accordance with any optimal BIC mechanism (resp. with probability $2/3$) requires communication at least $n$ (resp. $\Omega(n)$).
\end{proof}

\section{Conclusion} \label{sec:conclusion}
We establish that optimal multi-dimensional mechanisms are not locally-implementable: in order to evaluate the auction on just a single valuation profile, one must know (essentially) the entire distribution. In contrast, optimal single-dimensional mechanisms are locally-implementable: evaluating the auction on a single valuation profile requires barely more bits from each $D_i$ than simply stating $v_i$ itself. Our construction establishes that this separation holds already in (essentially) the simplest possible multi-dimensional setting: one single-dimensional bidder and one two-day FedEx bidder. We also show that optimal auctions for single-dimensional buyers with public budget constraints are not locally-implementable. Moreover, both results follow the same outline, highlighting the robustness of our techniques.

Our work establishes a novel complexity of optimal multi-dimensional mechanisms distinct from optimal single-dimensional mechanisms. In particular, unlike prior work, this complexity is inherently a multi-bidder phenomenon, rather than inherited from the single-bidder setting. Indeed, every optimal single-bidder auction for any instance considered by our reductions simply sets a price of $n^2+1$. Locality can serve as a quantitative lens for future work to study the complexity of multi-bidder auctions in multi-dimensional settings where single-bidder auctions are tractable. For example:
\begin{itemize}
\item Do there exist approximately-optimal multi-dimensional auctions that are locally-implementable? One significant technical barrier to this direction is an alternative line of attack beyond complementary slackness (as complementary slackness holds only for optimal primal/dual pairs). Note that the Marginal Revenue Mechanism of~\cite{AlaeiFHH13} is locally-implementable and approximately optimal in restricted ``approximately revenue-linear'' settings. But, it remains unknown whether approximately-optimal locally-implementable mechanisms exist generally.
\item What are the implications of (non)-locality for streaming or online-learning variants of optimal auction design? In this direction, it is important that we study locality via communication complexity, due to the strong connection between communication complexity and streaming lower bounds~\cite{AlonMS99, RaoY20}.
\end{itemize}


\newcommand{\etalchar}[1]{$^{#1}$}

\appendix
\section{Optimal Single-Dimensional Mechanisms are Local}\label{app:single}

In this section we are going to show that optimal single-item auctions are local. Our analysis here adopts the characterization by~\cite{CaiDW16}.

We assume each  $D_i$ is a distribution supported on $n$ valuations ($v_i^1, v_i^2, \ldots, v_i^n$), and each valuation has an integer value between $0$ and $p(n)$ ($p$ is a polynomial in $n$) for each outcome, and the probability of each valuation is an integer multiple of $1/g(n)$ (g is another polynomial in $n$). We use the notation $R_i(v_i^k):=\sum_{k' \geq k} f_i(v_i^{k'})$.\footnote{For simplicity of notation, denote by $R_i(v_i^{n+1}) := 0$.}

First we define the discrete Myerson virtual value.
\begin{definition}[Single-dimensional Virtual Value]
  For a discrete distribution $D_i$, define
  \[ \varphi^{D_i}(v_i^k):= v_i^k - \frac{(v_i^{k + 1} - v_i^k) \cdot R_i(v_i^{k+1})}{f_i(v_i^k)}.\]
\end{definition}

We now show that a single virtual value $\varphi^{D_i}(v_i^k)$ contains barely more information from $D_i$ than the value $v_i^k$ itself (which has zero information from $D_i$).

\begin{lemma}
  For fixed polynomials $p$ and $g$, given distribution $D_i$ and $k$, $\varphi^{D_i}(v_i^k)$ can be represented in $O(\log n)$ bits.
\end{lemma}

\begin{proof}
  We can use $v_i^k$, $v_i^{k + 1}$, $R_i(v_i^{k+1})$, and  $f_i(v_i^k)$ to compute $\varphi^{D_i}(v_i^k)$, thus we can trivially encode each of them into $O(\log n)$ bits separately.
\end{proof}

We will then prove a bit stronger result to show the weighted average virtual value over an interval can also be represented in $O(\log n)$ bits.
\begin{lemma} \label{lem:averagevv}
  For fixed polynomials $p$ and $g$, given distribution $D_i$ and $k, l$, the weighted average virtual value over an $[k, l]$ \[\frac{\sum_{j=k}^l f_i(v_i^j) \cdot \varphi^{D_i}(v_i^j)} {\sum_{j=k}^l f_i(v_i^j) }\] can be represented in $O(\log n)$ bits.
\end{lemma}

\begin{proof}
  Since the denominator $\sum_{j=k}^l f_i(v_i^j)$ is an integer multiple of $1/g(n)$ and less than or equal to $1$, it can be easily encoded in $O(\log n)$ bits. So it is sufficient to show that the numerator $\sum_{j=k}^l f_i(v_i^j) \cdot \varphi^{D_i}(v_i^j)$ can be represented in $O(\log n)$ bits. To achieve this, observe that:
  \begin{align*}
    \sum_{j=k}^l f_i(v_i^j) \cdot \varphi^{D_i}(v_i^j) & = \sum_{j=k}^l \left ( f_i(v_i^j) \cdot v_i^j - (v_i^{k + 1} - v_i^k) \cdot R_i(v_i^{k+1}) \right ) \\
    & = \sum_{j=k}^l  f_i(v_i^j) \cdot v_i^j - \left( v_i^{l + 1} \cdot R_i(v_i^{l + 1}) - v_i^{k} \cdot R_i(v_i^{k + 1})  + \sum_{j=k+1}^{l}  f_i(v_i^j) \cdot v_i^{j} \right) \\
    & = f_i(v_i^k) \cdot v_i^k - v_i^{l + 1} \cdot R_i(v_i^{l + 1}) + v_i^{k} \cdot R_i(v_i^{k + 1}) \\
    & = v_i^{k} \cdot R_i(v_i^{k}) - v_i^{l + 1} \cdot R_i(v_i^{l + 1}).
  \end{align*}
   Note that $v_i^{k}$, $R_i(v_i^{k})$, $ v_i^{l + 1}$, and $R_i(v_i^{l + 1})$ can all be encoded in $O(\log n)$ bits.
\end{proof}

The following theorem characterizes the optimal mechanism in this single-dimensional setting.

\begin{theorem}[\cite{Myerson81, CaiDW16}] \label{thm:single}
    The revenue-optimal BIC mechanism awards the item to the bidder with the highest non-negative ironed virtual value (if one exists), breaking ties arbitrarily but consistently across inputs, where ironed virtual value for type $v_i^k$: $\bar{\varphi}^{D_i}_i(v_i^k)$ is the weighted average virtual value over an interval containing $v_i^k$.\footnote{See~\cite{CaiDW16} for explicit expressions of ironed virtual values. Here we only need to know that an ironed virtual value is the weighted average over an interval of virtual values.}  If no such bidder exists, the item remains unallocated.
\end{theorem}

Finally with Lemma~\ref{lem:averagevv} and Theorem~\ref{thm:single} we conclude the following theorem, which implies that the optimal single-dimensional auction is local. 

\begin{theorem}
    When each $D_i$ is single-dimensional, supported on $n$ valuations, and each valuation has an integer value between $0$ and $\poly(n)$ for each outcome, and the probability of each valuation is an integer multiple of $1/\poly(n)$, then $O(\log n)$ bits from each $D_i$ suffice to compute $\opt(\vec{v},\vec{D})$.
\end{theorem}
\section{Omitted Proofs from Section~\ref{sec:reduction}}\label{app:reduction}

In order to distinguish different helpers defined in the constructions and improve readability, we will use the following notations throughout Appendices: 
\begin{itemize}
  \item We use $c^k$ to represent type $(v_1^k, 1)$. Notation $z^c_k$ represents the helper $z_k$ for Bidder One on day1.  
  \item We use $d^k$ to represent type $(v_1^k, 2)$. Notation $z^d_k$ represents the helper $z_k$ for Bidder One on day2. 
  \item We use $e^k$ to represent type $(v_2^k, 1)$. Notation $z^e_k$ represents the helper $z_k$ for Bidder Two on day1. 
\end{itemize}
Also we will use $f(\cdot)$ instead of $f_i(\cdot)$, and $\Phi(\cdot)$ instead of $\Phi_i(\cdot)$ in all proofs when there is no ambiguity.

\subsection{Bidder 1 day 1}

Lemmas below provide some useful properties of the helpers $z^c$s, at the end of this subsection, we will use these to prove the distribution we constructed is valid and nearly-uniform.

The following two lemmas bound the gap between two consecutive $z^cs$. 

\begin{lemma} \label{zcgap}
  $z^c_{i + 1} - z^c_{i + 2} \le \frac{n^3}{(n - i + 2)(n - i + 1)}$ for all $i \in [1, n]$.
\end{lemma}

\begin{proof}
\begin{align*}
        z^c_{i + 2} & = \frac{b - \sum_{j = 1}^{i + 1} f(c^j) \cdot 2b} {n - i + 1} \\
        & = \frac{b - f(c^{i + 1})\cdot 2b - \sum_{j = 1}^{i} f(c^j) \cdot 2b} {n - i + 1} \\
        & = \frac{z^c_{i + 1} \cdot (n - i + 2) - f(c^{i+1}) \cdot 2b}{n - i + 1} \\
        & = z^c_{i+1} + \frac{z^c_{i+1} - f(c^{i+1}) \cdot 2b}{n - i + 1} \\
        & = z^c_{i+1} + \frac{z^c_{i+1} -  \left \lfloor  z^c_{i + 1} + \frac{n^3}{n - i + 2}   \right \rfloor}{n - i + 1} \\
        & \ge z^c_{i+1} + \frac{z^c_{i+1} -  \left (  z^c_{i + 1} + \frac{n^3}{n - i + 2}  \right )}{n - i + 1} \\
        & = z^c_{i+1} - \frac{n^3}{(n - i + 2)(n - i + 1)}.
\end{align*}
\end{proof}

\begin{lemma}\label{zcmonotone}
$z^c_{i + 2} < z^c_{i + 1}$ for all $i \in [1, n]$.
\end{lemma}
\begin{proof}
By definition, we have
\begin{align*}
        z^c_{i + 2} & = \frac{b - \sum_{j = 1}^{i + 1} f(c^j) \cdot 2b} {n - i + 1} \\
        & = \frac{b - f(c^{i + 1})\cdot 2b - \sum_{j = 1}^{i} f(c^j) \cdot 2b} {n - i + 1} \\
        & = \frac{z^c_{i + 1} \cdot (n - i + 2) - f(c^{i+1}) \cdot 2b}{n - i + 1} \\
        & = z^c_{i+1} + \frac{z^c_{i+1} - f(c^{i+1}) \cdot 2b}{n - i + 1} \\
        & = z^c_{i+1} + \frac{z^c_{i+1} -  \left \lfloor  z^c_{i + 1} + \frac{n^3}{n - i + 2}   \right \rfloor}{n - i + 1} \\
        & < z^c_{i+1} + \frac{z^c_{i+1} -  \left (  z^c_{i + 1} + \frac{n^3}{n - i + 2} - 1  \right )}{n - i + 1} \\
        & = z^c_{i+1} - \frac {\frac{n^3}{n - i + 2} - 1 }{n - i + 1} \\
        & < z^c_{i+1}.
\end{align*}
\end{proof}

The following lemma bounds the range of $z^c$.

\begin{lemma}\label{zcrange}
$z^c_{i+1} \in [a - n^3, a]$, for all $i \in [1, n + 1]$.
\end{lemma}
\begin{proof}
 Since by Lemma \ref{zcmonotone} $z^c$ is monotone decreasing, and by definition $z^c_2 = a$, we only need to bound $z^c_{n+2}$. Again by the fact that $z^c_2 = a$ we have
\begin{align*}
    z^c_{n + 2} & = z^c_2 + \sum_{i = 1}^{n} (z^c_{i + 2} - z^c_{i+1})  \\
    & = z^c_2 +  \sum_{i = 1}^{n} \frac{z^c_{i+1} \cdot (n-i+2) - f(c^{i+1}) \cdot 2b}{n - i + 1} - z^c_{i+1} \\
    & = z^c_2 +  \sum_{i = 1}^{n} \frac{z^c_{i+1} - f(c^{i+1}) \cdot 2b}{n - i + 1} \\
    & = z^c_2 + \sum_{i = 1}^n \frac{z^c_{i + 1} - \left \lfloor  z^c_{i + 1} + \frac{n^3}{n - i + 2}   \right \rfloor}{n-i+1} \\
    & \ge z^c_2 + \sum_{i = 1}^n \frac{z^c_{i + 1} - \left ( z^c_{i + 1} + \frac{n^3}{n - i + 2}   \right )}{n-i+1} \\
    & = z^c_2 - \sum_{i = 1}^n \frac{n^3} {(n-i+1)(n-i+2)} \\
    & = z^c_2 -  n^3\sum_{i = 1}^n\frac1 {(n-i+1)(n-i+2)} \\
    &  = z^c_2 - n^3 \left (\frac{1}{(n + 1)n} + \frac{1}{n(n - 1)} + \ldots + \frac{1}{1\cdot2} \right) \\
    &  = z^c_2 - n^3 \left [ \left (\frac{1}{n} - \frac{1}{n+1} \right) + \left (\frac{1}{n-1} - \frac{1}{n} \right) + \ldots +  \left(1 - \frac12\right) \right] \\
    & = z^c_2 - n^3 \left(1-\frac1{n+1} \right) \\
    & > a - n^3,
\end{align*}
which concludes the proof.
\end{proof}

The following proof proves Lemma~\ref{lem:subprobs11}, which shows our distribution is valid.

\begin{proof}[Proof of Lemma~\ref{lem:subprobs11}]
\begin{align*}
\sum_{k=1}^{n+2} f(c^k) &= f(c^{n+2}) + \sum_{k=1}^{n+1} f(c^k)  \\
&= \frac{b - \sum_{k=1}^{n+1}  f(c^k) \cdot 2b}{2b} + \sum_{k=1}^{n+1} f(c^k) \\
&= 1/2 - \sum_{k=1}^{n+1} f(c^k) + \sum_{k=1}^{n+1} f(c^k)\\
&=1/2.
\end{align*}
\end{proof}

Finally, utilizing the properties above, we conclude this subsection by proving a key property of our construction. It guarantees that Bidder One's day1 distribution is nearly-uniform over $v_1^2,\ldots, v_1^{n+2}$.

\begin{proof}[Proof of Lemma~\ref{crange}]
By the definition of $f(c^{i + 1})$ and Lemma \ref{zcrange} and \ref{zcmonotone}  we have
\begin{align*}
    f(c^{i+1}) \cdot 2b & = \left \lfloor  z^c_{i + 1} + \frac{n^3}{n - i + 2}   \right \rfloor \\
    & > z^c_{i + 1} + \frac{n^3}{n - i + 2} - 1 \\
    & > z^c_{i + 1} - 1 \\
    & > a - n^3 - 1.
\end{align*}
On the other hand, we have
\begin{align*}
        f(c^{i+1}) \cdot 2b & = \left \lfloor  z^c_{i + 1} + \frac{n^3}{n - i + 2}   \right \rfloor \\
    & \le z^c_{i + 1} + \frac{n^3}{n - i + 2} \\
    & \le z^c_{i + 1} + n^3 \\
    & \le a + n^3,
\end{align*}
which concludes the lemma.

\end{proof}

\subsection{Bidder 1 Day 2}

Lemmas below provide some useful properties of the helpers $z^d$s, at the end of this subsection, we will use these to prove the distribution we constructed is valid and nearly-uniform.

The following lemma shows the relationship between $f(d^{i+1})$ and $z^d_{i+2}$.
\begin{lemma} \label{fzd}
 $f(d^{i + 1}) \cdot 2b \ge z^d_{i + 2}$ for all $i \in [1, n]$.
\end{lemma}
\begin{proof}
   If $x_i = 1$, recall that $f(d^{i+1}) = \left \lceil z^d_{i+1} \right \rceil /2b$, in the other case, by definition we know $f(d^{i+1}) = \left \lfloor z^d_{i+1} + \frac{n^3}{n - i + 2} \right \rfloor /2b$. In both cases it is clear that $f(d^{i + 1}) \cdot 2b \ge z^d_{i+1}$, thus we have
  \begin{align*}
        z^d_{i + 2} & = \frac{b - \sum_{j = 1}^{i + 1} f(d^j) \cdot 2b} {n - i + 1} \\
        & = \frac{b - f(d^{i + 1})\cdot 2b - \sum_{j = 1}^{i} f(d^j) \cdot 2b} {n - i + 1} \\
        & = \frac{z^d_{i + 1} \cdot (n - i + 2) - f(d^{i+1}) \cdot 2b}{n - i + 1} \\
        & \le \frac{f(d^{i + 1}) \cdot 2b \cdot (n - i + 2) - f(d^{i+1}) \cdot 2b}{n - i + 1} \\
        & = f(d^{i + 1}) \cdot 2b.
  \end{align*}
\end{proof}

The following lemmas bound the gap between two consecutive $z^ds$.

\begin{lemma} \label{zd1gap}
  If $x_i = 0$, $z^d_{i + 1} - z^d_{i + 2} \le \frac{n^3}{(n - i + 2)(n - i + 1)}$.
\end{lemma}

\begin{proof}
\begin{align*}
        z^d_{i + 2} & = \frac{b - \sum_{j = 1}^{i + 1} f(d^j) \cdot 2b} {n - i + 1} \\
        & = \frac{b - f(d^{i + 1})\cdot 2b - \sum_{j = 1}^{i} f(d^j) \cdot 2b} {n - i + 1} \\
        & = \frac{z^d_{i + 1} \cdot (n - i + 2) - f(d^{i+1}) \cdot 2b}{n - i + 1} \\
        & = z^d_{i+1} + \frac{z^d_{i+1} - f(d^{i+1}) \cdot 2b}{n - i + 1} \\
        & = z^d_{i+1} + \frac{z^d_{i+1} -  \left \lfloor  z^d_{i + 1} + \frac{n^3}{n - i + 2}   \right \rfloor}{n - i + 1} \\
        & \ge z^d_{i+1} + \frac{z^d_{i+1} -  \left (  z^d_{i + 1} + \frac{n^3}{n - i + 2}  \right )}{n - i + 1} \\
        & = z^d_{i+1} - \frac{n^3}{(n - i + 2)(n - i + 1)}\\
\end{align*}
\end{proof}

\begin{lemma} \label{zd0gap}
  If $x_i = 1$, $z^d_{i + 1} - z^d_{i + 2} \le \frac{1}{n - i + 1}$.
\end{lemma}

\begin{proof}

\begin{align*}
        z^d_{i + 2} & = \frac{b - \sum_{j = 1}^{i + 1} f(d^j) \cdot 2b} {n - i + 1} \\
        & = \frac{b - f(d^{i + 1})\cdot 2b - \sum_{j = 1}^{i} f(d^j) \cdot 2b} {n - i + 1} \\
        & = \frac{z^d_{i + 1} \cdot (n - i + 2) - f(d^{i+1}) \cdot 2b}{n - i + 1} \\
        & = z^d_{i+1} + \frac{z^d_{i+1} - f(d^{i+1}) \cdot 2b}{n - i + 1} \\
        & = z^d_{i+1} + \frac{z^d_{i+1} -  \left \lceil z^d_{i + 1}   \right \rceil}{n - i + 1} \\
        & \ge z^d_{i+1} + \frac{z^d_{i+1} -  \left (  z^d_{i + 1} + 1  \right )}{n - i + 1} \\
        & = z^d_{i+1} - \frac{1}{n - i + 1}\\
\end{align*}
\end{proof}

\begin{lemma} \label{zdmonotone}
$z^d_{i + 2} < z^d_{i + 1}$ for all $i \in [1, n]$.
\end{lemma}

\begin{proof}
  If $x_i = 1$, then we have
\begin{align*}
        z^d_{i + 2} & = \frac{b - \sum_{j = 1}^{i + 1} f(d^j) \cdot 2b} {n - i + 1} \\
        & = \frac{b - f(d^{i + 1})\cdot 2b - \sum_{j = 1}^{i} f(d^j) \cdot 2b} {n - i + 1} \\
        & = \frac{z^d_{i + 1} \cdot (n - i + 2) - f(d^{i+1}) \cdot 2b}{n - i + 1} \\
        & = z^d_{i+1} + \frac{z^d_{i+1} - f(d^{i+1}) \cdot 2b}{n - i + 1} \\
        & = z^d_{i+1} + \frac{z^d_{i+1} -  \left \lfloor  z^d_{i + 1} + \frac{n^3}{n - i + 2}   \right \rfloor}{n - i + 1} \\
        & < z^d_{i+1} + \frac{z^d_{i+1} -  \left (  z^d_{i + 1} + \frac{n^3}{n - i + 2} - 1  \right )}{n - i + 1} \\
        & = z^d_{i+1} - \frac {\frac{n^3}{n - i + 2} - 1 }{n - i + 1} \\
        & < z^d_{i+1}.
\end{align*}
Otherwise, we have
\begin{align*}
        z^d_{i + 2} & = \frac{b - \sum_{j = 1}^{i + 1} f(d^j) \cdot 2b} {n - i + 1} \\
        & = \frac{b - f(d^{i + 1})\cdot 2b - \sum_{j = 1}^{i} f(d^j) \cdot 2b} {n - i + 1} \\
        & = \frac{z^d_{i + 1} \cdot (n - i + 2) - f(d^{i+1}) \cdot 2b}{n - i + 1} \\
        & = z^d_{i+1} + \frac{z^d_{i+1} - f(d^{i+1}) \cdot 2b}{n - i + 1} \\
        & = z^d_{i+1} + \frac{z^d_{i+1} -  \left \lceil  z^d_{i + 1}  \right \rceil}{n - i + 1} \\
        & < z^d_{i+1} + \frac{z^d_{i+1} -  z^d_{i + 1}  }{n - i + 1} \\
        & = z^d_{i+1},
\end{align*}
which concludes the proof.
\end{proof}

The following lemma bounds the range of $z^d$.
\begin{lemma}\label{zdrange}
$z^d_{i+1} \in [a - n^3, a]$, for all $i \in [1, n + 1]$.
\end{lemma}
\begin{proof}
 Since by Lemma \ref{zdmonotone} $z^d$ is monotone decreasing, and by definition $z^d_2 = a$, we only need to bound $z^d_{n+2}$:
\begin{align*}
    z^d_{n + 2} & = z^d_2 + \sum_{i = 1}^{n} (z^d_{i + 2} - z^d_{i+1})  \\
    & = z^d_2 +  \sum_{i = 1}^{n} \frac{z^d_{i+1} \cdot (n-i+2) - f(d^{i+1}) \cdot 2b}{n - i + 1} - z^d_{i+1} \\
    & = z^d_2 +  \sum_{i = 1}^{n} \frac{z^d_{i+1} - f(d^{i+1}) \cdot 2b}{n - i + 1} \\
    & = z^d_2 + \sum_{i = 1}^n \frac{z^d_{i + 1} - \left \lfloor  z^d_{i + 1} + \frac{n^3}{n - i + 2}   \right \rfloor \cdot [x_i = 0] - \lceil z^d_{i+1}\rceil \cdot [x_i = 1]}{n-i+1}   \\
    & \ge z^d_2 + \sum_{i = 1}^n \frac{z^d_{i + 1} - \left \lfloor  z^d_{i + 1} + \frac{n^3}{n - i + 2}   \right \rfloor }{n-i+1}   \\
    & \ge z^d_2 + \sum_{i = 1}^n \frac{z^d_{i + 1} - \left ( z^d_{i + 1} + \frac{n^3}{n - i + 2}   \right )}{n-i+1} \\
    & = z^d_2 - \sum_{i = 1}^n \frac{n^3} {(n-i+1)(n-i+2)} \\
    & = z^d_2 -  n^3\sum_{i = 1}^n\frac1 {(n-i+1)(n-i+2)} \\
    &  = z^d_2 - n^3 \left (\frac{1}{(n + 1)n} + \frac{1}{n(n - 1)} + \ldots + \frac{1}{1\cdot2} \right) \\
    &  = z^d_2 - n^3 \left [ \left (\frac{1}{n} - \frac{1}{n+1} \right) + \left (\frac{1}{n-1} - \frac{1}{n} \right) + \ldots +  \left(1 - \frac12\right) \right] \\
    & = z^d_2 - n^3 \left(1-\frac1{n+1} \right) \\
    & > a - n^3,
\end{align*}
where the last inequality comes from the fact that  $z^d_2 = a$.
\end{proof}

The following proof proves Lemma~\ref{lem:subprobs12}, which shows our distribution is valid.

\begin{proof}[Proof of Lemma~\ref{lem:subprobs12}]
\begin{align*}
\sum_{k=1}^{n+2} f(d^k) &= f(d^{n+2}) + \sum_{k=1}^{n+1} f(d^k)  \\
&= \frac{b - \sum_{k=1}^{n+1}  f(d^k) \cdot 2b}{2b} + \sum_{k=1}^{n+1} f(d^k) \\
&= 1/2 - \sum_{k=1}^{n+1} f(d^k) + \sum_{k=1}^{n+1} f(d^k)\\
&=1/2.
\end{align*}
\end{proof}

Similarly, utilizing the properties above, we conclude this subsection by proving a key property of our construction. It guarantees that Bidder One's day2 distribution is nearly-uniform over $v_1^2,\ldots, v_1^{n+2}$.

\begin{proof}[Proof of Lemma~\ref{drange}]

By \ref{zdrange} we have
\begin{align*}
    f(d^{i+1}) \cdot 2b & = \left \lceil  z^d_{i + 1} \right \rceil \cdot [x_i = 1] + \left \lfloor  z^d_{i + 1} + \frac{n^3}{n - i + 2}   \right \rfloor \cdot [x_i = 0] \\
    & \ge \left \lceil  z^d_{i + 1} \right \rceil \\
    & \ge z^d_{i + 1} \\
    & \ge a - n^3.
\end{align*}
On the other hand,  we have
\begin{align*}
     f(d^{i+1}) \cdot 2b & = \left \lceil  z^d_{i + 1} \right \rceil \cdot [x_i = 1] + \left \lfloor  z^d_{i + 1} + \frac{n^3}{n - i + 2}   \right \rfloor \cdot [x_i = 0] \\
    & \le \left \lfloor  z^d_{i + 1} + \frac{n^3}{n - i + 2}   \right \rfloor \\
    & \le z^d_{i + 1} + \frac{n^3}{n - i + 2} \\
    & \le z^d_{i + 1} + n^3 \\
    & \le a + n^3,
\end{align*}
which concludes the proof.

\end{proof}

\subsection{Bidder 2}

Lemmas below provide some useful properties of the helpers $z^e$s, at the end of this subsection, we will use these to prove the distribution we constructed is valid and nearly-uniform.

The following lemma shows the relationship between $f(e^{i+1})$ and $z^e_{i+2}$ when $y_i = 0$.
\begin{lemma} \label{fze}
If $y_i = 0, f(e^{i + 1}) \cdot b < z^e_{i + 2}$ for all $i \in [1, n]$.
\end{lemma}
\begin{proof}
   Recall that $f(e^{i+1}) = \left \lfloor z^e_{i+1} - 1 \right \rfloor /b$ when $y_i = 0$, then we have
  \begin{align*}
        z^e_{i + 2} & = \frac{b - \sum_{j = 1}^{i + 1} f(e^j) \cdot b} {n - i + 1} \\
        & = \frac{b - f(e^{i + 1})\cdot b - \sum_{j = 1}^{i} f(e^j) \cdot b} {n - i + 1} \\
        & = \frac{z^e_{i + 1} \cdot (n - i + 2) - f(e^{i+1}) \cdot b}{n - i + 1} \\
        & > \frac{f(e^{i + 1}) \cdot b \cdot (n - i + 2) - f(e^{i+1}) \cdot b}{n - i + 1} \\
        & = f(e^{i + 1}) \cdot b
  \end{align*}
\end{proof}

The following lemmas bound the gap between two consecutive $z^es$.

\begin{lemma} \label{ze1gap}
  If $y_i = 1$, $z^e_{i + 1} - z^e_{i + 2} \le \frac{n^2}{(n - i + 2)(n - i + 1)}$.
\end{lemma}

\begin{proof}
\begin{align*}
        z^e_{i + 2} & = \frac{b - \sum_{j = 1}^{i + 1} f(e^j) \cdot b} {n - i + 1} \\
        & = \frac{b - f(e^{i + 1})\cdot b - \sum_{j = 1}^{i} f(e^j) \cdot b} {n - i + 1} \\
        & = \frac{z^e_{i + 1} \cdot (n - i + 2) - f(e^{i+1}) \cdot b}{n - i + 1} \\
        & = z^e_{i+1} + \frac{z^e_{i+1} - f(e^{i+1}) \cdot b}{n - i + 1} \\
        & = z^e_{i+1} + \frac{z^e_{i+1} -  \left \lfloor  z^e_{i + 1} + \frac{n^2}{n - i + 2}   \right \rfloor}{n - i + 1} \\
        & \ge z^e_{i+1} + \frac{z^e_{i+1} -  \left (  z^e_{i + 1} + \frac{n^2}{n - i + 2}  \right )}{n - i + 1} \\
        & = z^e_{i+1} - \frac{n^2}{(n - i + 2)(n - i + 1)}\\
\end{align*}
\end{proof}

\begin{lemma} \label{ze0gap}
  If $y_i = 0$, $z^e_{i + 1} - z^e_{i + 2} \le -\frac{1}{n-i+1}$.
\end{lemma}

\begin{proof}
\begin{align*}
        z^e_{i + 2} & = \frac{b - \sum_{j = 1}^{i + 1} f(e^j) \cdot b} {n - i + 1} \\
        & = \frac{b - f(e^{i + 1})\cdot b - \sum_{j = 1}^{i} f(e^j) \cdot b} {n - i + 1} \\
        & = \frac{z^e_{i + 1} \cdot (n - i + 2) - f(e^{i+1}) \cdot b}{n - i + 1} \\
        & = z^e_{i+1} + \frac{z^e_{i+1} - f(e^{i+1}) \cdot b}{n - i + 1} \\
        & = z^e_{i+1} + \frac{z^e_{i+1} -  \left \lfloor  z^e_{i + 1} - 1  \right \rfloor}{n - i + 1} \\
        & \ge z^e_{i+1} + \frac{z^e_{i+1} -  \left (  z^e_{i + 1}  - 1   \right )}{n - i + 1} \\
        & = z^e_{i+1} + \frac{1}{(n - i + 1)}
\end{align*}
\end{proof}

\begin{lemma} \label{zenearlymonotone}
$z^e_{i + 2} < z^e_{i + 1} + 2$ for all $i \in [1, n]$.

\end{lemma}

\begin{proof}
  If $y_i = 1$, then we have
\begin{align*}
        z^e_{i + 2} & = \frac{b - \sum_{j = 1}^{i + 1} f(e^j) \cdot b} {n - i + 1} \\
        & = \frac{b - f(e^{i + 1})\cdot b - \sum_{j = 1}^{i} f(e^j) \cdot b} {n - i + 1} \\
        & = \frac{z^e_{i + 1} \cdot (n - i + 2) - f(e^{i+1}) \cdot b}{n - i + 1} \\
        & = z^e_{i+1} + \frac{z^e_{i+1} - f(e^{i+1}) \cdot b}{n - i + 1} \\
        & = z^e_{i+1} + \frac{z^e_{i+1} -  \left \lfloor  z^e_{i + 1} + \frac{n^2}{n - i + 2}   \right \rfloor}{n - i + 1} \\
        & < z^e_{i+1} + \frac{z^e_{i+1} -  \left (  z^e_{i + 1} + \frac{n^2}{n - i + 2} - 1  \right )}{n - i + 1} \\
        & = z^e_{i+1} - \frac {\frac{n^2}{n - i + 2} - 1 }{n - i + 1} \\
        & < z^e_{i+1}.
\end{align*}
Otherwise, we have
\begin{align*}
        z^e_{i + 2} & = \frac{b - \sum_{j = 1}^{i + 1} f(e^j) \cdot b} {n - i + 1} \\
        & = \frac{b - f(e^{i + 1})\cdot b - \sum_{j = 1}^{i} f(e^j) \cdot b} {n - i + 1} \\
        & = \frac{z^e_{i + 1} \cdot (n - i + 2) - f(e^{i+1}) \cdot b}{n - i + 1} \\
        & = z^e_{i+1} + \frac{z^e_{i+1} - f(e^{i+1}) \cdot b}{n - i + 1} \\
        & = z^e_{i+1} + \frac{z^e_{i+1} -  \left \lfloor  z^e_{i + 1}  \right \rfloor + 1}{n - i + 1} \\
        & < z^e_{i+1} + \frac{z^e_{i+1} -  z^e_{i + 1} + 2  }{n - i + 1} \\
        & = z^e_{i+1} + \frac{2}{n - i + 1} \\
        & \le z^e_{i+1} + 2,
\end{align*}
which concludes the proof.
\end{proof}

\begin{lemma}\label{zerange}
$z^e_{i+1} \in [a - 2n^2, a + 2n^2]$, for all $i \in [1, n + 1]$.
\end{lemma}
\begin{proof}
  Since by Lemma \ref{zenearlymonotone}, $z^e_{i + 2} < z^e_{i + 1} + 2$, which indicates $z^e_{n+2} - 2n < z^e_{i + 1} < z^e_2 + 2n$ for all $i \in [1, n + 1]$, and by definition $z^e_2 = \frac{b - f(e^1) \cdot b}{n + 1} = \frac{b - \frac{b}{10n} + 1}{n + 1} = a + \frac{1}{n + 1}$, we only need to lower bound $z^e_{n+2}$. Again by the fact that $z^e_2 = a + \frac{1}{n + 1}$ we have
\begin{align*}
    z^e_{n + 2} & = z^e_2 + \sum_{i = 1}^{n} (z^e_{i + 2} - z^e_{i+1})  \\
    & = z^e_2 +  \sum_{i = 1}^{n} \frac{z^e_{i+1} \cdot (n-i+2) - f(e^{i+1}) \cdot b}{n - i + 1} - z^e_{i+1} \\
    & = z^e_2 +  \sum_{i = 1}^{n} \frac{z^e_{i+1} - f(e^{i+1}) \cdot b}{n - i + 1} \\
    & = z^e_2 + \sum_{i = 1}^n \frac{z^e_{i + 1} - \left \lfloor  z^e_{i + 1} + \frac{n^2}{n - i + 2}   \right \rfloor \cdot [y_i = 1] - \lfloor z^e_{i+1} - 1 \rfloor \cdot [y_i = 0]}{n-i+1}   \\
    & \ge z^e_2 + \sum_{i = 1}^n \frac{z^e_{i + 1} - \left \lfloor  z^e_{i + 1} + \frac{n^2}{n - i + 2}   \right \rfloor }{n-i+1}   \\
    & \ge z^e_2 + \sum_{i = 1}^n \frac{z^e_{i + 1} - \left ( z^e_{i + 1} + \frac{n^2}{n - i + 2}   \right )}{n-i+1} \\
    & = z^e_2 - \sum_{i = 1}^n \frac{n^2} {(n-i+1)(n-i+2)} \\
    & = z^e_2 -  n^2\sum_{i = 1}^n\frac1 {(n-i+1)(n-i+2)} \\
    &  = z^e_2 - n^2 \left (\frac{1}{(n + 1)n} + \frac{1}{n(n - 1)} + \ldots + \frac{1}{1\cdot2} \right) \\
    &  = z^e_2 - n^2 \left [ \left (\frac{1}{n} - \frac{1}{n+1} \right) + \left (\frac{1}{n-1} - \frac{1}{n} \right) + \ldots +  \left(1 - \frac12\right) \right] \\
    & = z^e_2 - n^2 \left(1-\frac1{n+1} \right) \\
    & > a + \frac{1}{n + 1} - n^2,
\end{align*}
which concludes the proof.

\end{proof}

The following proof proves Lemma~\ref{lem:subprobs2} which ensures our distribution is valid.

\begin{proof}[Proof of Lemma~\ref{lem:subprobs2}]
\begin{align*}
\sum_{k=1}^{n+2} f(e^k) &= f(e^{n+2}) + \sum_{k=1}^{n+1} f(e^k)  \\
&= \frac{b - \sum_{k=1}^{n+1}  f(e^k) \cdot b}{b} + \sum_{k=1}^{n+1} f(e^k) \\
&= 1/2 - \sum_{k=1}^{n+1} f(e^k) + \sum_{k=1}^{n+1} f(e^k)\\
&=1/2.
\end{align*}
\end{proof}

Again, utilizing the properties above, we conclude this subsection by proving a key property of our construction. It guarantees that Bidder Two's distribution is nearly-uniform over $v_1^2,\ldots, v_1^{n+2}$.

\begin{proof}[Proof of Lemma~\ref{erange}]

By Lemma \ref{zerange} we have
\begin{align*}
    f(e^{i+1}) \cdot b & = \left \lfloor  z^e_{i + 1} - 1 \right \rfloor \cdot [y_i = 0] + \left \lfloor  z^e_{i + 1} + \frac{n^2}{n - i + 2}   \right \rfloor \cdot [y_i = 1] \\
    & \ge \left \lfloor  z^e_{i + 1} - 1 \right \rfloor \\
    & > z^e_{i + 1} - 2 \\
    & \ge a - 2n^2  - 2.
\end{align*}
On the other hand, we have
\begin{align*}
     f(e^{i+1}) \cdot b & = \left \lceil  z^e_{i + 1} \right \rceil \cdot [y_i = 0] + \left \lfloor  z^e_{i + 1} + \frac{n^2}{n - i + 2}   \right \rfloor \cdot [y_i = 1] \\
    & \le \left \lfloor  z^e_{i + 1} + \frac{n^2}{n - i + 2}   \right \rfloor \\
    & \le z^e_{i + 1} + \frac{n^2}{n - i + 2} \\
    & \le z^e_{i + 1} + n^2 \\
    & \le a + 3n^2,
\end{align*}
which concludes the proof.

\end{proof}

\section{Omitted Proofs from Section~\ref{sec:flow}}\label{app:flow}

\subsection{Analyzing the Canonical Flow}

Here we are going to use properties proved in Appendix~\ref{app:reduction} to analyze the canonical flow. First we will present some technical lemmas which asymptotically bound the gaps between virtual values. By comparing the orders of those gaps, we then wrap them up to obtain Proposition~\ref{prop:Myerson}.

The following lemmas establish some useful bounds for Myerson virtual values in our reduction.

\begin{lemma} \label{cbound}
  $ v^{i+1} - (n - i + 1) + \frac{\frac12 n^3 - n }{a + 3n^3} \le \Phi^{\alpha^M,\lambda^M}(c^{i+1}) \le v^{i+1} - (n - i + 1) + \frac{2n^3}{a - 2n^3}$ for all $i \in [1,n]$  under the canonical flow.
\end{lemma}

\begin{proof}
  Recall that, by definition $z^c_{i + 1} = \frac{b - \sum_{j = 1}^{i} f(c^j) \cdot 2b}{n - i + 2} = \frac{\sum_{j = i+1}^{n + 2} f(c^j) \cdot 2b}{n - i + 2}$.
  Let us first lower bound $\Phi^{\alpha^M,\lambda^M}(c^{i+1})$. By Lemma \ref{zcmonotone} and \ref{zcrange}, we know $z^c_{i + 2} < z^c_{i + 1}$ and $z^c_{i+1} \le a + 2n^3 $ for all $i \in [1, n]$, we then have the following inequality,
\begin{align*}
    \Phi^{\alpha^M,\lambda^M}(c^{i+1}) & = v^{i+1} - \frac{\sum_{k=i+2}^{n+2} f(c^k)}{f(c^{i+1})} \\
    & = v^{i+1} - (n - i + 1)\frac{z^c_{i+2}}{f(c^{i+1})\cdot 2b} \\
    & = v^{i+1} - (n - i + 1)\frac{z^c_{i+2}}{\left \lfloor z^c_{i+1} + \frac{n^3}{n-i+2}\right \rfloor} \\
    & \ge v^{i+1} - (n - i + 1)\frac{z^c_{i+2}}{z^c_{i+1} + \frac{n^3}{n-i+2} - 1}  \\
    & \ge v^{i+1} - (n - i + 1)\frac{z^c_{i+1} }{z^c_{i+1} + \frac{n^3}{n-i+2} - 1}  \\
    & = v^{i+1} - (n - i + 1) \left ( 1 - \frac{\frac{n^3}{n-i+2} - 1 }{z^c_{i+1} + \frac{n^3}{n-i+2} - 1} \right) \\
    & = v^{i+1} - (n - i + 1) + (n - i + 1) \frac{\frac{n^3}{n-i+2} -1 }{z^c_{i+1} + \frac{n^3}{n-i+2} - 1} \\
    & \ge v^{i+1} - (n - i + 1) + \frac{\frac12 n^3 - n }{z^c_{i+1} + \frac{n^3}{n-i+2}} \\
    & \ge v^{i+1} - (n - i + 1) + \frac{\frac12 n^3 - n }{a + 3n^3}.
\end{align*}
On the other hand, by Lemma \ref{zcgap} we know $z^c_{i + 1} - z^c_{i + 2} \le \frac{n^3}{(n - i + 2)(n - i + 1)}$ for all $i \in [1, n]$, then
\begin{align*}
    \Phi^{\alpha^M,\lambda^M}(c^{i+1}) & = v^{i+1} - \frac{\sum_{k=i+2}^{n+2} f(c^k)}{f(c^{i+1})} \\
    & = v^{i+1} - (n - i + 1)\frac{z^c_{i+2}}{f(c^{i+1})\cdot 2b} \\
    & = v^{i+1} - (n - i + 1)\frac{z^c_{i+2}}{\left \lfloor z^c_{i+1} + \frac{n^3}{n-i+2}\right \rfloor} \\
    & \le v^{i+1} - (n - i + 1)\frac{z^c_{i+2}}{z^c_{i+1} + \frac{n^3}{n-i+2} }  \\
    & \le v^{i+1} - (n - i + 1)\frac{z^c_{i+1} - \frac{n^3}{(n - i + 1)(n - i + 2)} }{z^c_{i+1} + \frac{n^3}{n-i+2} }  \\
    & = v^{i+1} - (n - i + 1) \left ( 1 - \frac{\frac{n^3}{n-i+2} + \frac{n^3}{(n - i + 1)(n - i + 2)}}{z^c_{i+1} + \frac{n^3}{n-i+2}} \right) \\
    & = v^{i+1} - (n - i + 1) + (n - i + 1) \frac{\frac{n^3}{n-i+2} + \frac{n^3}{(n - i + 1)(n - i + 2) }}{z^c_{i+1} + \frac{n^3}{n-i+2}} \\
    & \le v^{i+1} - (n - i + 1) + (n - i + 1) \frac{\frac{n^3}{n-i+1} + \frac{n^3}{(n - i + 1)(n - i + 2) }}{z^c_{i+1}} \\
    & \le v^{i+1} - (n - i + 1) + \frac{n^3 + \frac{n^3}{n - i + 2 }}{z^c_{i+1}} \\
    & \le v^{i+1} - (n - i + 1) + \frac{2n^3}{a - 2n^3}.
\end{align*}
\end{proof}

\begin{lemma} \label{d1bound}
   If $x_i = 0$, $ v^{i+1} - (n - i + 1) + \frac{\frac12 n^3 - n }{a + 3n^3} \le \Phi^{\alpha^M,\lambda^M}(d^{i+1}) \le v^{i+1} - (n - i + 1) + \frac{2n^3}{a - 2n^3}$ for all $i \in [1,n]$  under the canonical flow.
\end{lemma}

\begin{proof}
  Recall that, by definition $z^d_{i + 1} = \frac{b - \sum_{j = 1}^{i} f(d^j) \cdot 2b}{n - i + 2} = \frac{\sum_{j = i+1}^{n + 2} f(d^j) \cdot 2b}{n - i + 2}$.
  Let us first lower bound $\Phi^{\alpha^M,\lambda^M}(d^{i+1})$. By Lemma \ref{zdmonotone} and \ref{zdrange} we have the following inequality,
\begin{align*}
    \Phi^{\alpha^M,\lambda^M}(d^{i+1}) & = v^{i+1} - \frac{\sum_{k=i+2}^{n+2} f(d^k)}{f(d^{i+1})} \\
    & = v^{i+1} - (n - i + 1)\frac{z^d_{i+2}}{f(d^{i+1})\cdot 2b} \\
    & = v^{i+1} - (n - i + 1)\frac{z^d_{i+2}}{\left \lfloor z^d_{i+1} + \frac{n^3}{n-i+2}\right \rfloor} \\
    & \ge v^{i+1} - (n - i + 1)\frac{z^d_{i+2}}{z^d_{i+1} + \frac{n^3}{n-i+2} - 1}  \\
    & \ge v^{i+1} - (n - i + 1)\frac{z^d_{i+1} }{z^d_{i+1} + \frac{n^3}{n-i+2} - 1}  \\
    & = v^{i+1} - (n - i + 1) \left ( 1 - \frac{\frac{n^3}{n-i+2} - 1 }{z^d_{i+1} + \frac{n^3}{n-i+2} - 1} \right) \\
    & = v^{i+1} - (n - i + 1) + (n - i + 1) \frac{\frac{n^3}{n-i+2} -1 }{z^d_{i+1} + \frac{n^3}{n-i+2} - 1} \\
    & \ge v^{i+1} - (n - i + 1) + \frac{\frac12 n^3 - n }{z^d_{i+1} + \frac{n^3}{n-i+2}} \\
    & \ge v^{i+1} - (n - i + 1) + \frac{\frac12 n^3 - n }{a + 3n^3}.
\end{align*}
On the other hand, by Lemma \ref{zd1gap} and \ref{zdrange} we have
\begin{align*}
    \Phi^{\alpha^M,\lambda^M}(d^{i+1}) & = v^{i+1} - \frac{\sum_{k=i+2}^{n+2} f(d^k)}{f(d^{i+1})} \\
    & = v^{i+1} - (n - i + 1)\frac{z^d_{i+2}}{f(d^{i+1})\cdot 2b} \\
    & = v^{i+1} - (n - i + 1)\frac{z^d_{i+2}}{\left \lfloor z^d_{i+1} + \frac{n^3}{n-i+2}\right \rfloor} \\
    & \le v^{i+1} - (n - i + 1)\frac{z^d_{i+2}}{z^d_{i+1} + \frac{n^3}{n-i+2} }  \\
    & \le v^{i+1} - (n - i + 1)\frac{z^d_{i+1} - \frac{n^3}{(n - i + 1)(n - i + 2)} }{z^d_{i+1} + \frac{n^3}{n-i+2} }  \\
    & = v^{i+1} - (n - i + 1) \left ( 1 - \frac{\frac{n^3}{n-i+2} + \frac{n^3}{(n - i + 1)(n - i + 2)}}{z^d_{i+1} + \frac{n^3}{n-i+2}} \right) \\
    & = v^{i+1} - (n - i + 1) + (n - i + 1) \frac{\frac{n^3}{n-i+2} + \frac{n^3}{(n - i + 1)(n - i + 2) }}{z^d_{i+1} + \frac{n^3}{n-i+2}} \\
    & \le v^{i+1} - (n - i + 1) + (n - i + 1) \frac{\frac{n^3}{n-i+1} + \frac{n^3}{(n - i + 1)(n - i + 2) }}{z^d_{i+1}} \\
    & \le v^{i+1} - (n - i + 1) + \frac{n^3 + \frac{n^3}{n - i + 2 }}{z^d_{i+1}} \\
    & \le v^{i+1} - (n - i + 1) + \frac{2n^3}{a - 2n^3}.
\end{align*}
\end{proof}

\begin{lemma} \label{d0bound}
  If $x_i = 1$, $ v^{i+1} - (n - i + 1) \le \Phi^{\alpha^M,\lambda^M}(d^{i+1}) \le v^{i+1} - (n - i + 1) + \frac{2n}{a - 2n^3}$ for all $i \in [1,n]$  under the canonical flow.
\end{lemma}

\begin{proof}
  Recall that, by definition $z^d_{i + 1} = \frac{b - \sum_{j = 1}^{i} f(d^j) \cdot 2b}{n - i + 2} = \frac{\sum_{j = i+1}^{n + 2} f(d^j) \cdot 2b}{n - i + 2}$.
  Let us first lower bound $\Phi^{\alpha^M,\lambda^M}(d^{i+1})$. By Lemma \ref{zdmonotone} and \ref{zdrange} we have the following inequality,
\begin{align*}
    \Phi^{\alpha^M,\lambda^M}(d^{i+1}) & = v^{i+1} - \frac{\sum_{k=i+2}^{n+2} f(d^k)}{f(d^{i+1})} \\
    & = v^{i+1} - (n - i + 1)\frac{z^d_{i+2}}{f(d^{i+1})\cdot 2b} \\
    & = v^{i+1} - (n - i + 1)\frac{z^d_{i+2}}{\left \lceil z^d_{i+1} \right \rceil} \\
    & \ge v^{i+1} - (n - i + 1)\frac{z^d_{i+2}}{z^d_{i+1}}  \\
    & \ge v^{i+1} - (n - i + 1)\frac{z^d_{i+1} }{z^d_{i+1}}  \\
    & \ge v^{i+1} - (n - i + 1).
\end{align*}
On the other hand, by Lemma \ref{zd0gap} and \ref{zdrange} we have
\begin{align*}
    \Phi^{\alpha^M,\lambda^M}(d^{i+1}) & = v^{i+1} - \frac{\sum_{k=i+2}^{n+2} f(d^k)}{f(d^{i+1})} \\
    & = v^{i+1} - (n - i + 1)\frac{z^d_{i+2}}{f(d^{i+1})\cdot 2b} \\
    & = v^{i+1} - (n - i + 1)\frac{z^d_{i+2}}{\left \lceil z^d_{i+1} \right \rceil} \\
    & \le v^{i+1} - (n - i + 1)\frac{z^d_{i+2}}{z^d_{i+1} + 1 }  \\
    & \le v^{i+1} - (n - i + 1)\frac{z^d_{i+1} - \frac{1}{n - i + 1} }{z^d_{i+1} + 1 }  \\
    & = v^{i+1} - (n - i + 1) \left (1 - \frac{1 + \frac{1}{n - i + 1}}{z^d_{i+1} + 1} \right) \\
    & = v^{i+1} - (n - i + 1) + (n - i + 1) \frac{1 + \frac{1}{n - i + 1 }}{z^d_{i+1} + 1} \\
    & \le v^{i+1} - (n - i + 1) + (n - i + 1) \frac{1 + \frac{1}{n - i + 1 }}{z^d_{i+1}} \\
    & \le v^{i+1} - (n - i + 1) + \frac{n - i + 2}{z^d_{i+1}} \\
    & \le v^{i+1} - (n - i + 1) + \frac{2n}{a - 2n^3}.
\end{align*}
\end{proof}

\begin{lemma} \label{e1bound}
  If $y_i = 1$, $ v^{i+1} - (n - i + 1) + \frac{\frac12 n^2 - 3n }{a + 3n^3} \le \Phi^{\alpha^M,\lambda^M}(e^{i+1}) \le v^{i+1} - (n - i + 1) + \frac{2n^2}{a - 2n^3}$ for all $i \in [1,n]$  under the canonical flow.
\end{lemma}

\begin{proof}
  Recall that, by definition $z^e_{i + 1} = \frac{b - \sum_{j = 1}^{i} f(e^j) \cdot b}{n - i + 2} = \frac{\sum_{j = i+1}^{n + 2} f(e^j) \cdot b}{n - i + 2}$.
  Let us first lower bound $\Phi^{\alpha^M,\lambda^M}(e^{i+1})$. By Lemma \ref{zenearlymonotone} and \ref{zerange}, $z^e_{i + 2} < z^e_{i + 1} + 2$ for all $i \in [1, n]$, and $z^e_{i+1} \le a + 2n^2$, for all $i \in [1, n + 1]$, we have the following inequality,
\begin{align*}
    \Phi^{\alpha^M,\lambda^M}(e^{i+1}) & = v^{i+1} - \frac{\sum_{k=i+2}^{n+2} f(e^k)}{f(e^{i+1})} \\
    & = v^{i+1} - (n - i + 1)\frac{z^e_{i+2}}{f(e^{i+1})\cdot b} \\
    & = v^{i+1} - (n - i + 1)\frac{z^e_{i+2}}{\left \lfloor z^e_{i+1} + \frac{n^3}{n-i+2}\right \rfloor} \\
    & \ge v^{i+1} - (n - i + 1)\frac{z^e_{i+2}}{z^e_{i+1} + \frac{n^2}{n-i+2} - 1}  \\
    & \ge v^{i+1} - (n - i + 1)\frac{z^e_{i+1} + 2 }{z^e_{i+1} + \frac{n^2}{n-i+2} - 1}  \\
    & = v^{i+1} - (n - i + 1) \left ( 1 - \frac{\frac{n^2}{n-i+2} - 3 }{z^e_{i+1} + \frac{n^2}{n-i+2} - 1} \right) \\
    & = v^{i+1} - (n - i + 1) + (n - i + 1) \frac{\frac{n^2}{n-i+2} -3 }{z^e_{i+1} + \frac{n^2}{n-i+2} - 1} \\
    & \ge v^{i+1} - (n - i + 1) + \frac{\frac12 n^2 - 3n }{z^e_{i+1} + \frac{n^2}{n-i+2}} \\
    & \ge v^{i+1} - (n - i + 1) + \frac{\frac12 n^2 - 3n }{a + 3n^3}.
\end{align*}
On the other hand, by Lemma \ref{ze1gap}, $z^e_{i + 1} - z^e_{i + 2} \le \frac{n^2}{(n - i + 2)(n - i + 1)}$, so we have
\begin{align*}
    \Phi^{\alpha^M,\lambda^M}(e^{i+1}) & = v^{i+1} - \frac{\sum_{k=i+2}^{n+2} f(e^k)}{f(e^{i+1})} \\
    & = v^{i+1} - (n - i + 1)\frac{z^e_{i+2}}{f(e^{i+1})\cdot b} \\
    & = v^{i+1} - (n - i + 1)\frac{z^e_{i+2}}{\left \lfloor z^e_{i+1} + \frac{n^2}{n-i+2}\right \rfloor} \\
    & \le v^{i+1} - (n - i + 1)\frac{z^e_{i+2}}{z^e_{i+1} + \frac{n^2}{n-i+2} }  \\
    & \le v^{i+1} - (n - i + 1)\frac{z^e_{i+1} - \frac{n^2}{(n - i + 1)(n - i + 2)} }{z^e_{i+1} + \frac{n^2}{n-i+2} }  \\
    & = v^{i+1} - (n - i + 1) \left ( 1 - \frac{\frac{n^2}{n-i+2} + \frac{n^2}{(n - i + 1)(n - i + 2)}}{z^e_{i+1} + \frac{n^2}{n-i+2}} \right) \\
    & = v^{i+1} - (n - i + 1) + (n - i + 1) \frac{\frac{n^2}{n-i+2} + \frac{n^2}{(n - i + 1)(n - i + 2) }}{z^e_{i+1} + \frac{n^2}{n-i+2}} \\
    & \le v^{i+1} - (n - i + 1) + (n - i + 1) \frac{\frac{n^2}{n-i+1} + \frac{n^2}{(n - i + 1)(n - i + 2) }}{z^e_{i+1}} \\
    & \le v^{i+1} - (n - i + 1) + \frac{n^2 + \frac{n^2}{n - i + 2 }}{z^e_{i+1}} \\
    & \le v^{i+1} - (n - i + 1) + \frac{2n^2}{a - 2n^3}.
\end{align*}
\end{proof}

\begin{lemma} \label{e0bound}
  If $y_i = 0$, $ v^{i+1} - (n - i + 1) - \frac{4n}{a - 3n^3} \le \Phi^{\alpha^M,\lambda^M}(e^{i+1}) \le v^{i+1} - (n - i + 1) - \frac{1}{a + 2n^3}$ for all $i \in [1,n]$  under the canonical flow.
\end{lemma}

\begin{proof}
  Recall that, by definition $z^e_{i + 1} = \frac{b - \sum_{j = 1}^{i} f(e^j) \cdot b}{n - i + 2} = \frac{\sum_{j = i+1}^{n + 2} f(e^j) \cdot b}{n - i + 2}$.
  Let us first lower bound $\Phi^{\alpha^M,\lambda^M}(e^{i+1})$. By Lemma \ref{zenearlymonotone} and \ref{zerange}, $z^e_{i + 2} < z^e_{i + 1} + 2$ for all $i \in [1, n]$, and $z^e_{i+1} \ge a - 2n^2$, for all $i \in [1, n + 1]$, we have the following inequality,
\begin{align*}
    \Phi^{\alpha^M,\lambda^M}(e^{i+1}) & = v^{i+1} - \frac{\sum_{k=i+2}^{n+2} f(e^k)}{f(e^{i+1})} \\
    & = v^{i+1} - (n - i + 1)\frac{z^e_{i+2}}{f(e^{i+1})\cdot b} \\
    & = v^{i+1} - (n - i + 1)\frac{z^e_{i+2}}{\left \lfloor z^e_{i+1}  - 1\right \rfloor} \\
    & \ge v^{i+1} - (n - i + 1)\frac{z^e_{i+2}}{z^e_{i+1} - 2}  \\
    & \ge v^{i+1} - (n - i + 1)\frac{z^e_{i+1} + 2 }{z^e_{i+1} - 2}  \\
    & = v^{i+1} - (n - i + 1)\left ( 1 + \frac{4}{z^e_{i+1} - 2} \right )  \\
    & \ge v^{i+1} - (n - i + 1) - \frac{4n}{a - 3n^3}.
\end{align*}
On the other hand, by Lemma \ref{ze0gap} $z^e_{i + 1} - z^e_{i + 2} \le -\frac{1}{n-i+1}$, we have
\begin{align*}
    \Phi^{\alpha^M,\lambda^M}(e^{i+1}) & = v^{i+1} - \frac{\sum_{k=i+2}^{n+2} f(e^k)}{f(e^{i+1})} \\
    & = v^{i+1} - (n - i + 1)\frac{z^e_{i+2}}{f(e^{i+1})\cdot b} \\
    & = v^{i+1} - (n - i + 1)\frac{z^e_{i+2}}{\left \lfloor z^e_{i+1}  - 1\right \rfloor} \\
    & \le v^{i+1} - (n - i + 1)\frac{z^e_{i+2}}{z^e_{i+1} }  \\
    & \le v^{i+1} - (n - i + 1)\frac{z^e_{i+1} + \frac{1}{n - i + 1} }{z^e_{i+1} }  \\
    & = v^{i+1} - (n - i + 1) \left (1 + \frac{\frac{1}{n - i + 1}}{z^e_{i+1}} \right) \\
    & = v^{i+1} - (n - i + 1) - \frac{1}{z^e_{i+1}} \\
    & \le v^{i+1} - (n - i + 1) - \frac{1}{a + 2n^3}.
\end{align*}
\end{proof}

We now have enough tools to prove Proposition~\ref{prop:Myerson}. Firstly we will prove the monotonicity of virtual values, which establishes the first two bullets of Proposition~\ref{prop:Myerson}.
\begin{lemma} \label{monotone}
For all $x, y \in \{0, 1\}^n$, we have \[
\max \left (\Phi^{\alpha^M,\lambda^M}(c^{i}) , \Phi^{\alpha^M,\lambda^M}(d^{i}),  \Phi^{\alpha^M,\lambda^M}(e^{i})  \right) <  \min \left (\Phi^{\alpha^M,\lambda^M}(c^{i + 1}) , \Phi^{\alpha^M,\lambda^M}(d^{i + 1}),  \Phi^{\alpha^M,\lambda^M}(e^{i + 1})  \right) \] for $i \in [1, n + 1]$.
\end{lemma}

\begin{proof}
For convenience, let \[L_i = \min \left (\Phi^{\alpha^M,\lambda^M}(c^{i + 1}) , \Phi^{\alpha^M,\lambda^M}(d^{i + 1}),  \Phi^{\alpha^M,\lambda^M}(e^{i + 1})  \right),\] and \[R_i = \max \left (\Phi^{\alpha^M,\lambda^M}(c^{i + 1}) , \Phi^{\alpha^M,\lambda^M}(d^{i + 1}),  \Phi^{\alpha^M,\lambda^M}(e^{i + 1})  \right).\]
We prove this lemma by showing the gap between $R_i$ and $L_{i+1}$ is at least 1 for all $i$ under the canonical flow.

Firstly, $R_1 = \Phi^{\alpha^M,\lambda^M}(c^{1}) = n^2 - 10n + 2$, and by Lemma \ref{crange}, \ref{drange} and \ref{erange}, $L_2 \ge n^2 + 2 - n \cdot \frac{a + 2n^3}{a-2n^3}$. This lemma clearly holds for $i = 1$.

For $i > 1$, again by  Lemma \ref{crange}, \ref{drange} and \ref{erange},  we have
\begin{align*}
    L_{i + 1} - R_{i} & \ge 1 - (n - i + 1)\frac{a + 2n^3}{a - 2n^3} + (n - i + 2)\frac{a - 2n^3}{a + 2n^3} \\
    & = 1 - (n - i + 1) \left ( 1 + \frac{4n^3}{a - 2n^3} \right) + (n - i + 2) \left ( 1 - \frac{4n^3}{a + 2n^3} \right) \\
    & = 2 -  (n - i + 1) \frac{4n^3}{a - 2n^3}  + (n - i + 2)  \frac{4n^3}{a + 2n^3} \\
    & \ge 1
\end{align*}

\end{proof}

The next three lemmas will establish bullet three and four of Proposition~\ref{prop:Myerson}.

\begin{lemma} \label{le}
If either $x_{i} = 0$ or $y_{i} = 0$, then $\Phi^{\alpha^M,\lambda^M}(d^{i+1}) > \Phi^{\alpha^M,\lambda^M}(e^{i+1})$  for all $i \in [1,n]$ for the canonical flow.
\end{lemma}
\begin{proof}
By Lemma \ref{e0bound}, \ref{d0bound} and \ref{d1bound}, we know  $\Phi^{\alpha^M,\lambda^M}(d^{i+1}) \ge v^{i+1} - (n - i + 1)$ whereas  $\Phi^{\alpha^M,\lambda^M}(e^{i+1}) < v^{i+1} - (n - i + 1)$ when $y_i = 0$. Thus the only case left is $x_i = 0$ and $y_i = 1$, again by Lemma \ref{d1bound} and \ref{e1bound}, we have
\begin{align*}
  \Phi^{\alpha^M,\lambda^M}(d^{i+1}) - \Phi^{\alpha^M,\lambda^M}(e^{i+1}) & \ge \frac{\frac12 n^3 - 3n }{a + 3n^3} -  \frac{2n^2}{a - 2n^3} \\
  & = \Omega(1/n^2),
\end{align*}
where the last equality is from the fact that $a = \Theta(n^5)$.

\end{proof}

\begin{lemma} \label{upper}
 $\Phi^{\alpha^M,\lambda^M}(c^{i+1}) > \Phi^{\alpha^M,\lambda^M}(e^{i+1})$ for all $i \in [1,n]$ under the canonical flow. Furthermore,
 \[\min_{i \in [1,n]} \left (\Phi^{\alpha^M,\lambda^M}(c^{i+1}) - \Phi^{\alpha^M,\lambda^M}(e^{i+1}) \right ) = \Theta(1/n^2),\] and \[\max_{i \in [1,n]} \left (\Phi^{\alpha^M,\lambda^M}(c^{i+1}) - \Phi^{\alpha^M,\lambda^M}(e^{i+1}) \right ) = \Theta(1/n^2).\]
\end{lemma}
\begin{proof}
By lemma \ref{cbound} and \ref{e0bound} we have
\begin{align*}
    \Phi^{\alpha^M,\lambda^M}(c^{i+1}) - \Phi^{\alpha^M,\lambda^M}(e^{i+1}) & \ge \frac{\frac12 n^3 - n }{a + 3n^3} + \frac{1}{a + 2n^3} \\
    & = \Omega(1/n^2),
\end{align*}
where the last equality is from the fact that $a = \Theta(n^5)$.
On the other hand, we have
\begin{align*}
    \Phi^{\alpha^M,\lambda^M}(c^{i+1}) - \Phi^{\alpha^M,\lambda^M}(e^{i+1}) & \le \frac{2n^3}{a - 2n^3} + \frac{4n}{a - 3n^3} \\
    & = O(1/n^2),
\end{align*}
where the last equality is from the fact that $a = \Theta(n^5)$.
\end{proof}

\begin{lemma}\label{ge}
If $x_{i} = y_{i} = 1$, then $\Phi^{\alpha^M,\lambda^M}(d^{i+1}) < \Phi^{\alpha^M,\lambda^M}(e^{i+1})$ for the canonical flow. Furthermore, \[\min_{i : x_i = y_i = 1} \left (\Phi^{\alpha^M,\lambda^M}(e^{i+1}) -\Phi^{\alpha^M,\lambda^M}(d^{i+1})  \right ) = \Theta(1/n^3),\] and \[\max_{i: x_i = y_i = 1} \left (\Phi^{\alpha^M,\lambda^M}(e^{i+1}) -\Phi^{\alpha^M,\lambda^M}(d^{i+1}) \right ) = \Theta(1/n^3).\]
\end{lemma}
\begin{proof}
When $x_i = y_i = 1$, by Lemma \ref{d0bound} and \ref{e1bound} we have
\begin{align*}
\Phi^{\alpha^M,\lambda^M}(e^{i+1}) - \Phi^{\alpha^M,\lambda^M}(d^{i+1}) & \ge \frac{\frac12 n^2 - 3n }{a + 3n^3} - \frac{2n}{a - 2n^3} \\
& = \Omega(1/n^3),
\end{align*}
where the last equality is from the fact that $a = \Theta(n^5)$.
On the other hand, we have
\begin{align*}
\Phi^{\alpha^M,\lambda^M}(e^{i+1}) - \Phi^{\alpha^M,\lambda^M}(d^{i+1}) & \le \frac{2n^2}{a - 2n^3} - 0 \\
& = O(1/n^3),
\end{align*}
where the last equality is from the fact that $a = \Theta(n^5)$.
\end{proof}

Now we would like to show the relationship of virtual values of the lowest type for the canonical flow, which establishes bullet five of Proposition~\ref{prop:Myerson}.

\begin{lemma} \label{before}
$ \Phi^{\alpha^M,\lambda^M}(d^1) = \Phi^{\alpha^M,\lambda^M}(c^1)  > \Phi^{\alpha^M,\lambda^M}(e^1)$ for the canonical flow.
\end{lemma}

\begin{proof}
The first equality holds since $f(c^1) = f(d^1)$, and the probability for day $1$ and day $2$ are both $1/2$. The inequality holds since
\begin{align*}
    \Phi^{\alpha^M,\lambda^M}(e^1) & = v^1 - (n + 1) \frac{z^e_2}{f(e^1) \cdot b} \\
    & =  v^1 - (n + 1) \frac{\frac{b - \frac{b}{10n} + 1}{n + 1}}{f(e^1) \cdot b} \\
    & = v^1 - (n + 1) \frac{a + \frac{1}{n + 1}}{\frac{b}{10n} - 1} \\
    & < v^1 - (n + 1) \frac{a}{\frac{b}{10n}} \\
    & =  \Phi^{\alpha^M,\lambda^M}(c^1).
\end{align*}
\end{proof}

Bullet six is easy to verify, by the definition of virtual values, $\Phi^{\alpha^M,\lambda^M}_1((v_1^{n+2},j)) = \Phi^{\alpha^M,\lambda^M}_2((v_2^{n+2},1)) = n^2 + n + 2$.

Finally, the following lemma shows all virtual values are positive which implies the last bullet of Proposition~\ref{prop:Myerson}.
\begin{lemma} \label{positive}
  By Lemma~\ref{monotone} and Lemma~\ref{before}, it is sufficient to show $\Phi^{\alpha^M,\lambda^M}(e^1) > 0$. To this end, observe that:
  \begin{align*}
    \Phi^{\alpha^M,\lambda^M}(e^1) & = v^1 - (n + 1) \frac{z^e_2}{f(e^1) \cdot b} \\
    & =  v^1 - (n + 1) \frac{\frac{b - \frac{b}{10n} + 1}{n + 1}}{f(e^1) \cdot b} \\
    & = v^1 - (n + 1) \frac{a + \frac{1}{n + 1}}{\frac{b}{10n} - 1} \\
    & = n^2 + 1 - (n + 1) \frac{a + \frac{1}{n + 1}}{\frac{b}{10n} - 1} \\
    & > n^2 + 1 - 20(n + 1) \\
    & > 0
  \end{align*}
\end{lemma}

\subsection{Analyzing the Modified Flow}

First we bound the flow we need for the boosting operation:
\begin{lemma} \label{boundofflow}
When flow is needed, $\varepsilon = \max \left (\frac{\Phi^{\alpha^M,\lambda^M}(e^{i}) - \Phi^{\alpha^M,\lambda^M}(d^{i})}{f(d^i)} \right) = \Theta(1/n^4)$.
\end{lemma}

\begin{proof}
This is directly from Lemma \ref{drange} and \ref{ge}.
\end{proof}

After knowing the order of flow we need, we can compare it with the bounds for gaps between virtual values we have already established. We will use it and all properties we proved above to establish Proposition~\ref{prop:valid}. The following lemma guarantees that monotonicity still hold after boosting, which in particular implies the first two bullets of Proposition~\ref{prop:valid}.

\begin{lemma} \label{monotone1}
For all $x, y \in \{0, 1\}^n$, we have \[
\max \left (\Phi^{\alpha^*,\lambda^*}(c^{i}) , \Phi^{\alpha^*,\lambda^*}(d^{i}),  \Phi^{\alpha^*,\lambda^*}(e^{i})  \right) <  \min \left (\Phi^{\alpha^*,\lambda^*}(c^{i + 1}) , \Phi^{\alpha^*,\lambda^*}(d^{i + 1}),  \Phi^{\alpha^*,\lambda^*}(e^{i + 1})  \right) \] for $i \in [1, n + 1]$.
\end{lemma}

\begin{proof}
By Lemma~\ref{monotone}, we know the gap under canonical flow is at least $1$, and we are only adding negligible amount of flow ($O(\frac{1}{n^4})$), so the overall perturbation on virtual values will be no more than $\frac{\varepsilon}{f(\cdot)}$, which is $O(\frac{1}{n^3})$ by Lemma~\ref{crange}, Lemma~\ref{drange} and Lemma~\ref{erange}. Thus the monotonicity still holds.
\end{proof}

The next lemma together with the definition of modified flow establishes the bullet three of Proposition~\ref{prop:valid}.

\begin{lemma} \label{ce}
For all $x, y$, $\Phi^{\alpha^*,\lambda^*}(c^{i + 1}) > \Phi^{\alpha^*,\lambda^*}(e^{i + 1})$ for all $i \in [1, n]$ in the new flow.
\end{lemma}
\begin{proof}
By Lemma \ref{upper}, we know  \[\min_{i \in [1,n]} \left (\Phi^{\alpha^*,\lambda^*}(c^{i+1}) - \Phi^{\alpha^*,\lambda^*}(e^{i+1}) \right ) = \Theta(1/n^2).\] And by Lemma \ref{boundofflow}, we know that $\varepsilon = \Theta(1/n^4)$. Thus for $i \in [1, n]$ the virtual value of $c^{i+1}$ decreases only by $\Theta(1/n^3) < \Theta(1/n^2)$, which completes the proof.
\end{proof}

Note that the bullet four of Proposition~\ref{prop:valid} is directly from the bullet four of Proposition~\ref{prop:Myerson}, which guarantees that if $\disj = \mathsf{no}$, there must be some $k$ such that $\Phi^{\alpha^M,\lambda^M}_1((v_1^{k + 1},1)) > \Phi^{\alpha^M,\lambda^M}_2((v_2^{k + 1},1)) > \Phi^{\alpha^M,\lambda^M}_1((v_1^{k + 1},2))$. Then by the definition of modified flow, there must be some $k^*$ satisfies the property.

The following lemma shows bullet five of Proposition~\ref{prop:valid}.

\begin{lemma} \label{after}
If additional flow in step $2$ is needed, $ \Phi^{\alpha^*,\lambda^*}(d^1) > \Phi^{\alpha^*,\lambda^*}(e^1) > \Phi^{\alpha^*,\lambda^*}(c^1) $ for the new flow.
\end{lemma}

\begin{proof}
We prove this lemma by calculating the gap between $\Phi^{\alpha^*,\lambda^*}(c^1)$ and $\Phi^{\alpha^*,\lambda^*}(e^1)$  under the canonical flow, from which we are able to get the minimum amount of the additional flow needed to let the flip happen. Finally we show this minimum amount of flow is smaller than $\varepsilon$ for sufficiently large $n$.

As we calculated previously, under the canonical flow,

\begin{align*}
    \Phi^{\alpha^*,\lambda^*}(c^1) -  \Phi^{\alpha^*,\lambda^*}(e^1) & = (n + 1) \left( \frac{a + \frac{1}{n + 1}}{\frac{b}{10n} - 1}  -   \frac{a}{\frac{b}{10n}} \right) \\
    & = (n + 1) \left ( \frac{\frac{ab}{10n} + \frac{b}{10n^2 + 10n}}{\frac{b^2}{100n^2} - \frac{b}{10n}} - \frac{\frac{ab}{10n} - a}{\frac{b^2}{100n^2} - \frac{b}{10n}}  \right)\\
    & = (n + 1) \frac{\frac{b}{10n^2 + 10n} + a}{\frac{b^2}{100n^2} - \frac{b}{10n}}\\
    & = \Theta(1/n^4),
\end{align*}

where the last equality is by the fact that $b = \Theta(n^6)$ and $a = \Theta(n^5)$. Thus, combining with the fact that $f(c^1) = \Theta(1/n)$ ,the minimum amount of flow needed is  $(\Phi^{\alpha^*,\lambda^*}(c^1) - \Phi^{\alpha^*,\lambda^*}(e^1)) \cdot f(c^1) = \frac{1}{20n} = O(1/n^5)$. From Lemma \ref{boundofflow} we konw $\varepsilon = \Theta(1/n^4)$ which is larger than the $O(1/n^5)$ threshold.

\end{proof}

Bullet six is easy to verify, by the definition of virtual values, $\Phi^{\alpha^*,\lambda^*}_1((v_1^{n+2},j)) = \Phi^{\alpha^*,\lambda^*}_2((v_2^{n+2},1)) = n^2 + n + 2$.

Finally, we prove the last lemma of the paper.

\begin{proof}[Proof of Lemma~\ref{lem:spaBIC}]
We first confirm that the second-price auction with careful tie-breaking at $k^*$ is BIC. Observe first that, because the second-price auction with careful tie-breaking has a monotone allocation rule (and satisfies the payment identity) that no bidder can ever benefit by misreporting their value, but honestly reporting their interest. This follows immediately from the definition of the payment identity, and is well-known~\cite{Myerson81}. In particular, this implies that the auction is BIC for Bidder Two.

Now let's consider the utility of Bidder One. For notational convenience, define $u(t, t')$ to be the utility that Bidder One gets with type $t$ for reporting $t'$. Because the second-price auction yields a monotone allocation rule, it is well-known (and follows from straight-forward calculations) that $u( (v_1^k,j),(v_1^k,j)) \geq u( (v_1^k,j),(v_1^\ell,j))$ for both $j$ and all $k$~\cite{Myerson81}. This means that if Bidder One can benefit from misreporting, it must be because they misreport their interest. Clearly, Bidder One cannot benefit by misreporting their interest when their interest is day1 (because then they guarantee non-positive utility), so the only consideration is when their interest is day2.

Now, it is also easy to see (and observed in~\cite{FiatGKK16}) that $u( (v_1^k,2),(v_1^\ell,1)) = u( (v_1^k,1),(v_1^\ell,1))$ for all $k,\ell$. This means that if $u( (v_1^k,2),(v_1^k,2)) \geq u( (v_1^k,1),(v_1^k,1))$, we can conclude that $u( (v_1^k,2),(v_1^k,2)) \geq u( (v_1^k,1),(v_1^k,1)) \geq u( (v_1^k,1),(v_1^\ell,1)) =  u( (v_1^k,2),(v_1^\ell,1))$ for all $\ell$. So we need only establish that $u( (v_1^k,2),(v_1^k,2)) \geq u( (v_1^k,1),(v_1^k,1))$ for all $k$ to conclude that the auction is BIC. To this end, observe that:

\begin{align}
u((v_1^k,j),(v_1^k,j)) &= v_1^k \cdot \pi_1((v_1^k,j)) - p_1((v_1^k,j))\nonumber\\
&=v_1^k \cdot \pi_1((v_1^k,j)) - \sum_{\ell =1}^{ k} v_1^\ell \cdot (\pi_1( (v_1^\ell,j)) - \pi_1((v_1^{\ell-1},j)))\nonumber\\
&= \sum_{\ell =1}^{ k} (v_1^\ell-v_1^{\ell-1}) \cdot \pi_1 ((v_1^{\ell-1},j))\nonumber\\
&= \sum_{\ell=0}^{k-1} \pi_1 ((v_1^{\ell},j)).\label{eq:utility}
\end{align}

Now, recalling the format of our auction, we have that:
\begin{itemize}
\item When $k > 1$: $\pi_1((v_1^{k},1)) = \sum_{\ell \leq k} f_2((v_2^{\ell},1))$.
\item When $k = 1$: $\pi_1((v_1^{1},1)) = 0$.
\item When $k \neq k^*$: $\pi_1((v_1^k,2)) = \sum_{\ell \leq k} f_2((v_2^{\ell},1))$.
\item When $k = k^*$: $\pi_1((v_1^{k^*},1)) = \sum_{\ell \leq k^*} f_2((v_2^{\ell},1)) - f_2((v_2^1,1))$.\footnote{To see this, observe: $\pi_1((v_1^{k^*},2)) = f_2((v_2^{k^*},1))\cdot \frac{f_2((v_2^{k^*},1))-f_2((v_2^1,1))}{f_2((v_2^{k^*},1))}+\sum_{\ell <k^*} f_2((v_2^{\ell},1)) = f_2((v_2^{k^*},1))-f_2((v_2^1,1)) +\sum_{\ell < k^*} f_2((v_2^\ell,1)) = \sum_{\ell \leq k^*} f_2((v_2^{\ell},1)) - f_2((v_2^1,1))$.}
\end{itemize}

In particular, it is easy to conclude the following:
\begin{align}
\text{For all $k \notin \{1,k^*\}$: } \pi_1((v_1^k,2)) &= \pi_1((v_1^k,1)).\label{eq:equal}\\
\pi_1((v_1^1,2)) &> \pi_1((v_1^1,1)).\label{eq:1isgood}\\
\pi_1((v_1^1,2)) + \pi_1((v_1^{k^*},2)) &= \pi_1((v_1^1,1)) + \pi_1((v_1^{k^*},1)).\label{eq:kstargood}
\end{align}

Equation~\eqref{eq:utility} combined with Equations~\eqref{eq:equal} and~\eqref{eq:1isgood} implies that $u((v_1^k,2),(v_1^k,2)) > u((v_1^k,1),(v_1^k,1))$ for all $k \leq k^*$. Combined instead with Equations~\eqref{eq:equal} and~\eqref{eq:kstargood}, it implies that $u((v_1^k,2),(v_1^k,2))= u((v_1^k,1),(v_1^k,1))$ for all $k > k^*$. This completes the proof.
\end{proof}

\section{Omitted Proof of Proposition~\ref{prop:singlesimple}}
\label{app:onebidder}

\begin{proof}[Proof of Proposition~\ref{prop:singlesimple}]
We will show that the single-bidder auction which sets a take-it-or-leave-it price of $n^2 + 1$ on one-day shipping witnesses optimality for $(\alpha^M, \lambda^M)$, where $(\alpha^M, \lambda^M)$ is the canonical Myerson flow (see Definition~\ref{canonical}).

We first confirm that the single-bidder auction which sets a take-it-or-leave-it price of $n^2 + 1$ on one-day shipping is BIC for $D_1$ and $D_2$. Observe that this mechanism has a monotone allocation rule and satisfies the payment identity, so no bidder has incentive to to misreport their value but truthfully report their interest. Furthermore, the auctioneer in this mechanism always chooses the one-day shipping. Also, allocation probability and payment only depend on the value. Therefore, no bidder can benefit by misreporting their type. Thus the first two bullets of Definition~\ref{def:witness} hold.

To see the final bullet holds for $D_1$ and $D_2$, observe that the mechanism awards the item to the only bidder if and only if their value is larger than or equal to $n^2 + 1$. Since $n^2 + 1$ is the smallest possible non-zero value for the bidder (see the type space defined in Section~\ref{sec:reduction}), the final bullet holds if and only if: all virtual values of non-zero types are non-negative (which directly follows from the last bullet of Proposition~\ref{prop:Myerson}).

This concludes the proof of Proposition~\ref{prop:singlesimple}.

 \end{proof}

\section{Bidders with Public Budgets}\label{sec:budget}

\subsection{Linear Programs for Public Budget Constraints} \label{sec:budgetprelimfull}

The revenue-optimal BIC auction for the public budget setting is the following LP. In the LP, the variables are $X, P, \pi, p$. $X$ and $P$ refer to the ex-post allocation/price rules, as defined in Section~\ref{sec:fedexprelim}. $\pi, p$ refer to the interim allocation/price rules, which satisfy the equalities in Equations~\eqref{eq:budgetinterim} and~\eqref{eq:budgetprice}.

\begin{align}
\max_{X,P,\pi, p} \quad & \sum_{i} \sum_{j = 1}^{n_i} f_i(v^j_i)\cdot p_i(v^j_i)\nonumber\\
\textrm{subject to} \quad & X_i(v^j_1,v^\ell_2) \in [0,1]\text{ for all bidders $i$ and all $j,\ell$}.\nonumber\\
&X_1(v_1^0,v_2^\ell) = P_1(v_1^0,v_2^\ell) =0\text{ for all $\ell \in [0,n_2]$}.\nonumber\\
&X_2(v_1^k,v_2^0) = P_2(v_1^k,v_2^0) =0\text{ for all $k \in [0,n_1]$.}\nonumber\\
&X_1(v_1^j, v_2^\ell) + X_2(v_1^j,v_2^\ell) \leq 1\text{ for all $j \in [n_1],\ell \in [n_2]$.}\nonumber\\
\quad & \pi_i(v^j_i)=\sum_{\ell=1}^{n_{3-i}}f_{3-i}(v^\ell_{3-i}) \cdot X_i(v^j_i; v^\ell_{3-i}) \text{ for all bidders $i$ and $j \in [0,n_i]$ }. \label{eq:budgetinterim}\\
 \quad & p_i(v^j_i)=\sum_{\ell=1}^{n_{3-i}}f_{3-i}(t^\ell_{3-i}) \cdot P_i(v^j_i; v^\ell_{3-i})\text{ for all bidders $i$ and $j \in [0,n_i]$}.\label{eq:budgetprice}\\
 & \pi_i(v_i^j) \cdot v^j_i - p_i(v_i^j) \ge \pi_i(v_i^{j'}) \cdot v^j_i - p_i(v_i^{j'}) \text{ for all bidders $i$ and $j,j' \in [0,n_i]$}\label{eq:budgetbic}\\
 & \pi_1(v_1^{n_1}) \cdot B \ge p_1(v_1^{n_1}), \label{eq:budgetrespecting} 
\end{align}

The objective is simply the expected revenue. Constraints~\eqref{eq:budgetinterim} and~\eqref{eq:budgetprice} simply confirm that the interim rules are computed correctly. Equation~\eqref{eq:budgetbic} guarantees the auction is BIC. Equation~\eqref{eq:budgetrespecting} guarantees that the mechanism is budget-respecting. 

To see the correctness of Equation~\eqref{eq:budgetrespecting}, consider the following payment rule: whenever Bidder One receives the item with type $v_1^j$, they will pay $p_1(v_1^j)/\pi_1(v_1^j)$. Observe first that this payment rule satisfies Equation~\eqref{eq:budgetprice}. Second, observe that interim individual rationality guarantees that $p_1(v_1^j) \leq \pi_1(v_1^j)\cdot v_1^j$, so therefore $p_1(v_1^j)/\pi_1(v_1^j) \leq v_1^j$, and the auction is ex-post IR. Moreover, Equation~\eqref{eq:budgetrespecting} implies that $p_1(v_1^j)/\pi_1(v_1^j) \leq B$ (because $p_1(v_1^j)/\pi_1(v_1^j)$ is monotone non-decreasing in $j$ for any BIC auction), so it is also ex-post budget-respecting. Therefore, any auction satisfying~\eqref{eq:budgetbic} and~\eqref{eq:budgetrespecting} is BIC and can be implemented as ex-post IR. To see the other direction, observe that if bidder one pays at most $B$ whenever they win the item (necessary for ex-post IR), then we must have Equation~\eqref{eq:budgetrespecting}.

We'll refer to this particular ex-post payment rule as ``canonical,'' and restrict attention to auctions that have this particular ex-post payment rule, as the above paragraph shows it is without loss.

\begin{definition}[Canonical Ex-Post Payments] We say that $(X,P)$ has \emph{canonical ex-post payments} if for all $i,j_1,j_2$: $X_i(v_1^{j_1},v_2^{j_2})/P(v_1^{j_1},v_2^{j_2}) = \pi_i(v_i^{j_i})/p_i(v_i^{j_1})$.
\end{definition}

\subsection{Lagrangian Duality}\label{sec:budgetdualityprelim}

The purpose of this section is to build up Lagrangian dual we need for the public budget setting.

\begin{enumerate}
      \item[(i)] For constraint $\pi_1(v^{n_1 - 1}_1) \cdot v^{n_1 - 1}_1 - p_1(v^{n_1 - 1}_{1}) \ge \pi_1(v^{n_1}_1) \cdot v^{n_1 - 1}_1 - p_1(v^{n_1}_1)$, we use a Lagrangian multiplier of $\alpha$.
      \item[(ii)] For constraint $\pi_1(v_1^{n_1}) \cdot B \ge p_1(v_1^{n_1})$, we use a Lagrangian multiplier of $\gamma$.
    \item[(iii)] For constraints of the form: $\pi_i(v^k_i) \cdot v^{k}_i - p_i(v^k_i) \ge \pi_i(v^{k-1}_i) \cdot v^{k}_i - p_i(v^{k-1}_i)$, we use a Lagrangian multiplier of $\lambda_i(k)$ (for all bidders $i$, and $k \in [1,n_i]$).
\item [(iv)] For all remaining BIC constraints, we use a Lagrangian multiplier of $0$.
\item [(v)] To emphasize: for all other constraints (i.e. all the constraints which are unrelated to BIC), we don't use Lagrangian multipliers, and keep them as constraints.
\end{enumerate}

Every choice of Lagrangian multipliers $(\alpha,\lambda)$ induces a Lagrangian relaxation with objective function:
\begin{align*}
 \mathcal{L}(\alpha, \gamma ,\lambda)&:=\sum_i \sum_{j=1}^{n_i} f_i(t_i^j) \cdot p_i(v_i^j)+ \alpha \cdot \left ( \pi_1(v^{n_1 - 1}_1) \cdot v^{n_1 - 1}_1 - p_1(v^{n_1 - 1}_{1}) - \pi_1(v^{n_1}_1) \cdot v^{n_1 - 1}_1 - p_1(v^{n_1}_1) \right) \\
 & + \gamma \cdot \left ( \pi_1(v_1^{n_1}) \cdot B - p_1(v_1^{n_1}) \right) + \sum_i \sum_{k=1}^{n_i} \lambda_i(k) \cdot \left( \pi_i(v_i^k,j) \cdot v^k_i - p_i(v_i^k) - \pi_i(v_i^{k-1}) \cdot v^k_i + p_i(v_i^{k-1})\right).
\end{align*}

Now we define the flow in this setting:

\begin{definition}[Flow]\label{def:budgetflow} A set of Lagrangian multipliers form a \emph{flow} if the following hold for all $i$:
\begin{itemize}
\item $f_i(v_i^{k}) +\lambda_i(k+1)  = \lambda_i(k)$, for all $k \in [1,n_i-2]$.
\item $f_1(v_1^{n_1 - 1}) + \lambda_1(n_1)  = \lambda_1(n_1 - 1) + \alpha$.
\item $f_1(v_1^{n_1}) + \alpha   = \lambda_1(n_1) + \gamma$.
\end{itemize}
\end{definition}

\begin{definition}[Virtual Values]\label{def:budgetvv} For a given set of Lagrangian multipliers $\alpha, \gamma, \lambda$, define:\footnote{For simplicity of notation, denote by $\lambda_i(n_i+1) := 0$, $v_i^{n_i + 1}:=v_i^{n_i}$.}

\begin{align*}
\Phi_{1}^{\alpha, \gamma, \lambda}(v_1^k) & :=v_1^k - \frac{(v_1^{k+1}-v_1^k)\cdot \lambda_1(k+1)}{f_1(v_i^k)}, \text{ for all $k < n_1$}. \\
\Phi_{1}^{\alpha, \gamma, \lambda}(v_1^{n_1}) & :=v_1^{n_1} - \frac{-\alpha + \gamma \cdot (v_1^{n_1} - B)}{f_1(v_i^{n_1})}. \\
\Phi_{2}^{\alpha, \gamma, \lambda}(v_2^k) & :=v_2^k - \frac{(v_2^{k+1}-v_2^k)\cdot \lambda_2(k+1)}{f_2(v_i^k)}. \\
\end{align*}
\end{definition}

\begin{observation}[\cite{CaiDW16}]\label{obs:budgetvv} For any $(\alpha, \gamma, \lambda)$ which form a flow: $\mathcal{L}(\alpha, \gamma, \lambda)= \sum_i \sum_{j=1}^{n_i} f_i(v_i^j) \cdot \pi_i(v_i^j) \cdot \Phi_i^{\alpha, \gamma ,\lambda}(v_i^j).$
\end{observation}


\begin{definition}[Witness Optimality]\label{def:budgetwitness} Let $(\alpha, \gamma, \lambda)$ be a flow and $(X,P)$ be a BIC auction such that:
\begin{itemize}
\item $p$ satisfies the payment identity for $\pi$.
\item $(X, P)$ has canonical ex-post payments.
\item $\alpha > 0 \Rightarrow \pi_1(v^{n_1 - 1}_1) \cdot v^{n_1 - 1}_1 - p_1(v^{n_1 - 1}_{1}) = \pi_1(v^{n_1}_1) \cdot v^{n_1 - 1}_1 - p_1(v^{n_1}_1)$.
\item $\gamma > 0  \Rightarrow \pi_1(v_1^{n_1}) \cdot B = p_1(v_1^{n_1})$.
\item On all $(t_1^k,t_2^{k'})$, $X$ awards the item to a bidder with highest non-negative virtual value.
\end{itemize}
Then we say that $(\alpha, \gamma ,\lambda)$ \emph{witnesses optimality} for $(X,P)$, and $(X,P)$ \emph{witnesses optimality} for $(\alpha, \gamma ,\lambda)$.
\end{definition}

\begin{theorem}[\cite{CaiDW16, DevanurW17}]\label{thm:CDWbudget} Let $(\alpha, \gamma, \lambda)$ witness optimality for $(X,P)$. Then $(X,P)$ is a revenue-optimal BIC auction. Moreover, all revenue-optimal auctions witness optimality for $(\alpha, \gamma, \lambda)$.
\end{theorem}

We now define the type space. Throughout this section, $n$ denotes the size of the input to \disj.\\

\noindent\textbf{The type space.} In our reduction, the type space does not depend on $x, y$ (only the distribution does). For every input, $n_1 = n_2 = n + 2$ (meaning that each bidder has a total of $n+2$ non-zero types). For all $k \in [n + 1]$ and both $i$, $v_i^k := n^2+k$. We then define  $v_1^{n+2} := n^2 + n + 2$, and $v_2^{n+2} := n^2 + n + 1.9$.

Now we define bidder one's public budget. Note that this budget is invariant throughout all inputs $x, y$.

\begin{definition}[Bidder One's Budget] Define bidder one's public budget $B$ as follows:  
\[
B:= \frac{1}{2n} + \frac{a}{b} \sum_{i =1}^{n + 1} v_1^i. 
\]
\end{definition}

\noindent\textbf{The distribution.} The distribution in our construction depends on $x,y$, but in all cases is nearly uniform. Below, for simplicity of notation let $b := 10n^8, a := \frac{b - 1}{n+1}$. All probabilities will be an integer multiple of $\frac{1}{b}$. 

\begin{definition}[Bidder One's  distribution] Define $f_1(v_1^k)$ (as a function of $x$) as follows:
\begin{enumerate}
    \item Set $f_1(v_1^{n+2}) := \frac{1}{b} $, for all $x$.
    \item For $k=1$ to $n $, define helper $z_{k} = \frac{b - 1 - \sum_{j = 1}^{k-1} f_1(v_1^j) \cdot b}{n-k+2}$.
\begin{itemize}
\item If $x_k=0$, then set $f_1(v_1^{k}) := \frac{\left \lfloor z_{k} + \frac{n^3}{n - k + 2}  \right \rfloor}{b}$.
\item Otherwise ($x_k=1$), set $f_1(v_1^{k}):= \frac{\left \lceil z_{k} \right \rceil}{b}$.
\end{itemize}
    \item For $k=n$, define helper $z_{n+1} := b - 1 - \sum_{j = 1}^{n} f_1(v_1^{j}) \cdot b$. Set $f_1(v_1^{n+1}) := \frac{z_{n+1}}{b}$.
\end{enumerate}

\end{definition}

The two lemmas below establish that the total mass of Bidder One is always $1$, and that Bidder One's distribution is always nearly-uniform over $v_1^1,\ldots, v_1^{n+1}$.

\begin{lemma}\label{lem:budget1mass}
$\sum_{k=1}^{n+2} f_1(v_1^k) = 1$.
\end{lemma}

\begin{lemma} \label{budgetdrange}
For all $x$, $f_1(v_1^k) \cdot b \in [a - n^3, a + n^3]$ for all $k \in [1,n+1]$.
\end{lemma}

Finally, we define the distribution for Bidder Two. Bidder Two's distribution will be truly single-parameter since they does not have a budget. Bidder Two's distribution depends on $y$, and is constructed so that $y_k$ has a significant impact on $f_2(v_2^k)$.

\begin{definition}[Bidder Two's distribution] Define $f_2(v_2^k)$ (as a function of $y$) as follows:
\begin{enumerate}
    \item Set $f_2(v_2^{n+2}) := \frac{n^4}{b}$.
    \item For $k=1$ to $n$, define helper $z_{k} := \frac{b - n^4 - \sum_{j = 1}^{k - 1} f_2(v_2^j) \cdot b}{n - k + 2}$.
\begin{itemize}
\item If $y_k = 1$, then set $f_2(v_2^k) :=  \frac { \left \lfloor z_k + \frac{n^2}{n-k+2} \right  \rfloor} {b}$.
\item Otherwise, set $f_2(v_2^k) := \frac{\left \lfloor z_k - 1\right  \rfloor}{b} $.
\end{itemize}
    \item For $k=n+1$, define helper $z_{n+1} := b - n^4 - \sum_{j = 1}^{n} f_2(v_2^{j}) \cdot b$, set $f_2(v_2^{n+1}) := \frac{z_{n+1}} {b}$.
\end{enumerate}

\end{definition}

The two lemmas below similarly establish that the total mass of Bidder Two is always $1$, and that Bidder Two's distribution is always nearly-uniform over $v_2^1,\ldots, v_2^{n+1}$.

\begin{lemma}\label{lem:budget2mass}
$\sum_{k=1}^{n+2} f_2(v_2^k) = 1$.
\end{lemma}

\begin{lemma} \label{budgeterange}
For all $y$, $f_2(v_2^k) \cdot b \in [a - 2n^3, a + 2n^3]$ for all $k \in [1,n+1]$.
\end{lemma}

Finally, we quickly state that the optimal \emph{single-bidder} auction for any distribution considered in our reduction is especially simple: it sets the same take-it-or-leave-it price of $n^2+1$.

\begin{proposition}\label{prop:budgetsinglesimple} For all $x$ (resp., $y$), the revenue-optimal single-bidder auction for the resulting distribution $D_1$ (resp. $D_2$) simply sets a take-it-or-leave-it price of $n^2+1$.
\end{proposition}

\section{Constructing a Flow for Budget Constraints} \label{sec:budgetflow}
In this section, we construct a flow which is optimal for all instances of our construction. We proceed in two steps. First, we consider a canonical flow and establish that this flow is optimal if and only if $\disj(x,y)=\mathsf{yes}$. Next, we show how to modify the flow to be optimal when $\disj(x,y) = \mathsf{no}$.

\subsection{A Canonical Flow}
We first define a canonical flow, and then argue it is optimal when $\disj(x,y) = \mathsf{yes}$.

\begin{definition}[Canonical Flow] $(\alpha^M,\gamma^M,\lambda^M)$ is the canonical Myerson flow, where: 
\label{budgetcanonical}
\begin{itemize}
\item $\gamma = 0$.
\item $\alpha = 0$.
\item $\lambda_i(k) = R_i(v_i^k)$ for both bidders $i$, and all $k \in [n+2]$.
\end{itemize}
\end{definition}

It is easy to confirm that $(\alpha^M,\gamma^M,\lambda^M)$ is a flow. We can also quickly execute Definition~\ref{def:budgetvv} to compute $\Phi^{\alpha^M,\gamma^M,\lambda^M}$:

\begin{observation} 
For both bidders $i$, and all $k \in [1,n+2]$, $\Phi^{\alpha^M,\gamma^M,\lambda^M}_i(v_i^k)= v_i^k - \frac{R_i(v_i^{k+1})}{f_i(v_i^k)}$.
\end{observation}

Proposition~\ref{prop:budgetMyerson} below captures the key properties of our construction and this flow. These properties are motivated by bullet five of Definition~\ref{def:budgetwitness}: we need to compare virtual values of types of Bidder One with those for types of Bidder Two to determine if a certain allocation is optimal. The proof is in Appendix~\ref{app:budgetflow}.

\begin{proposition}\label{prop:budgetMyerson}
For all $x,y$, the flow $(\alpha^M,\gamma^M,\lambda^M)$ satisfies the following:
\begin{itemize}
\item $ k > k' \Rightarrow \Phi^{\alpha^M,\gamma^M,\lambda^M}_1(v_1^k) > \Phi^{\alpha^M,\gamma^M,\lambda^M}_2(v_2^{k'})$.
\item $k < k' \Rightarrow \Phi^{\alpha^M,\gamma^M,\lambda^M}_1(v_1^k) < \Phi^{\alpha^M,\gamma^M,\lambda^M}_2(v_2^{k'})$.
\item For all $k \in [n]$: if $x_{k} = 0 \text{ OR } y_{k} = 0$, then $\Phi^{\alpha^M,\gamma^M,\lambda^M}_1(v_1^{k}) > \Phi^{\alpha^M,\gamma^M,\lambda^M}_2(v_2^{k})$.
\item For all $k \in [n]$, if $x_k = 1 \text{ AND } y_k = 1$, then $\Phi^{\alpha^M,\gamma^M,\lambda^M}_2(v_2^{k}) > \Phi^{\alpha^M,\gamma^M,\lambda^M}_1(v_1^{k})$.
\item $\Phi^{\alpha^M,\gamma^M,\lambda^M}_1(v_1^{n+1}) > \Phi^{\alpha^M,\gamma^M,\lambda^M}_2(v_2^{n+1})$.
\item $\Phi^{\alpha^M,\gamma^M,\lambda^M}_1(v_1^{n+2}) > \Phi^{\alpha^M,\gamma^M,\lambda^M}_2(v_2^{n+2})$.
\item For both $i$, and all $k \geq 1$: $\Phi_i^{\alpha^M,\gamma^M,\lambda^M}(v_i^k)>0$.
\end{itemize}
\end{proposition}

The first two bullets assert that a bidder with strictly higher value has strictly higher virtual value as well. The next two bullets concern virtual values when both bidders' values are the same. Importantly, they assert that the relative comparison of virtual values when both bidders have value $n^2+k$ depends \emph{only on $x_k$ and $y_k$, and not on $x_{-k}$ or $y_{-k}$}. Bullet seven implies that any allocation rule that witnesses optimality for $(\alpha^M,\gamma^M,\lambda^M)$ must learn which bidder has higher virtual value. We analyze a potential such auction next.

\begin{definition}[Second-Price Auction, tie-breaking for Bidder One]\label{def:budgetSPA} The \emph{second-price auction, tie-breaking for Bidder One}, gives the item to the bidder with highest value and breaks ties in favor of Bidder One. Payments are charged to satisfy the payment identity, and to be ex-post canonical.
\end{definition}

\begin{lemma}\label{lem:budgetrespectingdisjoint}
The second-price auction, tie-breaking for Bidder One is budget-respecting for Bidder One.
\end{lemma}

\begin{proof}
By the definition of payment identity, and Lemma~\ref{budgeterange}, the payment $P_1(v_1^{n+2})$ can be calculated as follows: 
\begin{align*}
P_1(v_1^{n+2}) & = \sum_{k = 1}^{n+2} v_1^k \cdot (\pi_1(v_1^k) - \pi_1((v_1^{k-1})))\\
& = \sum_{k = 1}^{n+2} v_1^k \cdot f_2(v_2^k) \\
& = v_1^{n+2} \cdot f_2(v_2^{n+2}) + \sum_{k = 1}^{n+1} v_1^k \cdot f_2(v_2^k) \\
& < \frac{v_1^{n+2}  \cdot n^4}{b} + \sum_{k = 1}^{n+1} v_1^k \cdot \frac{a + 2n^3}{b}  \\
& < \frac{3n^6}{b} + \frac{2n^3(n^2+2)(n+1)}{b} + \frac{a}{b}\sum_{k = 1}^{n+1} v_1^k \\
& < B. 
\end{align*} 
The last inequality comes from the fact that $b = \Theta(n^8)$.  We can then easily see that the second-price auction, tie-breaking for Bidder One, with the payment identity is budget-respecting for Bidder One since $\pi_1(v_1^{n+2}) = 1$.
\end{proof}

\begin{theorem}\label{thm:budgetdisjoint} 
The second-price auction, tie-breaking for Bidder One, witnesses optimality for $(\alpha^M, \gamma^M, \lambda^M)$ \emph{if and only if} $\disj(x,y) = \mathsf{yes}$.
\end{theorem}

\begin{proof}

First, it is well-known (and easy to see) that the second-price auction, tie-breaking for Bidder One, with the payment identity is BIC (and budget-respecting for Bidder One, by Lemma~\ref{lem:budgetrespectingdisjoint}). After this, there are five bullets to check in Definition~\ref{def:budgetwitness}. We claim that the first two hold for all $x,y$, and the third holds if and only if $\disj(x,y) = \mathsf{yes}$.

The first bullet holds trivially, as payments are specifically defined to satisfy the payment identity. The second bullet also holds trivially, as payments are specified to be ex-post canonical.

The third and fourth bullet hold vacuously, as $\alpha = \gamma = 0$.

To see that final bullet holds if and only if $\disj(x,y) = \mathsf{yes}$, observe that the second-price auction awards the item to the bidder with highest value, tie-breaking for Bidder One. So the final bullet holds if and only if: (a) a higher value implies a higher virtual value (which immediately follows from the first two bullets of Proposition~\ref{prop:budgetMyerson}), (b) all virtual values are non-negative (which immediately follows from bullet seven of Proposition~\ref{prop:budgetMyerson}), and (c) Bidder One has a higher virtual value whenever both bidders have the same value (which holds \emph{if and only if} $\disj(x,y) = \mathsf{yes}$, by bullets three and four of Proposition~\ref{prop:budgetMyerson}).

This completes the proof: Definition~\ref{def:budgetwitness} is satisfied if and only if $\disj(x,y) = \mathsf{yes}$.
\end{proof}

Observe that Theorem~\ref{thm:budgetdisjoint} implies that the Second-Price Auction, tie-breaking for Bidder One is one optimal auction for $(D_1,D_2)$ when $\disj(x,y) = \mathsf{yes}$. We now conclude the following simple corollary:

\begin{corollary}\label{cor:budgetdisjoint} If $\disj(x,y) = \mathsf{yes}$, \emph{every} optimal BIC auction $(X,P)$ for $D_1,D_2$ has $X_1(v_1^{n+2}, v_2^{n+2}) = 1$.
\end{corollary}
\begin{proof}
Because a BIC auction witnesses optimality for $(\alpha^M, \gamma^M, \lambda^M)$ (by Theorem~\ref{thm:budgetdisjoint}), every optimal BIC auction for $D_1,D_2$ witnesses optimality for $(\alpha^M, \gamma^M, \lambda^M)$. Because $\Phi^{\alpha^M, \gamma^M , \lambda^M}_1(v_1^{n+2})>\Phi^{\alpha^M, \gamma^M, \lambda^M}_2(v_2^{n+2})$ by Proposition~\ref{prop:budgetMyerson}, bullet five of Definition~\ref{def:budgetwitness} asserts that every optimal BIC auction satisfies $X_1(v_1^{n+2},v_2^{n+2}) = 1$.
\end{proof}

Corollary~\ref{cor:budgetdisjoint} proves half of Theorem~\ref{thm:mainbudget}: that Bidder One wins the item in all optimal auctions on $(v_1^{n+2},v_2^{n+2})$ when $\disj(x,y) = \mathsf{yes}$. Theorem~\ref{thm:budgetdisjoint} makes clear the key distinction when $\disj(x,y) = \mathsf{no}$: $(\alpha^M, \gamma^M ,\lambda^M)$ does not witness optimality, so we need a new flow.

\subsection{Modifying the Canonical Flow}
We now modify the canonical flow to find an optimal $(\alpha', \gamma', \lambda')$ in the case when $\disj(x,y) = \mathsf{no}$. Fortunately, the necessary modification is simple to describe (although verifying the desired properties is comlpex). We will make two modifications, defined below, and also used in~\cite{DevanurW17}.

\begin{definition}[Boosting, \cite{DevanurW17}] Beginning with a flow $(\alpha, \gamma, \lambda)$, \emph{boosting $(\alpha,  \gamma, \lambda)$ for Bidder One by $\varepsilon$}, produces a new flow $(\alpha', \gamma' ,\lambda')$ with:
\begin{itemize}
\item  $\gamma' :=  \gamma + \varepsilon$.
\item  $(\lambda')_1(k):=\lambda_1(k) - \varepsilon$, for all $k \leq n + 2$.
\item If not already specified, then $\alpha' = \alpha$ and $\lambda'=\lambda$.
\end{itemize}
\end{definition}

It is not hard to see that Boosting preserves the flow conditions (provided that $\varepsilon \leq \lambda_1(k)$ for all $k \leq n + 2$). It is also not hard to see that Boosting for Bidder One \emph{increases} the virtual value for all $v_1^{k}$ for all $k <  n + 2$, \emph{decreases} the virtual value for $v_1^{n+2}$, and leaves all other virtual values unchanged.

The second operation is ironing.

\begin{definition}[Ironing] Beginning with a flow $(\alpha, \gamma, \lambda)$, \emph{ironing $(\alpha,  \gamma, \lambda)$ for Bidder One by $\delta$},  produces a new flow $(\alpha', \gamma' ,\lambda')$ with:
\begin{itemize}
\item  $\alpha := \delta$.
\item  $(\lambda')_1(n+2):=\lambda_1(n+2) + \delta$.
\item If not already specified, then $\gamma' = \gamma$ and $\lambda'=\lambda$.
\end{itemize}
\end{definition}

\begin{definition}[Modified Flow] The modified flow $(\alpha^*, \gamma^*,\lambda^*)$ proceeds as follows:
\begin{enumerate}
\item Begin with $(\alpha, \gamma, \lambda) =(\alpha^M, \gamma^M, \lambda^M)$.
\item Boost $(\alpha, \gamma, \lambda)$ for Bidder One by $\varepsilon$. Here, $\varepsilon$ is the minimum boost which results in $\Phi_1^{\alpha',  \gamma', \lambda'}(v_1^k)\geq \Phi_2^{\alpha', \gamma',\lambda'}(v_2^k)$ for all $k \in [1, n]$.
\item Iron $(\alpha, \gamma, \lambda)$ for Bidder One by $\delta$. Here $\delta$ is the ironing which results in  $\Phi_1^{\alpha',  \gamma', \lambda'}(v_1^{n+1}) = \Phi_1^{\alpha',  \gamma', \lambda'}(v_1^{n+2})$.
\end{enumerate}
\end{definition}

Note that $\lambda_1(n+2)$ will be negative after boosting (step two), but the ironing in step three will ensure that $\lambda_1(n+2)$ becomes non-negative again.

The rest of our analysis proceeds as follows. First, we need to establish that this modified flow indeed exists, because the required boost for Bullet 2 is small enough to be valid. Proposition~\ref{prop:budgetvalid} states this, and also several useful properties of this flow. The proof of Proposition~\ref{prop:budgetvalid} is in Appendix~\ref{app:budgetflow}, and this relies on many of the precise choices in defining our instance.

\begin{proposition} \label{prop:budgetvalid}
 For all $x,y$, $(\alpha^*, \gamma^* ,\lambda^*)$ is a valid flow. Moreover, it satisfies the following properties:
\begin{itemize}
\item $k > k' \Rightarrow \Phi^{\alpha^*,\gamma^*,\lambda^*}_1(v_1^k) > \Phi^{\alpha^*,\gamma^*,\lambda^*}_2(v_2^{k'})$.
\item $k < k' \Rightarrow \Phi^{\alpha^*,\gamma^*,\lambda^*}_1(v_1^k) < \Phi^{\alpha^*,\gamma^*,\lambda^*}_2(v_2^{k'})$.
\item For all $k \in [n]$:  $\Phi^{\alpha^*,\gamma^*,\lambda^*}_1(v_1^{k}) \geq  \Phi^{\alpha^*,\gamma^*,\lambda^*}_2(v_2^{k})$.
\item When $\disj(x,y) = \mathsf{no}$, there exists a $k^* \in [1,n]$ such that: $\Phi_1^{\alpha^*, \gamma^*, \lambda^*}(v_1^{k^*}) = \Phi_2^{\alpha^*, \gamma^*, \lambda^*}(v_2^{k^*})$.
\item When $\disj(x,y) = \mathsf{no}$, then: $\Phi^{\alpha^*, \gamma^* ,\lambda^*}_2(v_2^{n+2}) > \Phi^{\alpha^*, \gamma^*, \lambda^*}_1(v_1^{n +2}) = \Phi^{\alpha^*, \gamma^*, \lambda^*}_1(v_1^{n + 1}) > \Phi^{\alpha^*,\lambda^*}_2(v_2^{n+1})$.
\item For both $i$, and all $k \geq 1$: $\Phi_i^{\alpha^*,\gamma^*,\lambda^*}(v_i^k)>0$.

 \end{itemize}
\end{proposition}

We now define an auction that witnesses optimality for $(\alpha^*, \gamma*, \lambda^*)$, and conclude implications for $X_1(v_1^{n+2},v_2^{n+2})$.

For convenience, we define 
\[ \Delta =  \left ( 1 - \frac{n^4}{b} \right) \cdot B -  \sum_{i = 1}^{n+1} f_2(v_2^i) \cdot v_1^i. \]

 Note that $\sum_{i = 1}^{n+1} f_2(v_2^i) \cdot v_1^i$ is exactly the payment of $v_1^{n+2}$ when the auction is awarding the item to the bidder with highest virtual value, tie-breaking always for bidder one. The idea here is that we'd like to carefully breaks ties for type $k*$ to increase the payment for  $v_1^{n+2}$ to make $\Delta$ vanish. 

The following lemma gives a bound of $\Delta$.

\begin{lemma} 
If $\disj(x,y) = \mathsf{no}$, $\Delta = \frac1{2n} + O(1/n^2)$.
\end{lemma}

\begin{proof}
By Proposition~\ref{prop:budgetvalid} and Lemma~\ref{budgeterange} we have
\begin{align*}
\Delta & =   \left ( 1 - \frac{n^4}{b} \right) \cdot B -  \sum_{i = 1}^{n+1} f_2(v_2^i) \cdot v_1^i\\
& = B - \frac{n^4 \cdot B}{b} -  \frac{a}{b} \sum_{i = 1}^{n+1} v_1^i  -  \sum_{i = 1}^{n+1} v_1^i \cdot \left (f_2(v_2^i) - \frac{a}{b} \right )\\
& = \frac1{2n} - \Theta(1/n^2) + \Theta(1/n^2) \\
& = \frac1{2n} + O(1/n^2).
\end{align*}
\end{proof}

\begin{definition}[Second-Price Auction, Except for $n+2$, careful tie-breaking at $k^*$]\label{def:budgetcareful} The second-price auction except for $n+2$ with careful tie-breaking at $k^* \in [1,n]$ gives the item to the bidder with highest value (with an exception of always giving the item to Bidder Two when her value is $v_2^{n+2}$). If both bidders have the same value $n^2+k$, break ties in the following manner (in all cases, charge payments satisfying the payment identity, and to be ex-post canonical):
\begin{itemize}
\item If $k \neq k^*$, give the item to Bidder One.
\item If $k = k^*$, give Bidder One the item with probability $1-\frac{\Delta}{f_2(v_2^{k^*})}$, and to Bidder Two with probability $\frac{\Delta}{f_2(v_2^{k^*})}$.\footnote{Observe that this is feasible, as we've guaranteed that $\Delta =  \frac1{2n} + O(1/n^2) < f_2(v_2^k)$ for all $k \in [1, n]$.}
\end{itemize}
\end{definition}

\begin{lemma} \label{lem:budgetBIC}
For all $x,y$ such that $\disj(x,y) = \mathsf{no}$, the Second-Price Auction, Except for $n+2$ with careful tie-breaking at $k^*$ is BIC. Moreover, $\pi_1(v_1^{n+2}) \cdot B =  P_1(v_1^{n+2})$.
\end{lemma}

With Lemma~\ref{lem:budgetBIC} in hand, the proof of Theorem~\ref{thm:budgetnotdisjoint} follows similarly to that of Theorem~\ref{thm:budgetdisjoint}.

\begin{theorem}\label{thm:budgetnotdisjoint} When $\disj(x,y) = \mathsf{no}$, let $k^*$ be the index promised by bullet four of Proposition~\ref{prop:budgetvalid}. Then the second-price auction except for $n+2$ with careful tie-breaking at $k^*$ witnesses optimality for $(\alpha^*, \gamma^*, \lambda^*)$.
\end{theorem}

\begin{proof}
First of all, the auction always awards the item to a bidder with highest non-negative virtual value by Proposition~\ref{prop:budgetvalid}. Furthermore, we have already established in Lemma~\ref{lem:budgetBIC} that the second-price auction with careful tie-breaking at $k^*$ is BIC. So we just need to check three bullet. We have also explicitly defined payments to be ex-post canonical and satisfy the payment identity, so first two bullets are satisfied. For bullet four, the budget equality holds by Lemma~\ref{lem:budgetBIC}. Finally, bullet three is easy to see, since $\pi_1(v_1^{n+1}) = \pi_1(v_1^{n+2}) = 1 - f_2(v_2^{n+2})$, and $P_1(v_1^{n+1}) =  P_1(v_1^{n+2})$.
\end{proof}

Again observe that Theorem~\ref{thm:budgetnotdisjoint} implies that the Second-Price Auction Except for $n+2$ with careful tie-breaking at $k^*$ is one optimal auction for $(D_1,D_2)$ when $\disj(x,y) = \mathsf{no}$. We again conclude the following corollary:

\begin{corollary}\label{cor:budgetnotdisjoint} If $\disj(x,y)=\mathsf{no}$, every optimal BIC auction $(X,P)$ for $D_1,D_2$ has $X_1(v_1^{n+2},v_2^{n+2}) = 0$.
\end{corollary}
\begin{proof}
Because a BIC auction witnesses optimality for $(\alpha^*, \gamma^* ,\lambda^*)$ (by Theorem~\ref{thm:budgetnotdisjoint}), every optimal BIC auction for $D_1,D_2$ witnesses optimality for $(\alpha^*, \gamma^* ,\lambda^*)$. Because $\Phi_1^{\alpha^*, \gamma^*,\lambda^*}(v_1^{n+2}) < \Phi_2^{\alpha^*,\gamma^* ,\lambda^*}(v_2^{n+2})$ by Proposition~\ref{prop:budgetvalid}, bullet five of Definition~\ref{def:budgetwitness} asserts that every optimal BIC auction satisfies $X_1(v_1^{n+2},v_2^{n+2}) = 0$.
\end{proof}

This wraps up the proof of Theorem~\ref{thm:mainbudget}.
\begin{proof}[Proof of Theorem~\ref{thm:mainbudget}]
Corollary~\ref{cor:budgetdisjoint} establishes that when $\disj(x,y) = \mathsf{yes}$, any optimal BIC auction must allocate the item to Bidder One on $(v_1^{n+2},v_2^{n+2})$ with probability one. Corollary~\ref{cor:budgetnotdisjoint} establishes that when $\disj(x,y) = \mathsf{no}$, any optimal BIC auction must allocate the item to Bidder One on $(v_1^{n+2},v_2^{n+2})$ with probability zero. Because $D_1$ can be constructed only as a function of $x$, and $D_2$ can be constructed only as a function of $y$, any communication protocol which correctly allocates the item on $(v_1^{n+2},v_2^{n+2})$ in accordance with \emph{any} optimal BIC mechanism (even with probability $2/3$) can also solve \disj\ (with probability $2/3$). Because any deterministic (resp. randomized, succeeding with probability $2/3$) protocol for disjointness requires communication $n$ (resp. $\Omega(n)$), this means that any deterministic (resp. randomized, succeeding with probability $2/3$) protocol which can correctly allocate the item on $(v_1^{n+2},v_2^{n+2})$ in accordance with any optimal BIC mechanism (resp. with probability $2/3$) requires communication at least $n$ (resp. $\Omega(n)$).
\end{proof}

\section{Omitted Proofs from Appendix~\ref{sec:budget}}\label{app:budgetproof}

\subsection{Bidder 1}

In order to distinguish different helpers defined in the constructions and improve readability, we will use the following notations throughout this appendix: 
\begin{itemize}
  \item We use $d^k$ to represent type $v_1^k$. Notation $z^d_k$ represents the helper $z_k$ for Bidder One. 
  \item We use $e^k$ to represent type $v_2^k$. Notation $z^e_k$ represents the helper $z_k$ for Bidder Two. 
\end{itemize}

Lemmas below provide some useful properties of the helpers $z^d$s, at the end of this subsection, we will use these to prove the distribution we constructed is valid and nearly-uniform.

The following lemma shows the relationship between $f(d^i)$ and $z^d_{i+1}$.
\begin{lemma} \label{budgetfzd}
 $f(d^i) \cdot b \ge z^d_{i + 1}$ for all $i \in [1, n]$.
\end{lemma}
\begin{proof}
   If $x_i = 1$, recall that $f(d^{i}) = \left \lceil z^d_{i} \right \rceil /b$, in the other case, by definition we know $f(d^i) = \left \lfloor z^d_{i} + \frac{n^3}{n - i + 2} \right \rfloor /b$. In both cases it is clear that $f(d^{i}) \cdot b \ge z^d_{i}$, thus we have
  \begin{align*}
        z^d_{i + 1} & = \frac{b - 1 - \sum_{j = 1}^{i} f(d^j) \cdot b} {n - i + 1} \\
        & = \frac{b - f(d^{i})\cdot b - 1 -  \sum_{j = 1}^{i-1} f(d^j) \cdot b} {n - i + 1} \\
        & = \frac{z^d_{i} \cdot (n - i + 2) - f(d^{i}) \cdot b}{n - i + 1} \\
        & \le \frac{f(d^{i}) \cdot b \cdot (n - i + 2) - f(d^{i}) \cdot b}{n - i + 1} \\
        & = f(d^{i}) \cdot b.
  \end{align*}
\end{proof}

The following lemmas bound the gap between two consecutive $z^ds$.

\begin{lemma} \label{budgetzd1gap}
  If $x_i = 0$, $z^d_{i} - z^d_{i + 1} \le \frac{n^3}{(n - i + 2)(n - i + 1)}$.
\end{lemma}

\begin{proof}
\begin{align*}
        z^d_{i + 1} & = \frac{b - 1 - \sum_{j = 1}^{i} f(d^j) \cdot b} {n - i + 1} \\
        & = \frac{b - 1 - f(d^{i})\cdot b - \sum_{j = 1}^{i-1} f(d^j) \cdot b} {n - i + 1} \\
        & = \frac{z^d_{i} \cdot (n - i + 2) - f(d^{i}) \cdot b}{n - i + 1} \\
        & = z^d_{i} + \frac{z^d_{i} - f(d^{i}) \cdot b}{n - i + 1} \\
        & = z^d_{i} + \frac{z^d_{i} -  \left \lfloor  z^d_{i} + \frac{n^3}{n - i + 2}   \right \rfloor}{n - i + 1} \\
        & \ge z^d_{i} + \frac{z^d_{i} -  \left (  z^d_{i} + \frac{n^3}{n - i + 2}  \right )}{n - i + 1} \\
        & = z^d_{i} - \frac{n^3}{(n - i + 2)(n - i + 1)}\\
\end{align*}
\end{proof}

\begin{lemma} \label{budgetzd0gap}
  If $x_i = 1$, $z^d_{i } - z^d_{i + 1} \le \frac{1}{n - i + 1}$.
\end{lemma}

\begin{proof}

\begin{align*}
        z^d_{i + 1} & = \frac{b -  1  - \sum_{j = 1}^{i} f(d^j) \cdot b} {n - i + 1} \\
        & = \frac{b - 1 - f(d^{i})\cdot b - \sum_{j = 1}^{i - 1} f(d^j) \cdot b} {n - i + 1} \\
        & = \frac{z^d_{i} \cdot (n - i + 2) - f(d^{i}) \cdot b}{n - i + 1} \\
        & = z^d_{i} + \frac{z^d_{i} - f(d^{i}) \cdot b}{n - i + 1} \\
        & = z^d_{i} + \frac{z^d_{i} -  \left \lceil z^d_{i}   \right \rceil}{n - i + 1} \\
        & \ge z^d_{i} + \frac{z^d_{i} -  \left (  z^d_{i} + 1  \right )}{n - i + 1} \\
        & = z^d_{i} - \frac{1}{n - i + 1}\\
\end{align*}
\end{proof}

\begin{lemma} \label{budgetzdmonotone}
$z^d_{i + 1} < z^d_{i}$ for all $i \in [1, n]$.
\end{lemma}

\begin{proof}
  If $x_i = 1$, then we have
\begin{align*}
        z^d_{i + 1} & = \frac{b - 1 - \sum_{j = 1}^{i} f(d^j) \cdot b} {n - i + 1} \\
        & = \frac{b - 1 - f(d^{i })\cdot b - \sum_{j = 1}^{i-1} f(d^j) \cdot b} {n - i + 1} \\
        & = \frac{z^d_{i} \cdot (n - i + 2) - f(d^{i}) \cdot b}{n - i + 1} \\
        & = z^d_{i} + \frac{z^d_{i} - f(d^{i}) \cdot b}{n - i + 1} \\
        & = z^d_{i} + \frac{z^d_{i} -  \left \lfloor  z^d_{i} + \frac{n^3}{n - i + 2}   \right \rfloor}{n - i + 1} \\
        & < z^d_{i} + \frac{z^d_{i} -  \left (  z^d_{i } + \frac{n^3}{n - i + 2} - 1  \right )}{n - i + 1} \\
        & = z^d_{i} - \frac {\frac{n^3}{n - i + 2} - 1 }{n - i + 1} \\
        & < z^d_{i}.
\end{align*}
Otherwise, we have
\begin{align*}
        z^d_{i + 1} & = \frac{b - 1 - \sum_{j = 1}^{i} f(d^j) \cdot b} {n - i + 1} \\
        & = \frac{b - 1 - f(d^{i})\cdot b - \sum_{j = 1}^{i-1} f(d^j) \cdot b} {n - i + 1} \\
        & = \frac{z^d_{i} \cdot (n - i + 2) - f(d^{i}) \cdot b}{n - i + 1} \\
        & = z^d_{i} + \frac{z^d_{i} - f(d^{i}) \cdot b}{n - i + 1} \\
        & = z^d_{i} + \frac{z^d_{i} -  \left \lceil  z^d_{i}  \right \rceil}{n - i + 1} \\
        & < z^d_{i} + \frac{z^d_{i} -  z^d_{i}  }{n - i + 1} \\
        & = z^d_{i},
\end{align*}
which concludes the proof.
\end{proof}

The following lemma bounds the range of $z^d$.
\begin{lemma}\label{budgetzdrange}
$z^d_{i} \in [a - n^3, a]$, for all $i \in [1, n + 1]$.
\end{lemma}
\begin{proof}
 Since by Lemma \ref{budgetzdmonotone} $z^d$ is monotone decreasing, and by definition $z^d_1 = a$, we only need to bound $z^d_{n+1}$:
\begin{align*}
    z^d_{n + 1} & = z^d_1 + \sum_{i = 1}^{n} (z^d_{i + 1} - z^d_{i})  \\
    & = z^d_1 +  \sum_{i = 1}^{n} \frac{z^d_{i} \cdot (n-i+2) - f(d^{i}) \cdot b}{n - i + 1} - z^d_{i} \\
    & = z^d_1 +  \sum_{i = 1}^{n} \frac{z^d_{i} - f(d^{i}) \cdot b}{n - i + 1} \\
    & = z^d_1 + \sum_{i = 1}^n \frac{z^d_{i} - \left \lfloor  z^d_{i} + \frac{n^3}{n - i + 2}   \right \rfloor \cdot [x_i = 0] - \lceil z^d_{i}\rceil \cdot [x_i = 1]}{n-i+1}   \\
    & \ge z^d_1 + \sum_{i = 1}^n \frac{z^d_{i} - \left \lfloor  z^d_{i} + \frac{n^3}{n - i + 2}   \right \rfloor }{n-i+1}   \\
    & \ge z^d_1 + \sum_{i = 1}^n \frac{z^d_{i} - \left ( z^d_{i} + \frac{n^3}{n - i + 2}   \right )}{n-i+1} \\
    & = z^d_1 - \sum_{i = 1}^n \frac{n^3} {(n-i+1)(n-i+2)} \\
    & = z^d_1 -  n^3 \sum_{i = 1}^n\frac1 {(n-i+1)(n-i+2)} \\
    &  = z^d_1 - n^3 \left (\frac{1}{(n + 1)n} + \frac{1}{n(n - 1)} + \ldots + \frac{1}{1\cdot2} \right) \\
    &  = z^d_1 - n^3 \left [ \left (\frac{1}{n} - \frac{1}{n+1} \right) + \left (\frac{1}{n-1} - \frac{1}{n} \right) + \ldots +  \left(1 - \frac12\right) \right] \\
    & = z^d_1 - n^3 \left(1-\frac1{n+1} \right) \\
    & > a - n^3,
\end{align*}
where the last inequality comes from the fact that  $z^d_1 = a$.
\end{proof}

The following proof proves Lemma~\ref{lem:budget1mass}, which shows our distribution is valid.
\begin{proof}[Proof of Lemma~\ref{lem:budget1mass}]
\begin{align*}
\sum_{k=1}^{n+2} f(d^k) &= f(d^{n+1}) + f(d^{n+2}) + \sum_{k=1}^{n} f(d^k)  \\
&= \frac{ b - 1 - \sum_{k = 1}^{n} f(d^k) \cdot b }{b} + f(d^{n+2}) + \sum_{k=1}^{n} f(d^k)  \\
&= 1
\end{align*}
\end{proof}

Similarly, utilizing the properties above, we conclude this subsection by proving a key property of our construction. It guarantees that Bidder One's distribution is nearly-uniform over $v_1^1,\ldots, v_1^{n+1}$.

\begin{proof}[Proof of Lemma~\ref{budgetdrange}]

By \ref{budgetzdrange} we have
\begin{align*}
    f(d^{i}) \cdot b & = \left \lceil  z^d_{i } \right \rceil \cdot [x_i = 1] + \left \lfloor  z^d_{i } + \frac{n^3}{n - i + 2}   \right \rfloor \cdot [x_i = 0] \\
    & \ge \left \lceil  z^d_{i } \right \rceil \\
    & \ge z^d_{i } \\
    & \ge a - n^3.
\end{align*}
On the other hand,  we have
\begin{align*}
     f(d^{i}) \cdot b & = \left \lceil  z^d_{i } \right \rceil \cdot [x_i = 1] + \left \lfloor  z^d_{i } + \frac{n^3}{n - i + 2}   \right \rfloor \cdot [x_i = 0] \\
    & \le \left \lfloor  z^d_{i} + \frac{n^3}{n - i + 2}   \right \rfloor \\
    & \le z^d_{i} + \frac{n^3}{n - i + 2} \\
    & \le z^d_{i} + n^3 \\
    & \le a + n^3,
\end{align*}
which concludes the proof.

\end{proof}

\subsection{Bidder 2}

Lemmas below provide some useful properties of the helpers $z^e$s, at the end of this subsection, we will use these to prove the distribution we constructed is valid and nearly-uniform.

The following lemma shows the relationship between $f(e^{i})$ and $z^e_{i+1}$ when $y_i = 0$.
\begin{lemma} \label{budgetfze}
If $y_i = 0, f(e^{i}) \cdot b < z^e_{i + 1}$ for all $i \in [1, n]$.
\end{lemma}
\begin{proof}
   Recall that $f(e^{i}) = \left \lfloor z^e_{i} - 1 \right \rfloor /b$ when $y_i = 0$, then we have
  \begin{align*}
        z^e_{i + 1} & = \frac{b  - n^4 - \sum_{j = 1}^{i} f(e^j) \cdot b} {n - i + 1} \\
        & = \frac{b - n^4 - f(e^{i})\cdot b - \sum_{j = 1}^{i-1} f(e^j) \cdot b} {n - i + 1} \\
        & = \frac{z^e_{i} \cdot (n - i + 2) - f(e^{i}) \cdot b}{n - i + 1} \\
        & > \frac{f(e^{i}) \cdot b \cdot (n - i + 2) - f(e^{i}) \cdot b}{n - i + 1} \\
        & = f(e^{i}) \cdot b
  \end{align*}
\end{proof}

The following lemmas bound the gap between two consecutive $z^es$.

\begin{lemma} \label{budgetze1gap}
  If $y_i = 1$, $z^e_{i} - z^e_{i + 1} \le \frac{n^2}{(n - i + 2)(n - i + 1)}$.
\end{lemma}

\begin{proof}
\begin{align*}
        z^e_{i + 1} & = \frac{b - n^4 - \sum_{j = 1}^{i} f(e^j) \cdot b} {n - i + 1} \\
        & = \frac{b - f(e^{i})\cdot b - n^4 - \sum_{j = 1}^{i-1} f(e^j) \cdot b} {n - i + 1} \\
        & = \frac{z^e_{i} \cdot (n - i + 2) - f(e^{i}) \cdot b}{n - i + 1} \\
        & = z^e_{i} + \frac{z^e_{i} - f(e^{i}) \cdot b}{n - i + 1} \\
        & = z^e_{i} + \frac{z^e_{i} -  \left \lfloor  z^e_{i} + \frac{n^2}{n - i + 2}   \right \rfloor}{n - i + 1} \\
        & \ge z^e_{i} + \frac{z^e_{i} -  \left (  z^e_{i} + \frac{n^2}{n - i + 2}  \right )}{n - i + 1} \\
        & = z^e_{i} - \frac{n^2}{(n - i + 2)(n - i + 1)}\\
\end{align*}
\end{proof}

\begin{lemma} \label{budgetze0gap}
  If $y_i = 0$, $z^e_{i} - z^e_{i + 1} \le -\frac{1}{n-i+1}$.
\end{lemma}

\begin{proof}
\begin{align*}
        z^e_{i + 1} & = \frac{b - n^4 - \sum_{j = 1}^{i} f(e^j) \cdot b} {n - i + 1} \\
        & = \frac{b - n^4 -  f(e^{i})\cdot b - \sum_{j = 1}^{i - 1} f(e^j) \cdot b} {n - i + 1} \\
        & = \frac{z^e_{i} \cdot (n - i + 2) - f(e^{i}) \cdot b}{n - i + 1} \\
        & = z^e_{i} + \frac{z^e_{i} - f(e^{i}) \cdot b}{n - i + 1} \\
        & = z^e_{i} + \frac{z^e_{i} -  \left \lfloor  z^e_{i} - 1  \right \rfloor}{n - i + 1} \\
        & \ge z^e_{i} + \frac{z^e_{i} -  \left (  z^e_{i}  - 1   \right )}{n - i + 1} \\
        & = z^e_{i} + \frac{1}{(n - i + 1)}
\end{align*}
\end{proof}

\begin{lemma} \label{budgetzenearlymonotone}
$z^e_{i + 1} < z^e_{i} + 2$ for all $i \in [1, n]$.

\end{lemma}

\begin{proof}
  If $y_i = 1$, then we have
\begin{align*}
        z^e_{i + 1} & = \frac{b - n^4 - \sum_{j = 1}^{i} f(e^j) \cdot b} {n - i + 1} \\
        & = \frac{b - n^4 - f(e^{i})\cdot b - \sum_{j = 1}^{i - 1} f(e^j) \cdot b} {n - i + 1} \\
        & = \frac{z^e_{i} \cdot (n - i + 2) - f(e^{i}) \cdot b}{n - i + 1} \\
        & = z^e_{i} + \frac{z^e_{i} - f(e^{i}) \cdot b}{n - i + 1} \\
        & = z^e_{i} + \frac{z^e_{i} -  \left \lfloor  z^e_{i } + \frac{n^2}{n - i + 2}   \right \rfloor}{n - i + 1} \\
        & < z^e_{i} + \frac{z^e_{i} -  \left (  z^e_{i} + \frac{n^2}{n - i + 2} - 1  \right )}{n - i + 1} \\
        & = z^e_{i} - \frac {\frac{n^2}{n - i + 2} - 1 }{n - i + 1} \\
        & < z^e_{i}.
\end{align*}
Otherwise, we have
\begin{align*}
        z^e_{i + 1} & = \frac{b -  n^4 - \sum_{j = 1}^{i} f(e^j) \cdot b} {n - i + 1} \\
        & = \frac{b - n^4 - f(e^{i })\cdot b - \sum_{j = 1}^{i-1} f(e^j) \cdot b} {n - i + 1} \\
        & = \frac{z^e_{i} \cdot (n - i + 2) - f(e^{i}) \cdot b}{n - i + 1} \\
        & = z^e_{i} + \frac{z^e_{i} - f(e^{i}) \cdot b}{n - i + 1} \\
        & = z^e_{i} + \frac{z^e_{i} -  \left \lfloor  z^e_{i }  \right \rfloor + 1}{n - i + 1} \\
        & < z^e_{i} + \frac{z^e_{i} -  z^e_{i } + 2  }{n - i + 1} \\
        & = z^e_{i} + \frac{2}{n - i + 1} \\
        & \le z^e_{i} + 2,
\end{align*}
which concludes the proof.
\end{proof}

\begin{lemma}\label{budgetzerange}
$z^e_{i} \in [a - n^3 - n, a]$, for all $i \in [1, n + 1]$.
\end{lemma}
\begin{proof}
  Since by Lemma~\ref{budgetzenearlymonotone}, $z^e_{i + 1} < z^e_{i} + 2$, which indicates $z^e_{n+1} - 2n < z^e_{i} < z^e_1 + 2n$ for all $i \in [1, n + 1]$, and by definition $z^e_1 = \frac{b - n^4}{n + 1} = a - \frac{n^4 -1}{n+1}$, we only need to lower bound $z^e_{n+1}$. Again by the fact that $z^e_1 = a - \frac{n^4 -1}{n+1} $ we have
\begin{align*}
    z^e_{n + 1} & = z^e_1 + \sum_{i = 1}^{n} (z^e_{i + 1} - z^e_{i})  \\
    & = z^e_1 +  \sum_{i = 1}^{n} \frac{z^e_{i} \cdot (n-i+2) - f(e^{i}) \cdot b}{n - i + 1} - z^e_{i} \\
    & = z^e_1 +  \sum_{i = 1}^{n} \frac{z^e_{i} - f(e^{i}) \cdot b}{n - i + 1} \\
    & = z^e_1 + \sum_{i = 1}^n \frac{z^e_{i} - \left \lfloor  z^e_{i} + \frac{n^2}{n - i + 2}   \right \rfloor \cdot [y_i = 1] - \lfloor z^e_{i} - 1 \rfloor \cdot [y_i = 0]}{n-i+1}   \\
    & \ge z^e_1 + \sum_{i = 1}^n \frac{z^e_{i} - \left \lfloor  z^e_{i} + \frac{n^2}{n - i + 2}   \right \rfloor }{n-i+1}   \\
    & \ge z^e_1 + \sum_{i = 1}^n \frac{z^e_{i} - \left ( z^e_{i} + \frac{n^2}{n - i + 2}   \right )}{n-i+1} \\
    & = z^e_1 - \sum_{i = 1}^n \frac{n^2} {(n-i+1)(n-i+2)} \\
    & = z^e_1 -  n^2\sum_{i = 1}^n\frac1 {(n-i+1)(n-i+2)} \\
    &  = z^e_1 - n^2 \left (\frac{1}{(n + 1)n} + \frac{1}{n(n - 1)} + \ldots + \frac{1}{1\cdot2} \right) \\
    &  = z^e_1 - n^2 \left [ \left (\frac{1}{n} - \frac{1}{n+1} \right) + \left (\frac{1}{n-1} - \frac{1}{n} \right) + \ldots +  \left(1 - \frac12\right) \right] \\
    & = z^e_1 - n^2 \left(1-\frac1{n+1} \right) \\
    & > a - n^3 - n,
\end{align*}
which concludes the proof.

\end{proof}

The following proof proves Lemma~\ref{lem:budget2mass} which ensures our distribution is valid.

\begin{proof}[Proof of Lemma~\ref{lem:budget2mass}]
\begin{align*}
\sum_{k=1}^{n+2} f(e^k) &= f(e^{n+1}) + f(d^{n+2})  + \sum_{k=1}^{n} f(e^k)  \\
&= \frac{b -  f(d^{n+2})  - \sum_{k=1}^{n}  f(e^k) \cdot b}{b} +  f(d^{n+2})  + \sum_{k=1}^{n+1} f(e^k) \\
&= 1.
\end{align*}
\end{proof}

Again, utilizing the properties above, we conclude this subsection by proving a key property of our construction. It guarantees that Bidder Two's distribution is nearly-uniform over $v_1^1,\ldots, v_1^{n+1}$.

\begin{proof}[Proof of Lemma~\ref{budgeterange}]

By Lemma~\ref{budgetzerange} we have
\begin{align*}
    f(e^{i}) \cdot b & = \left \lfloor  z^e_{i} - 1 \right \rfloor \cdot [y_i = 0] + \left \lfloor  z^e_{i + 1} + \frac{n^2}{n - i + 2}   \right \rfloor \cdot [y_i = 1] \\
    & \ge \left \lfloor  z^e_{i} - 1 \right \rfloor \\
    & > z^e_{i} - 2 \\
    & \ge a - n^3 - n  - 2.
\end{align*}
On the other hand, we have
\begin{align*}
     f(e^{i}) \cdot b & = \left \lceil  z^e_{i } \right \rceil \cdot [y_i = 0] + \left \lfloor  z^e_{i} + \frac{n^2}{n - i + 2}   \right \rfloor \cdot [y_i = 1] \\
    & \le \left \lfloor  z^e_{i } + \frac{n^2}{n - i + 2}   \right \rfloor \\
    & \le z^e_{i} + \frac{n^2}{n - i + 2} \\
    & \le z^e_{i } + n^2 \\
    & \le a + n^2,
\end{align*}
which concludes the proof.

\end{proof}

\begin{proof}[Proof of Proposition~\ref{prop:budgetsinglesimple}]
We will show that the single-bidder auction which sets a take-it-or-leave-it price of $n^2 + 1$ witnesses optimality for $(\alpha^M, \gamma^M , \lambda^M)$, where $(\alpha^M, \gamma^M ,\lambda^M)$ is the canonical Myerson flow (see Definition~\ref{budgetcanonical}).

We first confirm that the single-bidder auction which sets a take-it-or-leave-it price of $n^2 + 1$  is BIC for $D_1$ and $D_2$. Observe that this mechanism has a monotone allocation rule and satisfies the payment identity, so no bidder has incentive to to misreport their value. Furthermore, Bidder One always pays $n^2 + 1$ (which is less than her budget $B$) if and only if her value is at least as large as $n^2 + 1$.  Therefore, this mechanism is ex-post individually rational. Thus the first two bullets of Definition~\ref{def:budgetwitness} hold.

To see the final bullet holds for $D_1$ and $D_2$, observe that the mechanism awards the item to the only bidder if and only if their value is larger than or equal to $n^2 + 1$. Since $n^2 + 1$ is the smallest possible non-zero value for the bidder (see the type space defined in Section~\ref{sec:budget}), the final bullet holds if and only if: all virtual values of non-zero types are non-negative (which directly follows from the last bullet of Proposition~\ref{prop:budgetMyerson}).
\end{proof}

\section{Omitted Proofs from Appendix~\ref{sec:budgetflow}}\label{app:budgetflow}

\subsection{Analyzing the canonical flow}

Here we are going to use properties proved in Appendix~\ref{app:budgetproof} to analyze the canonical flow. First we will present some technical lemmas which asymptotically bound the gaps between virtual values. By comparing the orders of those gaps, we then wrap them up to obtain Proposition~\ref{prop:budgetMyerson}.

The following lemmas establish some useful bounds for Myerson virtual values in our reduction.

\begin{lemma} \label{budgetd1bound}
   If $x_i = 0$, $ v^{i} - (n - i + 1) + \frac{\frac12 n^3 - n }{a + 3n^3} \le \Phi^{\alpha^M,\gamma^M,\lambda^M}(d^{i}) \le v^{i} - (n - i + 1) + \frac{2n^3}{a - 2n^3}$ for all $i \in [1,n]$  under the canonical flow.
\end{lemma}

\begin{proof}
  Recall that, by definition $z^d_{i} = \frac{b - 1 - \sum_{j = 1}^{i - 1} f(d^j) \cdot b}{n - i + 2} = \frac{\sum_{j = i}^{n + 2} f(d^j) \cdot b}{n - i + 2}$.
  Let us first lower bound $\Phi^{\alpha^M,\gamma^M,\lambda^M}(d^{i})$. By Lemma \ref{budgetzdmonotone} and \ref{budgetzdrange} we have the following inequality,
\begin{align*}
    \Phi^{\alpha^M,\gamma^M,\lambda^M}(d^{i}) & = v^{i} - \frac{\sum_{k=i+1}^{n+2} f(d^k)}{f(d^{i})} \\
    & = v^{i} - (n - i + 1)\frac{z^d_{i+1}}{f(d^{i})\cdot b} \\
    & = v^{i} - (n - i + 1)\frac{z^d_{i+1}}{\left \lfloor z^d_{i} + \frac{n^3}{n-i+2}\right \rfloor} \\
    & \ge v^{i} - (n - i + 1)\frac{z^d_{i+1}}{z^d_{i} + \frac{n^3}{n-i+2} - 1}  \\
    & \ge v^{i} - (n - i + 1)\frac{z^d_{i} }{z^d_{i} + \frac{n^3}{n-i+2} - 1}  \\
    & = v^{i} - (n - i + 1) \left ( 1 - \frac{\frac{n^3}{n-i+2} - 1 }{z^d_{i} + \frac{n^3}{n-i+2} - 1} \right) \\
    & = v^{i} - (n - i + 1) + (n - i + 1) \frac{\frac{n^3}{n-i+2} -1 }{z^d_{i} + \frac{n^3}{n-i+2} - 1} \\
    & \ge v^{i} - (n - i + 1) + \frac{\frac12 n^3 - n }{z^d_{i} + \frac{n^3}{n-i+2}} \\
    & \ge v^{i} - (n - i + 1) + \frac{\frac12 n^3 - n }{a + 3n^3}.
\end{align*}
On the other hand, by Lemma \ref{budgetzd1gap} and \ref{budgetzdrange} we have
\begin{align*}
    \Phi^{\alpha^M,\gamma^M,\lambda^M}(d^{i}) & = v^{i} - \frac{\sum_{k=i+1}^{n+2} f(d^k)}{f(d^{i})} \\
    & = v^{i} - (n - i + 1)\frac{z^d_{i+1}}{f(d^{i})\cdot b} \\
    & = v^{i} - (n - i + 1)\frac{z^d_{i+1}}{\left \lfloor z^d_{i} + \frac{n^3}{n-i+2}\right \rfloor} \\
    & \le v^{i} - (n - i + 1)\frac{z^d_{i+1}}{z^d_{i} + \frac{n^3}{n-i+2} }  \\
    & \le v^{i} - (n - i + 1)\frac{z^d_{i} - \frac{n^3}{(n - i + 1)(n - i + 2)} }{z^d_{i} + \frac{n^3}{n-i+2} }  \\
    & = v^{i} - (n - i + 1) \left ( 1 - \frac{\frac{n^3}{n-i+2} + \frac{n^3}{(n - i + 1)(n - i + 2)}}{z^d_{i} + \frac{n^3}{n-i+2}} \right) \\
    & = v^{i} - (n - i + 1) + (n - i + 1) \frac{\frac{n^3}{n-i+2} + \frac{n^3}{(n - i + 1)(n - i + 2) }}{z^d_{i} + \frac{n^3}{n-i+2}} \\
    & \le v^{i} - (n - i + 1) + (n - i + 1) \frac{\frac{n^3}{n-i+1} + \frac{n^3}{(n - i + 1)(n - i + 2) }}{z^d_{i}} \\
    & \le v^{i} - (n - i + 1) + \frac{n^3 + \frac{n^3}{n - i + 2 }}{z^d_{i}} \\
    & \le v^{i} - (n - i + 1) + \frac{2n^3}{a - 2n^3}.
\end{align*}
\end{proof}

\begin{lemma} \label{budgetd0bound}
  If $x_i = 1$, $ v^{i} - (n - i + 1) \le \Phi^{\alpha^M,\gamma^M,\lambda^M}(d^{i}) \le v^{i} - (n - i + 1) + \frac{2n}{a - 2n^3}$ for all $i \in [1,n]$  under the canonical flow.
\end{lemma}

\begin{proof}
  Recall that, by definition $z^d_{i} = \frac{b - 1 - \sum_{j = 1}^{i - 1} f(d^j) \cdot b}{n - i + 2} = \frac{\sum_{j = i}^{n + 2} f(d^j) \cdot b}{n - i + 2}$.
  Let us first lower bound $\Phi^{\alpha^M,\gamma^M,\lambda^M}(d^{i})$. By Lemma \ref{budgetzdmonotone} and \ref{budgetzdrange} we have the following inequality,
\begin{align*}
    \Phi^{\alpha^M,\gamma^M,\lambda^M}(d^{i}) & = v^{i} - \frac{\sum_{k=i+1}^{n+2} f(d^k)}{f(d^{i})} \\
    & = v^{i} - (n - i + 1)\frac{z^d_{i+1}}{f(d^{i})\cdot b} \\
    & = v^{i} - (n - i + 1)\frac{z^d_{i+1}}{\left \lceil z^d_{i} \right \rceil} \\
    & \ge v^{i} - (n - i + 1)\frac{z^d_{i+1}}{z^d_{i}}  \\
    & \ge v^{i} - (n - i + 1)\frac{z^d_{i} }{z^d_{i}}  \\
    & \ge v^{i} - (n - i + 1).
\end{align*}
On the other hand, by Lemma \ref{budgetzd0gap} and \ref{budgetzdrange} we have
\begin{align*}
    \Phi^{\alpha^M,\gamma^M,\lambda^M}(d^{i}) & = v^{i} - \frac{\sum_{k=i+1}^{n+2} f(d^k)}{f(d^{i})} \\
    & = v^{i} - (n - i + 1)\frac{z^d_{i+1}}{f(d^{i})\cdot b} \\
    & = v^{i} - (n - i + 1)\frac{z^d_{i+1}}{\left \lceil z^d_{i} \right \rceil} \\
    & \le v^{i} - (n - i + 1)\frac{z^d_{i+1}}{z^d_{i} + 1 }  \\
    & \le v^{i} - (n - i + 1)\frac{z^d_{i} - \frac{1}{n - i + 1} }{z^d_{i} + 1 }  \\
    & = v^{i} - (n - i + 1) \left (1 - \frac{1 + \frac{1}{n - i + 1}}{z^d_{i} + 1} \right) \\
    & = v^{i} - (n - i + 1) + (n - i + 1) \frac{1 + \frac{1}{n - i + 1 }}{z^d_{i} + 1} \\
    & \le v^{i} - (n - i + 1) + (n - i + 1) \frac{1 + \frac{1}{n - i + 1 }}{z^d_{i}} \\
    & \le v^{i} - (n - i + 1) + \frac{n - i + 2}{z^d_{i}} \\
    & \le v^{i} - (n - i + 1) + \frac{2n}{a - 2n^3}.
\end{align*}
\end{proof}

\begin{lemma} \label{budgete1bound}
  If $y_i = 1$, $ v^{i} - (n - i + 1) + \frac{\frac12 n^2 - 3n }{a + 3n^3} \le \Phi^{\alpha^M,\gamma^M,\lambda^M}(e^{i}) \le v^{i} - (n - i + 1) + \frac{2n^2}{a - 2n^3}$ for all $i \in [1,n]$  under the canonical flow.
\end{lemma}

\begin{proof}
  Recall that, by definition $z^e_{i} = \frac{b - 2 - \sum_{j = 1}^{i - 1} f(e^j) \cdot b}{n - i + 2} = \frac{\sum_{j = i}^{n + 2} f(e^j) \cdot b}{n - i + 2}$.
  Let us first lower bound $\Phi^{\alpha^M,\gamma^M,\lambda^M}(e^{i})$. By Lemma \ref{budgetzenearlymonotone} and \ref{budgetzerange}, $z^e_{i + 1} < z^e_{i} + 2$ for all $i \in [1, n]$, and $z^e_{i} \le a$, for all $i \in [1, n + 1]$, we have the following inequality,
\begin{align*}
    \Phi^{\alpha^M,\gamma^M,\lambda^M}(e^{i}) & = v^{i} - \frac{\sum_{k=i+1}^{n+2} f(e^k)}{f(e^{i})} \\
    & = v^{i} - (n - i + 1)\frac{z^e_{i+1}}{f(e^{i})\cdot b} \\
    & = v^{i} - (n - i + 1)\frac{z^e_{i+1}}{\left \lfloor z^e_{i} + \frac{n^3}{n-i+2}\right \rfloor} \\
    & \ge v^{i} - (n - i + 1)\frac{z^e_{i+1}}{z^e_{i} + \frac{n^2}{n-i+2} - 1}  \\
    & \ge v^{i} - (n - i + 1)\frac{z^e_{i} + 2 }{z^e_{i} + \frac{n^2}{n-i+2} - 1}  \\
    & = v^{i} - (n - i + 1) \left ( 1 - \frac{\frac{n^2}{n-i+2} - 3 }{z^e_{i} + \frac{n^2}{n-i+2} - 1} \right) \\
    & = v^{i} - (n - i + 1) + (n - i + 1) \frac{\frac{n^2}{n-i+2} -3 }{z^e_{i} + \frac{n^2}{n-i+2} - 1} \\
    & \ge v^{i} - (n - i + 1) + \frac{\frac12 n^2 - 3n }{z^e_{i} + \frac{n^2}{n-i+2}} \\
    & \ge v^{i} - (n - i + 1) + \frac{\frac12 n^2 - 3n }{a + 3n^3}.
\end{align*}
On the other hand, by Lemma \ref{budgetze1gap}, $z^e_{i} - z^e_{i + 1} \le \frac{n^2}{(n - i + 2)(n - i + 1)}$, so we have
\begin{align*}
    \Phi^{\alpha^M,\gamma^M,\lambda^M}(e^{i}) & = v^{i} - \frac{\sum_{k=i+1}^{n+2} f(e^k)}{f(e^{i})} \\
    & = v^{i} - (n - i + 1)\frac{z^e_{i+1}}{f(e^{i})\cdot b} \\
    & = v^{i} - (n - i + 1)\frac{z^e_{i+1}}{\left \lfloor z^e_{i} + \frac{n^2}{n-i+2}\right \rfloor} \\
    & \le v^{i} - (n - i + 1)\frac{z^e_{i+1}}{z^e_{i} + \frac{n^2}{n-i+2} }  \\
    & \le v^{i} - (n - i + 1)\frac{z^e_{i} - \frac{n^2}{(n - i + 1)(n - i + 2)} }{z^e_{i} + \frac{n^2}{n-i+2} }  \\
    & = v^{i} - (n - i + 1) \left ( 1 - \frac{\frac{n^2}{n-i+2} + \frac{n^2}{(n - i + 1)(n - i + 2)}}{z^e_{i} + \frac{n^2}{n-i+2}} \right) \\
    & = v^{i} - (n - i + 1) + (n - i + 1) \frac{\frac{n^2}{n-i+2} + \frac{n^2}{(n - i + 1)(n - i + 2) }}{z^e_{i} + \frac{n^2}{n-i+2}} \\
    & \le v^{i} - (n - i + 1) + (n - i + 1) \frac{\frac{n^2}{n-i+1} + \frac{n^2}{(n - i + 1)(n - i + 2) }}{z^e_{i}} \\
    & \le v^{i} - (n - i + 1) + \frac{n^2 + \frac{n^2}{n - i + 2 }}{z^e_{i}} \\
    & \le v^{i} - (n - i + 1) + \frac{2n^2}{a - 2n^3}.
\end{align*}
\end{proof}

\begin{lemma} \label{budgete0bound}
  If $y_i = 0$, $ v^{i} - (n - i + 1) - \frac{4n}{a - 3n^3} \le \Phi^{\alpha^M,\gamma^M,\lambda^M}(e^{i}) \le v^{i} - (n - i + 1) - \frac{1}{a + 2n^3}$ for all $i \in [1,n]$  under the canonical flow.
\end{lemma}

\begin{proof}
  Recall that, by definition $z^e_{i} = \frac{b - 2 - \sum_{j = 1}^{i - 1} f(e^j) \cdot b}{n - i + 2} = \frac{\sum_{j = i}^{n + 2} f(e^j) \cdot b}{n - i + 2}$.
  Let us first lower bound $\Phi^{\alpha^M,\gamma^M,\lambda^M}(e^{i})$. By Lemma \ref{budgetzenearlymonotone} and \ref{budgetzerange}, $z^e_{i + 1} < z^e_{i} + 2$ for all $i \in [1, n]$, and $z^e_{i} \ge a - 2n^3$, for all $i \in [1, n + 1]$, we have the following inequality,
\begin{align*}
    \Phi^{\alpha^M,\gamma^M,\lambda^M}(e^{i}) & = v^{i} - \frac{\sum_{k=i+1}^{n+2} f(e^k)}{f(e^{i})} \\
    & = v^{i} - (n - i + 1)\frac{z^e_{i+1}}{f(e^{i})\cdot b} \\
    & = v^{i} - (n - i + 1)\frac{z^e_{i+1}}{\left \lfloor z^e_{i}  - 1\right \rfloor} \\
    & \ge v^{i} - (n - i + 1)\frac{z^e_{i+1}}{z^e_{i} - 2}  \\
    & \ge v^{i} - (n - i + 1)\frac{z^e_{i} + 2 }{z^e_{i} - 2}  \\
    & = v^{i} - (n - i + 1)\left ( 1 + \frac{4}{z^e_{i} - 2} \right )  \\
    & \ge v^{i} - (n - i + 1) - \frac{4n}{a - 3n^3}.
\end{align*}
On the other hand, by Lemma \ref{budgetze0gap} $z^e_{i} - z^e_{i + 1} \le -\frac{1}{n-i+1}$, we have
\begin{align*}
    \Phi^{\alpha^M,\gamma^M,\lambda^M}(e^{i}) & = v^{i} - \frac{\sum_{k=i+1}^{n+2} f(e^k)}{f(e^{i})} \\
    & = v^{i} - (n - i + 1)\frac{z^e_{i+1}}{f(e^{i})\cdot b} \\
    & = v^{i} - (n - i + 1)\frac{z^e_{i+1}}{\left \lfloor z^e_{i}  - 1\right \rfloor} \\
    & \le v^{i} - (n - i + 1)\frac{z^e_{i+1}}{z^e_{i} }  \\
    & \le v^{i} - (n - i + 1)\frac{z^e_{i} + \frac{1}{n - i + 1} }{z^e_{i} }  \\
    & = v^{i} - (n - i + 1) \left (1 + \frac{\frac{1}{n - i + 1}}{z^e_{i}} \right) \\
    & = v^{i} - (n - i + 1) - \frac{1}{z^e_{i}} \\
    & \le v^{i} - (n - i + 1) - \frac{1}{a + 2n^3}.
\end{align*}
\end{proof}

We now have enough tools to prove Proposition~\ref{prop:budgetMyerson}. Firstly we will prove the monotonicity of virtual values, which establishes the first two bullets of Proposition~\ref{prop:budgetMyerson}.
\begin{lemma} \label{budgetmonotone}
For all $x, y \in \{0, 1\}^n$, we have \[
\max \left (\Phi^{\alpha^M,\gamma^M,\lambda^M}(d^{i}),  \Phi^{\alpha^M,\gamma^M,\lambda^M}(e^{i})  \right) <  \min \left ( \Phi^{\alpha^M,\gamma^M,\lambda^M}(d^{i + 1}),  \Phi^{\alpha^M,\gamma^M,\lambda^M}(e^{i + 1})  \right) \] for $i \in [1, n + 1]$.
\end{lemma}

\begin{proof}
For convenience, let \[L_i = \min \left (\Phi^{\alpha^M,\gamma^M,\lambda^M}(d^{i}),  \Phi^{\alpha^M,\gamma^M,\lambda^M}(e^{i})  \right),\] and \[R_i = \max \left ( \Phi^{\alpha^M,\gamma^M,\lambda^M}(d^{i}),  \Phi^{\alpha^M,\gamma^M,\lambda^M}(e^{i})  \right).\]
We prove this lemma by showing the gap between $R_i$ and $L_{i+1}$ is at least 0.5 for all $i$ under the canonical flow.

First, $R_{n + 1} = \Phi^{\alpha^M,\gamma^M,\lambda^M}(d^{n + 1}) = n^2 + n + 1 - \frac{f(d^{n+2})}{f(d^{n+1})}$, and $L_{n+2} =  \Phi^{\alpha^M,\gamma^M,\lambda^M}(e^{n + 2}) =  n^2 + n + 1.9$. This lemma clearly holds for $i = n + 1$.

Second, $R_{n} = \Phi^{\alpha^M,\gamma^M,\lambda^M}(d^{n}) = n^2 + n - \frac{f(d^{n+2}) + f(d^{n+1})}{f(d^{n})}$, and $L_{n+1} =  \Phi^{\alpha^M,\gamma^M,\lambda^M}(e^{n + 1}) =  n^2 + n + 1 - \frac{0.9\cdot f(e^{n+2})}{f(e^{n+1})}$. This lemma clearly holds for $i = n$ by Lemma~\ref{budgetdrange} and Lemma~\ref{budgeterange}.

For $i \le n - 1$, again by  Lemma  \ref{budgetdrange} and \ref{budgeterange},  we have
\begin{align*}
    L_{i + 1} - R_{i} & \ge 1 - (n - i + 1)\frac{a + 2n^3}{a - 2n^3} + (n - i + 2)\frac{a - 2n^3}{a + 2n^3} \\
    & = 1 - (n - i + 1) \left ( 1 + \frac{4n^3}{a - 2n^3} \right) + (n - i + 2) \left ( 1 - \frac{4n^3}{a + 2n^3} \right) \\
    & = 2 -  (n - i + 1) \frac{4n^3}{a - 2n^3}  + (n - i + 2)  \frac{4n^3}{a + 2n^3} \\
    & \ge 0.5.
\end{align*}

\end{proof}

The next two lemmas will establish bullet three and four of Proposition~\ref{prop:budgetMyerson}.

\begin{lemma} \label{budgetle}
If either $x_{i} = 0$ or $y_{i} = 0$, then $\Phi^{\alpha^M,\gamma^M,\lambda^M}(d^{i}) > \Phi^{\alpha^M,\gamma^M,\lambda^M}(e^{i})$  for all $i \in [1,n]$ for the canonical flow.
\end{lemma}
\begin{proof}
By Lemma \ref{budgete0bound}, \ref{budgetd0bound} and \ref{budgetd1bound}, we know  $\Phi^{\alpha^M,\gamma^M,\lambda^M}(d^{i}) \ge v^{i} - (n - i + 1)$ whereas  $\Phi^{\alpha^M,\gamma^M,\lambda^M}(e^{i}) < v^{i} - (n - i + 1)$ when $y_i = 0$. Thus the only case left is $x_i = 0$ and $y_i = 1$, again by Lemma \ref{budgetd1bound} and \ref{budgete1bound}, we have
\begin{align*}
  \Phi^{\alpha^M,\gamma^M,\lambda^M}(d^{i}) - \Phi^{\alpha^M,\gamma^M,\lambda^M}(e^{i}) & \ge \frac{\frac12 n^3 - 3n }{a + 3n^3} -  \frac{2n^2}{a - 2n^3} \\
  & = \Omega(1/n^4),
\end{align*}
where the last equality is from the fact that $a = \Theta(n^7)$.

\end{proof}

\begin{lemma}\label{budgetge}
If $x_{i} = y_{i} = 1$, then $\Phi^{\alpha^M,\gamma^M,\lambda^M}(d^{i}) < \Phi^{\alpha^M,\gamma^M,\lambda^M}(e^{i})$ for the canonical flow. Furthermore, \[\min_{i : x_i = y_i = 1} \left (\Phi^{\alpha^M,\gamma^M,\lambda^M}(e^{i}) -\Phi^{\alpha^M,\gamma^M,\lambda^M}(d^{i})  \right ) = \Theta(1/n^5),\] and \[\max_{i: x_i = y_i = 1} \left (\Phi^{\alpha^M,\gamma^M,\lambda^M}(e^{i}) -\Phi^{\alpha^M,\gamma^M,\lambda^M}(d^{i}) \right ) = \Theta(1/n^5).\]
\end{lemma}
\begin{proof}
When $x_i = y_i = 1$, by Lemma \ref{budgetd0bound} and \ref{budgete1bound} we have
\begin{align*}
\Phi^{\alpha^M,\gamma^M,\lambda^M}(e^{i}) - \Phi^{\alpha^M,\gamma^M,\lambda^M}(d^{i}) & \ge \frac{\frac12 n^2 - 3n }{a + 3n^3} - \frac{2n}{a - 2n^3} \\
& = \Omega(1/n^5),
\end{align*}
where the last equality is from the fact that $a = \Theta(n^7)$.
On the other hand, we have
\begin{align*}
\Phi^{\alpha^M,\gamma^M,\lambda^M}(e^{i}) - \Phi^{\alpha^M,\gamma^M,\lambda^M}(d^{i}) & \le \frac{2n^2}{a - 2n^3} - 0 \\
& = O(1/n^5),
\end{align*}
where the last equality is from the fact that $a = \Theta(n^7)$.
\end{proof}

By Lemma~\ref{budgetdrange} and Lemma~\ref{budgeterange}, bullet five of Proposition~\ref{prop:budgetMyerson} is easy to veryfy:
\begin{align*}
\Phi^{\alpha^M,\gamma^M,\lambda^M}(d^{n + 1}) & = v^{n+1} - \frac{f(d^{n+2})}{f(d^{n+1})} \\
& = v^{n+1}  - \frac{1}{f(d^{n+1}) \cdot b} \\
& \ge  v^{n+1}  - \frac{1}{a - n^3} \\
& >  v^{n+1}  - \frac{n^4 \cdot 0.9}{a + n^3} \\
& \ge v^{n+1} - \frac{f(e^{n+2})}{f(e^{n+1})} \\
& = \Phi^{\alpha^M,\gamma^M,\lambda^M}(e^{n + 1}).
\end{align*}

Now we would like to show the relationship of virtual values of the highest type for the canonical flow, which establishes bullet six of Proposition~\ref{prop:budgetMyerson}.

\begin{lemma} \label{budgetbefore}
$ \Phi^{\alpha^M,\gamma^M,\lambda^M}(d^{n + 2})   > \Phi^{\alpha^M,\gamma^M,\lambda^M}(e^{n + 2})$ for the canonical flow.
\end{lemma}

\begin{proof}
This lemma is easy to see since $\Phi^{\alpha^M,\gamma^M,\lambda^M}(d^{n + 2}) = n^2+n+2 > n^2+n+1.9 = \Phi^{\alpha^M,\gamma^M,\lambda^M}(e^{n + 2})$.
\end{proof}

Finally, the following lemma shows all virtual values are positive which implies the last bullet of Proposition~\ref{prop:budgetMyerson}.
\begin{lemma} \label{budgetpositive}
  By Lemma~\ref{budgetmonotone} and Lemma~\ref{budgetbefore}, it is sufficient to show $\Phi^{\alpha^M,\gamma^M,\lambda^M}(e^1) > 0$. To this end, observe that:
  \begin{align*}
    \Phi^{\alpha^M,\gamma^M,\lambda^M}(e^1) & = v^1 - \frac{1 - f(e^1)}{f(e^1)} \\
    & > v^1 - (n + 2) \\
    & = n^2 + 1 - (n + 2) \\
    & > 0.
  \end{align*}
\end{lemma}

\subsection{Analyzing the Modified Flow}
First we bound the flow we need for the boosting operation:
\begin{lemma} \label{boundofflow}
When the modified flow is needed, $\varepsilon = \max \left (\frac{\Phi^{\alpha^M,\gamma^M,\lambda^M}(e^{i}) - \Phi^{\alpha^M,\gamma^M,\lambda^M}(d^{i})}{f(d^i)} \right) = \Theta(1/n^6)$.
\end{lemma}
\begin{proof}
This is directly from Lemma \ref{budgetdrange} and \ref{budgetge}.
\end{proof}

Next we give the order of $\delta$ for the ironing operation:
\begin{lemma} \label{boundofironing}
When the modified is needed, $\delta = \Theta(1/n^5)$.
\end{lemma}

\begin{proof}
After boosting, the virtual value of $d^{n + 2}$ : $\Phi(d^{n+2}) = v_1^{n+2} - \frac{\varepsilon \cdot (v_1^{n+2} - B)}{f(d^{n+2})}$ and the virtual value of $d^{n + 1}$ : $\Phi(d^{n+1}) = v_1^{n+1} + \frac{\varepsilon - f(d^{n+2})}{f(d^{n+1})}$.
Recall that we want to set $\delta$ such that $ \Phi^{\alpha^*,\gamma^*,\lambda^*}(d^{n + 1}) = \Phi^{\alpha^*,\gamma^*,\lambda^*}(d^{n + 2})$. By solving this equation, we have 
\begin{align*}
\delta &= \frac{1}{\frac{1}{f(d^{n+1})} + \frac{1}{f(d^{n+2})}} \left (\frac{\varepsilon \cdot(v_1^{n+2} - B)}{f(d^{n+2})} + \frac{\varepsilon}{f(d^{n+1})} - \frac{f(d^{n+2})}{f(d^{n+1})}  - 1\right)\\
& = \Theta(1/n^8) \cdot \left ( \Theta(n^3) + \Theta(1/n^5) - \Theta(1/n^7) - 1 \right) \\
& = \Theta(1/n^5).
\end{align*}
\end{proof}

After knowing the order of flow we need, we can compare it with the bounds for gaps between virtual values we have already established. We will use it and all properties we proved above to establish Proposition~\ref{prop:budgetvalid}. 

First we can now easily see that the modified flow $(\alpha^*, \gamma^*, \lambda^*)$ is a valid flow, since $\varepsilon = \Theta(1/n^6) < \lambda^M_1(k) $ for $k \in [1, n + 1]$. And for $n + 2$, $\lambda^*_1(n + 2) = f(d^{n+2}) - \varepsilon + \delta \ge 0$.

The following lemma shows bullet five of Proposition~\ref{prop:budgetvalid}.

\begin{lemma} \label{budgetafter}
When $\disj(x,y) = \mathsf{no}$, then: $\Phi^{\alpha^*, \gamma^* ,\lambda^*}(e^{n+2}) > \Phi^{\alpha^*, \gamma^*, \lambda^*}(d^{n +2}) = \Phi^{\alpha^*, \gamma^*, \lambda^*}(d^{n + 1}) > \Phi^{\alpha^*, \gamma^* ,\lambda^*}(e^{n+1})$.
\end{lemma}

\begin{proof}
By definition of the modified flow, we know that  $\Phi^{\alpha^*, \gamma^*, \lambda^*}(d^{n +2}) = \Phi^{\alpha^*, \gamma^*, \lambda^*}(d^{n + 1})$, so we only need to verify two inequalities. For the first one, by expanding $\Phi^{\alpha^*, \gamma^*, \lambda^*}(d^{n + 1})$ we have 
\begin{align*}
\Phi^{\alpha^*, \gamma^*, \lambda^*}(d^{n + 1}) &= v_1^{n+1} + \frac{\varepsilon - f(d^{n+2}) - \delta}{f(d^{n+1})} \\
& = n^2 + n + 1 - \Theta(1/n^4) \\
& < n^2 + n + 1.9 \\
& = \Phi^{\alpha^*, \gamma^* ,\lambda^*}(e^{n+2}).
\end{align*}
For the other inequality, note that $\Phi^{\alpha^*, \gamma^* ,\lambda^*}(e^{n+1}) = n^2 + n + 1 - \frac{0.9 \cdot n^4}{b \cdot f(e^{n+1})} = n^2 + n + 1 - \Theta(1/n^3) < \Phi^{\alpha^*, \gamma^*, \lambda^*}(d^{n + 1}) $.
\end{proof}

The following lemma guarantees that monotonicity still hold after boosting, which in particular (together with Lemma~\ref{budgetafter}) implies the first two bullets of Proposition~\ref{prop:budgetvalid} .

\begin{lemma} \label{budgetmonotone1}
For all $x, y \in \{0, 1\}^n$, we have \[
\max \left (\Phi^{\alpha^*,\gamma^*,\lambda^*}(d^{i}),  \Phi^{\alpha^*,\gamma^*,\lambda^*}(e^{i})  \right) <  \min \left (\Phi^{\alpha^*,\gamma^*,\lambda^*}(d^{i + 1}),  \Phi^{\alpha^*,\gamma^*,\lambda^*}(e^{i + 1})  \right)\] for $i \in [1, n]$.
\end{lemma}

\begin{proof}
By Lemma~\ref{budgetmonotone}, we know the gap under canonical flow is at least $0.5$, and we are only adding negligible amount of flow ($O(\frac{1}{n^6})$), so the overall perturbation on virtual values will be no more than $\frac{\varepsilon}{f(\cdot)}$, which is $O(\frac{1}{n^5})$ by Lemma~\ref{budgetdrange} and Lemma~\ref{budgeterange}. Thus the monotonicity still holds.
\end{proof}

Note that the bullet four of Proposition~\ref{prop:budgetvalid} is directly from the bullet four of Proposition~\ref{prop:budgetMyerson}, which guarantees that if $\disj = \mathsf{no}$, there must be some $k$ such that $\Phi^{\alpha^M, \gamma^M, \lambda^M}_2(v_2^{k}) > \Phi^{\alpha^M,\gamma^M,\lambda^M}_1(v_1^k )$. Then by the definition of modified flow, there must be some $k^*$ satisfies the property.

\begin{proof}[Proof of Lemma~\ref{lem:budgetBIC}]
We first confirm that the second-price auction except for $n+2$ with careful tie-breaking at $k^*$ is BIC. Observe first that, because the auction has a monotone allocation rule (and satisfies the payment identity) that no bidder can ever benefit by misreporting their value, but honestly reporting their interest. This follows immediately from the definition of the payment identity, and is well-known~\cite{Myerson81}. In particular, this implies that the auction is BIC for both bidders.

Now let's consider the budget constraint. By definition of the second-price auction except for $n+2$ with careful tie-breaking at $k^*$, we have 
\begin{align*}
 P_1(v_1^{n+2}) & = \sum_{k = 1}^{n+2} v_1^k \cdot (\pi_1(v_1^k) - \pi_1((v_1^{k-1})))\\
 & = \sum_{k = 1}^{n+1} v_1^k \cdot (\pi_1(v_1^k) - \pi_1((v_1^{k-1})))\\
 & =   (1-\Delta) \cdot f_2(v_2^k) \cdot v_1^{k^*} + \Delta \cdot f_2(v_2^k) \cdot v_1^{k^* + 1}  + \sum_{k \neq k^*} v_1^k \cdot f_2(v_2^k) \\
 & = \Delta \cdot (v_1^{k^* + 1} -  v_1^{k^*}) + \sum_{k = 1} ^{n+1} v_1^k \cdot f_2(v_2^k)\\
 & = \Delta + \sum_{k = 1} ^{n+1} v_1^k \cdot f_2(v_2^k) \\
 & =  \left ( 1 - \frac{n^4}{b} \right) \cdot B \\
 & = \pi_1(v_1^{n+2}) \cdot B,
\end{align*}
which concludes the proof.
\end{proof}

\end{document}